\documentclass[a4paper,oneside,12pt]{article}

\usepackage{jheppub}
\usepackage{amsfonts}
\usepackage{amsmath}
\usepackage{amssymb}
\usepackage{amsthm}
\usepackage[mathscr]{eucal}
\usepackage{bbold}
\usepackage{braket}
\usepackage{dsfont}
\usepackage{framed}
\usepackage{mathtools}
\usepackage{physics}
\usepackage[normalem]{ulem}
\usepackage{tensor}
\usepackage{thmtools}
\usepackage{thm-restate}
\usepackage{ulem}
\usepackage{tikz}
\usetikzlibrary{arrows.meta,arrows} 
\usepackage{centernot}
\usepackage{enumitem} 
\usepackage[a]{esvect}  
\usepackage{slashed}
\usepackage[only,llbracket,rrbracket]{stmaryrd} 
\usepackage[table]{xcolor} 
\usepackage{adjustbox}

\usepackage{yfonts}
\usepackage{float}
\usetikzlibrary{math} 
\usetikzlibrary{calc}		

\DeclareFontFamily{U}{jkpmia}{}
\DeclareFontShape{U}{jkpmia}{m}{it}{<->s*jkpmia}{}
\DeclareFontShape{U}{jkpmia}{bx}{it}{<->s*jkpbmia}{}
\DeclareMathAlphabet{\mathfrak}{U}{jkpmia}{m}{it}
\SetMathAlphabet{\mathfrak}{bold}{U}{jkpmia}{bx}{it}

\usepackage{subcaption}
\usepackage{float}
\usepackage{afterpage}
\renewcommand{\thefigure}{\arabic{figure}}

\usetikzlibrary{arrows}
\tikzset{elabelcolor/.style={color=blue} 
    vertex/.style={circle,draw,minimum size=1.5em},
    edge/.style={->,> = latex'}
}  

\usepackage{cleveref}

\definecolor{VHcolor}{rgb}{0.7,0.3,0.9}

\definecolor{MRocolor}{rgb}{0.1,0,1}

\newtheorem{cor}{Corollary}
\newtheorem{defi}{Definition}
\newtheorem{lemma}{Lemma}

\newtheorem{thm}{Theorem}

\newcommand{\N}{{\sf N}}
\newcommand{\D}{\sf{D}}

\newcommand{\uI}{X}
\newcommand{\uJ}{Y}
\newcommand{\uK}{Z}

\newcommand{\mmC}{\text{\textit{\u{C}}}} 

\newcommand{\nsp}{\llbracket\N\rrbracket} 

\newcommand{\mgs}{\mathcal{E}} 
\newcommand{\p}{\mathcal{U}} 
\newcommand{\x}{\mathcal{X}} 
\newcommand{\ds}{\mathcal{D}} 
\newcommand{\cds}{\comp{\mathcal{D}}} 
\newcommand{\ac}{\Check{\mathcal{A}}}  

\newcommand{\bset}{\beta\text{-set}} 
\newcommand{\bsets}{\beta\text{-sets}}
\newcommand{\bs}[1]{\beta(#1)}
\newcommand{\ess}{\mathfrak{B}^{^{\!\vee}}}
\newcommand{\ness}{\mathfrak{B}^{^{\!\bigtriangledown}}}
\newcommand{\pnote}{\mathfrak{B}^{^{\!+\setminus\vee}}}
\newcommand{\cess}{\mathfrak{B}^{^{\!\vee\!\!\!\vee}}}
\newcommand{\ncess}{\mathfrak{B}^{^{\!\vee\setminus\!\vee\!\!\!\vee}}}
\newcommand{\pos}{\mathfrak{B}^{^{\!+}}}
\newcommand{\parv}{\mathfrak{B}^{^{\!\oplus}}}
\newcommand{\van}{\mathfrak{B}^{^{\!0}}}
\newcommand{\xx}{\mathfrak{X}}
\newcommand{\posj}{\pos_{\!\text{{\tiny $\subset$}}\uJ}} 

\newcommand{\ent}{{\sf S}}  
\newcommand{\mi}{{\sf I}}  
\newcommand{\pmi}{\mathcal{P}}   
\newcommand{\cpmi}{\comp{\mathcal{P}}} 
\newcommand{\face}{\mathscr{F}}   

\newcommand{\gf}[1]{{\mathrm{#1}}}
\newcommand{\hp}{\gf{H}_{\mathcal{P}}}  
\newcommand{\lhp}{\gf{L}_{\mathcal{P}}}  
\newcommand{\subh}[1]{\left.\gf{H}_{\pmi}\right|_{#1}}

\newcommand{\lat}[1]{\mathscr{L}_{\text{#1}}}

\definecolor{PminusEcol}{rgb}{1,0.8,0.8}
\definecolor{EminusMcol}{rgb}{1,0.85,0.65}
\definecolor{Mcol}{rgb}{1,0.97,0.5}
\definecolor{Qcol}{rgb}{0.7,1,0.9}
\definecolor{Vcol}{rgb}{0.83,0.83,1}
\definecolor{EminusPcol}{rgb}{0.8,1,0.6}

\definecolor{Ccol}{rgb}{0.9,0.9,0.9}

\definecolor{bulkcol}{rgb}{0.8,0.8,0.8}

\newcommand{\comp}[1]{#1^\complement}

\title{Correlation hypergraph: a new representation of a quantum marginal independence pattern}
 
\author[a]{Veronika E. Hubeny,}
\emailAdd{veronika@physics.ucdavis.edu}

\author[b]{Massimiliano Rota}
\emailAdd{max.rota@bristol.ac.uk}

\affiliation[a]{Center for Quantum Mathematics and Physics (QMAP)\\ 
Department of Physics \& Astronomy, University of California, Davis, CA 95616 USA}
\affiliation[b]{School of Mathematics, University of Bristol,
Woodland Road, Bristol, BS8 1UG UK}

\abstract{We continue the study of the quantum marginal independence problem, namely the question of which faces of the subadditivity cone are achievable by quantum states. We introduce a new representation of the patterns of marginal independence  (PMIs, corresponding to faces of the subadditivity cone) based on certain correlation hypergraphs, and demonstrate that this representation provides a more efficient description of a PMI, and consequently of the set of PMIs which are compatible with strong subadditivity. We then show that these correlation hypergraphs generalize to arbitrary quantum systems the well known relation between positivity of mutual information and connectivity of entanglement wedges in holography, and further use this representation to derive new results about the combinatorial structure of collections of simultaneously decorrelated subsystems specifying SSA-compatible PMIs.
In the context of holography, we apply these techniques to derive a necessary condition for the realizability of entropy vectors by simple tree graph models, which were conjectured in \cite{Hernandez-Cuenca:2022pst} to provide the building blocks of the holographic entropy cone. Since this necessary condition is formulated in terms of chordality of a certain graph, it can be tested efficiently.}

\begin{document}
 

\maketitle

\section{Introduction}

A central question in gauge/gravity  duality~\cite{Maldacena:1997re,Witten:1998qj,Gubser:1998bc} is how bulk spacetime is encoded in the boundary theory. For semiclassical states, substantial progress in this direction\footnote{\,See for example~\cite{Harlow:2018fse} and references therein.} has relied, at least in part, on the HRRT prescription for computing the entanglement entropy of spatial regions in the boundary~\cite{Ryu:2006bv,Hubeny:2007xt}, along with its various generalizations~\cite{Faulkner:2013ana,Engelhardt:2014gca}. It is well known that in the purely classical regime—i.e., neglecting quantum effects in the bulk—this prescription implies that the entropies of boundary subregions satisfy certain linear inequalities that are generically violated by quantum states~\cite{Hayden:2011ag}. These inequalities thus reflect features of the ``entanglement structure'' characteristic of boundary states that admit a classical bulk dual. One may then hope that a complete understanding of such inequalities, for arbitrarily intricate and fine-grained partitions of the boundary, could yield deep insight into the nature of this encoding.

While for dynamical spacetimes it is still an open question \cite{Rota:2017ubr,Bao:2018wwd,Czech:2019lps,Grado-White:2024gtx,Grado-White:2025jci}
whether new inequalities of this type exist besides monogamy of mutual information \cite{Hayden:2011ag,Wall:2012uf}, in the static case many new inequalities have been found for five or more parties \cite{Bao:2015bfa,Cuenca:2019uzx,Czech:2022fzb,Hernandez-Cuenca:2023iqh,Czech:2023xed,Czech:2024rco}. However, even in this restricted setting, a complete derivation of all inequalities, for an arbitrary number of parties, is still an open and challenging problem. 

A convenient framework to study these inequalities for a fixed, but arbitrary, number of parties $\N$ was first introduced in \cite{Bao:2015bfa} (mimicking analogous constructions for arbitrary quantum states \cite{1193790} and classical probability distributions \cite{641556}), and it is known as the \textit{holographic entropy cone} (HEC). As already shown in \cite{Bao:2015bfa}, the HEC is a convex polyhedral cone for any $\N$, and can be thought of as the set of entropy vectors in $\mathbb{R}^{\D}$, with $\D=2^\N-1$, realized by certain graph models. 

It is a general mathematical fact that any polyhedral cone can be equivalently described either by a finite set of linear inequalities or by its set of extreme rays. Given one of these two representations, it is in principle always possible to derive the other; however, when the number of dimensions, inequalities, or extreme rays is large, this conversion problem can become computationally intractable, as no efficient algorithm is currently known. Setting aside questions of computability, one may then ask which of the two descriptions is more naturally suited to systematic derivation and physical interpretation. Following~\cite{Hernandez-Cuenca:2022pst} (building on the ideas of~\cite{Hubeny:2018trv,Hubeny:2018ijt,Hernandez-Cuenca:2019jpv}), one goal of this work is to advance the characterization of the extreme rays of the HEC.

It was conjectured in~\cite{Hernandez-Cuenca:2022pst} that the extreme rays of the holographic entropy cone (HEC), for any fixed number of parties $\N$, can be obtained from a particular subset of the extreme rays of the $\N'$-party ``subadditivity cone'' (SAC), for some finite $\N' \geq \N$ that depends on $\N$. As the name suggests, the SAC is the polyhedral cone defined by all instances of the subadditivity inequality, and therefore constitutes an immediate outer bound to the HEC. The subset of interest comprises those extreme rays that admit a realization via ``simple tree'' graph models—namely, graph models with tree topology in which each ``boundary vertex'' (equivalently, each connected boundary region in an associated holographic setup) is labeled by a distinct party. The mapping from these rays to candidate extreme rays of the HEC corresponds to a  ``subsystem coarse-graining'': a combinatorial identification where collections of the $\N'$ parties are reinterpreted as single parties in the $\N$-party system.

Assuming this conjecture holds, the problem of characterizing the extreme rays of the HEC reduces to that of identifying the extreme rays of the SAC that can be realized by simple tree graph models. One of the goals of this work is to derive a necessary condition for such realizability which applies to an arbitrary number of parties. Importantly, while we establish only necessity at present, we expect the condition to also be sufficient and to yield a constructive procedure for building a simple tree graph model realizing a given entropy vector when it exists (or establishing a no-go statement when it does not)~\cite{2025wip}. Moreover, as we will explain, the condition we derive is not restricted to the extreme rays of the SAC, but extends to faces of arbitrary dimension, and can be tested efficiently.

The second main goal of this work is motivated by quantum information-theoretic considerations. The \textit{quantum marginal independence problem} (QMIP), introduced in \cite{Hernandez-Cuenca:2019jpv}, intuitively asks: what are the implications of the presence or absence of correlations among certain subsystems of a larger quantum system for the presence or absence of correlations among other subsystems? More formally, it can be phrased as the question of which face of the $\N$-party SAC is achievable by the entropy vector of a density matrix. This framing arises because the saturation of an instance of subadditivity is equivalent to the vanishing of the corresponding mutual information, which is zero if and only if there is no correlation between the subsystems involved. From this perspective, the QMIP can also be understood as a question about the possible structures of $\N$-party density matrices, viewed through the factorizations of their marginals.

While the QMIP is formulated for arbitrary quantum states, it is also interesting to consider special instances of the problem where the set of allowed states is restricted. The classical version of this problem, for example—where states are restricted to joint probability distributions—has been studied in the context of probabilistic reasoning in AI \cite{pearl1988probabilistic} and was completely solved in \cite{GEIGER1991128}. Notably, this problem was solved despite the fact that not all entropy inequalities for probability distributions are known; indeed, already at four parties, there exists an infinite number of such inequalities \cite{2007:Matus}. This situation stands in remarkable contrast to the holographic case, where the aforementioned conjecture would imply that even a specific partial solution to this problem (namely, restricting to one-dimensional faces for arbitrary $\N$) would not only yield the general solution, but also allow for the reconstruction of all inequalities. 

Although this remains a conjecture, similar behavior has already been observed more generally in the class of states whose entropies are described by the hypergraph models of \cite{Bao:2020zgx}. In particular, it was shown in \cite{He:2023aif} that all extreme rays of the ``hypergraph entropy cone'' for $\N = 4$ can be obtained via subsystem coarse-grainings of extreme rays of the SAC for $\N = 9$. Since this hypergraph cone is contained within the entropy cone of stabilizer states \cite{Linden:2013kal,Walter:2020zvt}, and thus within the full quantum entropy cone \cite{1193790}, it is conceivable that solving the QMIP could shed light on the structure of these cones, and hence on the corresponding entropy inequalities.

To make progress in this direction, it is important to obtain an efficient and intuitive representation of the combinatorial structure of the set of vanishing (or, alternatively, strictly positive) instances of the mutual information for a face of the SAC. In holography, it is well known that the positivity of the leading term of the mutual information between regions diagnoses the connectivity of their joint entanglement wedge \cite{Headrick:2013zda,Hernandez-Cuenca:2023iqh}. As we will explain, this relation between the positivity of the mutual information and connectivity admits a natural generalization to any face of the subadditivity cone consistent with strong subadditivity (SSA).\footnote{\,More precisely, as we will explain, not even SSA is strictly necessary, for it can be replaced by a weaker combinatorial constraint introduced in \cite{He:2022bmi} and dubbed ``Klein's condition''.} Specifically, this notion of connectivity pertains to sub-hypergraphs of a certain hypergraph representation of a face, which we refer to as the \emph{correlation hypergraph},\footnote{\,We caution the reader that this correlation hypergraph representation is distinct from the hypergraph models of \cite{Bao:2020zgx}.} or more precisely, of the set of instances of the mutual information that are strictly positive for any entropy vector in the interior of the face. Moreover, we will show that this representation is more efficient than those based on MI-poset down-sets used in previous works \cite{He:2022bmi,He:2023aif,He:2023cco,He:2024xzq}, and we will use it to prove new results regarding the combinatorial structure of the set of mutual information instances that characterize an ``SSA-compatible'' face. The discovery of this hypergraph representation, and its close analogy with properties of holographic graph models, is a remarkable example of how holography—and in particular, the ``geometrization'' of quantum information measures—can yield deep insights into structural properties of arbitrary quantum states.

\paragraph{Structure of the paper:} We begin in \S\ref{subsec:KC-review} by reviewing all fundamental definitions about entropy space and the SAC, as well as the necessary concepts from \cite{He:2022bmi}, in particular the poset of instances of the mutual information and Klein’s condition. In \S\ref{subsec:graph_review}, we then review the fundamental definitions concerning holographic graph models and highlight how they encode the relationship between subsystem correlations and entanglement wedge connectivity. This subsection also includes a few known results that, to our knowledge, have not appeared before in the literature. The key new tool introduced in this work is the notion of a $\bset$, presented in \S\ref{subsec:beta-sets}. The characterization of KC-compatible faces of the SAC in terms of $\bsets$ is presented in \S\ref{subsec:KC-PMIs-BS-language}, and it is then used in \S\ref{subsec:hypergraph-rep} to define the correlation hypergraph representation. The necessary condition for the realizability of an entropy vector by a simple tree graph model that we mentioned above is derived in \S\ref{sec:necessary} using this framework. Finally, \S\ref{sec:discussion} outlines a list of directions for future investigation.

\paragraph{Guide for the reader:} For the reader who is mostly interested in the derivation of the necessary condition for the realizability of entropy vectors by simple tree graph models, very little is needed from the core section of this work, \S\ref{sec:bsets-corr-hyp}. We recommend starting from the review section, \S\ref{sec:review}, and then proceeding to the definition of positive $\bsets$ at the beginning of \S\ref{subsec:beta-sets} and the definition of the correlation hypergraph at the beginning of \S\ref{subsec:hypergraph-rep}. This is sufficient to then read \S\ref{sec:necessary}, where the necessary condition is derived.

The reader who is instead mostly interested in the quantum information aspects of our construction can start from the review in \S\ref{subsec:KC-review} and proceed directly to \S\ref{sec:bsets-corr-hyp}. A reading of \S\ref{subsec:graph_review} is recommended only to fully appreciate the similarities between properties of the correlation hypergraph and holographic graph models.

\paragraph{Conventions:} To simplify notation, when we subtract an element $a$ from a set $A$ we write $A\setminus a$ instead of $A\setminus \{a\}$. The complement of a subset $B$ of a given set (which will be clear from context and hence left implicit) is denoted by $\comp{B}$. We use the compact notation $[n]$ for the set $\{1,2,\ldots,n\}$,
and $\llbracket n \rrbracket$ for the enlarged set  $\{0,1,\ldots,n\}$ (explained below).

\section{Review}
\label{sec:review}

\subsection{Entropy space, the MI-poset, and Klein's condition}
\label{subsec:KC-review}

Given an $\N$-party quantum system, and a density matrix $\rho$, we will often consider a purification of $\rho$ by introducing an ancillary subsystem conventionally denoted by 0. 
To denote the collection of all N parties along with the purifier, we use the shorthand
$\nsp\coloneqq \{0,1,\ldots,\N\}$, and an arbitrary party in this set is denoted by $\ell$. We will typically consider subsets of parties in the full system, which could include the purifier, and denote them by $\uI,\uJ,\uK\subseteq\nsp$. Occasionally, we will also need to consider subsets restricted to \textit{not} include the purifier; these we will denote by $I,J,K\subseteq [\N]$. Conventionally, the word \textit{subsystem} will refer to a non-empty subset of parties, unless stated otherwise. 

For an arbitrary density matrix $\rho$, its \textit{entropy vector} $\vec{\ent}(\rho)$ is defined as the vector with components\footnote{\,To define this vector one also needs to make a choice for the order of these components. A conventional choice is lexicographic order, however the specific choice will not play any role for us.}
\begin{equation}
\label{eq:entropy-vector}
    \ent_J=\ent(\rho_J)\qquad \forall\,J\subseteq [\N]\quad \text{with}\quad J\neq\varnothing,
\end{equation}
where $\rho_J$ is the reduced density matrix corresponding to a collection of parties $J$, and $\ent$ is the von Neumann entropy. 
The space $\mathbb{R}_+^{\D}$, with $\D=2^{\N}-1$, where these vectors live is called \textit{entropy space}.
Any vector of this space will be called an ``entropy vector'', even when it does not necessarily correspond to the vector of entropies of a density matrix as in Eq.~\eqref{eq:entropy-vector}. When for a given vector $\vec{\ent}$ there exists a density matrix $\rho$ such that $\vec{\ent}=\vec{\ent}(\rho)$, we will say that $\vec{\ent}$ is \textit{realizable} by a quantum state. Occasionally we will be interested in restricted classes of states (for example stabilizer states, or geometric states in holography), and we will say that $\vec{\ent}$ is realizable in that class if there is a density matrix $\rho$ within that class such that $\vec{\ent}=\vec{\ent}(\rho)$. The topological closure of the set of entropy vectors that can be realized by arbitrary quantum states (for arbitrary Hilbert spaces) is a convex cone known as the $\N$-party \textit{quantum entropy cone} QEC$_{\N}$ \cite{1193790}.

Since the von Neumann entropy satisfies subadditivity, an obvious outer bound of QEC$_{\N}$ is the $\N$-party \textit{subadditivity cone} SAC$_{\N}$ \cite{Hernandez-Cuenca:2019jpv}, which is defined as the convex polyhedral cone specified by all instances of the subadditivity and Araki-Lieb inequalities \cite{Wilde_2013}.\footnote{\,Of course one can get a more stringent bound by also considering SSA, however this inequality will not play any particular role here. One reason will become clear momentarily, and we refer the reader to \cite{Hernandez-Cuenca:2019jpv,He:2022bmi} for more details. In the holographic context, another simple reason is that SSA is automatically implied by \textit{monogamy of mutual information} (MMI) \cite{Hayden:2011ag,Bao:2015bfa}.} It is well known that these two classes of inequalities are equivalent, since one can be derived from the other by introducing a purification of a given density matrix. In what follows we will not distinguish between them, and refer to all the inequalities that specify the SAC$_{\N}$ simply as instances of subadditivity. 

Each inequality will be written as $\mi(\uJ:\uK)\geq 0$, where $\mi(\uJ:\uK)$ is the corresponding instance of the \textit{mutual information} (MI). The set of inequalities that specify the SAC$_{\N}$ can then be written as 
\begin{equation}
    \mgs_{\N}=\left\{\mi(\uJ:\uK),\ \ \ \forall\,\uJ,\uK\;\; \text{with}\;\; \uJ\cap\uK=\varnothing,\; \ \uJ,\uK\neq\varnothing\right\}.
\end{equation}
Notice that written in this form, $\mi(\uJ:\uK)\geq 0$ is not an inequality in entropy space (which is defined by the non-empty subsystems $J\subseteq [\N]$), but in a larger space where we also have additional variables for subsystems that include the purifier. However, we will always implicitly assume that in any such inequality, whenever $\uJ$ or $\uK$ includes the purifier, we replace its entropy with the entropy of the complementary subsystem, i.e.\ $\ent_{\uJ} \to \ent_{\comp{\uJ}}$ where $\comp{\uJ}\coloneqq \nsp\setminus\uJ$ (and similarly for $\uK$). In the particular case of $\uK=\comp{\uJ}$, subadditivity reduces to non-negativity of the corresponding entropy $\ent_{\uJ}\geq 0$, since $\ent_{\nsp}=\ent_{\varnothing}=0$.

Any face $\face$ of the SAC$_{\N}$ can be described by the subset $\pmi\subseteq\mgs_{\N}$ of instances of the mutual information that vanish for any entropy vector $\vec{\ent}\in\text{int}(\face)$. We call this set $\pmi$ a \textit{pattern of marginal independence} (PMI) and denote by $\face_{\pmi}$ (or $\face(\pmi)$) the corresponding face, and by $\mathbb{S}_{\pmi}$ (or $\mathbb{S}(\pmi)$) the linear subspace spanned by the vectors in $\face_{\pmi}$. This linear subspace can then be viewed as the space of solutions to the following system of equations
\begin{align}
    & \mi(\uJ:\uK)=0  \qquad\quad\;\;\, \forall\, \mi(\uJ:\uK)\in\pmi, \label{eq:pmi-equations}\\
        & \ent_{\uJ}=\ent_{\comp{\uJ}} \qquad\qquad\quad\; \forall\,\uJ\subseteq \nsp. \label{eq:purity}
\end{align}
Given an arbitrary entropy vector $\vec{\ent}$ in the SAC$_{\N}$, its PMI $\pmi(\vec{\ent})$ is the set of all the instances of mutual information that vanish for $\vec{\ent}$. Geometrically $\pmi(\vec{\ent})$ is associated to the lowest dimensional face $\face_{\pmi}$ of the SAC$_{\N}$ that contains $\vec{\ent}$.

A PMI $\pmi$ will be said to be \textit{realizable} if there is a realizable entropy vector $\vec{\ent}$ such that $\vec{\ent}\in\text{int}(\face_{\pmi})$. The question of which PMIs are realizable is the \textit{quantum marginal independence problem} first introduced in \cite{Hernandez-Cuenca:2019jpv}. Similarly, $\pmi$ is said to be SSA-\textit{compatible} if there is an entropy vector $\vec{\ent}\in\text{int}(\face_{\pmi})$ that satisfies all instances of SSA.\footnote{Although this condition does not guarantee that SSA holds for \emph{all} entropy vectors $\vec{\ent}\in\text{int}(\face_{\pmi})$ (a simple example being the full-dimensional face $\text{int}(\face_{\varnothing})$), it is both more analogous to the notion of realizability and more convenient to work with.}
Clearly, since SSA is satisfied by all quantum states, a PMI is realizable only if it is SSA-compatible. 

For this reason we will mainly be interested in SSA-compatible PMIs, but it will be convenient to replace SSA with a slightly weaker condition and consider a correspondingly larger set of PMIs. This condition is not an inequality in entropy space, but rather a combinatorial constraint on the form of the set $\pmi$, and it has the advantage that it can be phrased entirely in terms of SA. It follows from the fact that in quantum mechanics the mutual information between two subsystems vanishes if and only if their joint density matrix factorizes, and in particular from the chain of implications
\begin{equation}
\label{eq:kc}
    \mi(\uI\uJ:\uK)=0\quad \implies\quad \rho_{\uI\uJ\uK}=\rho_{\uI\uJ}\otimes\rho_{\uK} \implies\quad \mi(\uI:\uK)=\mi(\uJ:\uK)=0.
\end{equation}
We can then focus only on PMIs $\pmi$ such that for any instance of the mutual information contained in $\pmi$, all other instances given by Eq.~\eqref{eq:kc} also belong to $\pmi$. We call this condition on a PMI the \textit{Klein's condition} (KC), and a PMI that satisfies this condition a KC-PMI. It is immediate to verify that any SSA-compatible PMI satisfies this condition.

\begin{figure}[tb]    
    \centering
        \begin{tikzpicture}
    \node (0a123) at (2.2,3) {{\scriptsize $\mi(0\!:\!123)$}};
     \node (1a023) at (3.7,3) {{\scriptsize $\mi(1\!:\!023)$}};
     \node (2a013) at (5.2,3) {{\scriptsize $\mi(2\!:\!013)$}};
     \node (3a012) at (6.7,3) {{\scriptsize $\mi(3\!:\!012)$}};
     \node (01a23) at (8.2,3) {{\scriptsize $\mi(01\!:\!23)$}};
     \node (02a13) at (9.7,3) {{\scriptsize $\mi(02\!:\!13)$}};
     \node (03a12) at (11.2,3) {{\scriptsize $\mi(03\!:\!12)$}};
     
    \node (01a2) at (0.1,1) {{\scriptsize $\mi(01\!:\!2)$}};
    \node (02a1) at (1.2,1) {{\scriptsize $\mi(02\!:\!1)$}};
    \node (12a0) at (2.3,1) {{\scriptsize $\mi(12\!:\!0)$}};
    \node (01a3) at (3.7,1) {{\scriptsize $\mi(01\!:\!3)$}};
    \node (03a1) at (4.8,1) {{\scriptsize $\mi(03\!:\!1)$}};
    \node (13a0) at (5.9,1) {{\scriptsize $\mi(13\!:\!0)$}};
    \node (02a3) at (7.3,1) {{\scriptsize $\mi(02\!:\!3)$}};
    \node (03a2) at (8.4,1) {{\scriptsize $\mi(03\!:\!2)$}};
    \node (23a0) at (9.5,1) {{\scriptsize $\mi(23\!:\!0)$}};
    \node (12a3) at (10.9,1) {{\scriptsize $\mi(12\!:\!3)$}};
    \node (13a2) at (12,1) {{\scriptsize $\mi(13\!:\!2)$}};
    \node (23a1) at (13.1,1) {{\scriptsize $\mi(23\!:\!1)$}};

    \node (0a1) at (1.8,-1) {{\scriptsize $\mi(0\!:\!1)$}};
     \node (0a2) at (3.8,-1) {{\scriptsize $\mi(0\!:\!2)$}};
     \node (1a2) at (5.8,-1) {{\scriptsize $\mi(1\!:\!2)$}};
     \node (0a3) at (7.8,-1) {{\scriptsize $\mi(0\!:\!3)$}};
     \node (1a3) at (9.8,-1) {{\scriptsize $\mi(1\!:\!3)$}};
     \node (2a3) at (11.8,-1) {{\scriptsize $\mi(2\!:\!3)$}};

    \draw[-,gray,very thin] (0a123.south) -- (12a0.north);
    \draw[-,gray,very thin] (0a123.south) -- (13a0.north);
    \draw[-,gray,very thin] (0a123.south) -- (23a0.north);
    \draw[-,gray,very thin] (1a023.south) -- (02a1.north);
    \draw[-,gray,very thin] (1a023.south) -- (03a1.north);
    \draw[-,gray,very thin] (1a023.south) -- (23a1.north);
    \draw[-,gray,very thin] (2a013.south) -- (01a2.north);
    \draw[-,gray,very thin] (2a013.south) -- (03a2.north);
    \draw[-,gray,very thin] (2a013.south) -- (13a2.north);
    \draw[-,gray,very thin] (3a012.south) -- (01a3.north);
    \draw[-,gray,very thin] (3a012.south) -- (02a3.north);
    \draw[-,gray,very thin] (3a012.south) -- (12a3.north);
    \draw[-,gray,very thin] (01a23.south) -- (01a2.north);
    \draw[-,gray,very thin] (01a23.south) -- (01a3.north);
    \draw[-,gray,very thin] (01a23.south) -- (23a0.north);
    \draw[-,gray,very thin] (01a23.south) -- (23a1.north);
    \draw[-,gray,very thin] (02a13.south) -- (02a1.north);
    \draw[-,gray,very thin] (02a13.south) -- (13a0.north);
    \draw[-,gray,very thin] (02a13.south) -- (02a3.north);
    \draw[-,gray,very thin] (02a13.south) -- (13a2.north);
    \draw[-,gray,very thin] (03a12.south) -- (12a0.north);
    \draw[-,gray,very thin] (03a12.south) -- (03a1.north);
    \draw[-,gray,very thin] (03a12.south) -- (03a2.north);
    \draw[-,gray,very thin] (03a12.south) -- (12a3.north);

    \draw[-,gray,very thin] (01a2.south) -- (0a2.north);
    \draw[-,gray,very thin] (01a2.south) -- (1a2.north);
    \draw[-,gray,very thin] (02a1.south) -- (0a1.north);
    \draw[-,gray,very thin] (02a1.south) -- (1a2.north);
    \draw[-,gray,very thin] (12a0.south) -- (0a1.north);
    \draw[-,gray,very thin] (12a0.south) -- (0a2.north);
    \draw[-,gray,very thin] (01a3.south) -- (0a3.north);
    \draw[-,gray,very thin] (01a3.south) -- (1a3.north);
    \draw[-,gray,very thin] (03a1.south) -- (0a1.north);
    \draw[-,gray,very thin] (03a1.south) -- (1a3.north);
    \draw[-,gray,very thin] (13a0.south) -- (0a1.north);
    \draw[-,gray,very thin] (13a0.south) -- (0a3.north);
    \draw[-,gray,very thin] (02a3.south) -- (0a3.north);
    \draw[-,gray,very thin] (02a3.south) -- (2a3.north);
    \draw[-,gray,very thin] (03a2.south) -- (0a2.north);
    \draw[-,gray,very thin] (03a2.south) -- (2a3.north);
    \draw[-,gray,very thin] (23a0.south) -- (0a2.north);
    \draw[-,gray,very thin] (23a0.south) -- (0a3.north);
    \draw[-,gray,very thin] (12a3.south) -- (1a3.north);
    \draw[-,gray,very thin] (12a3.south) -- (2a3.north);
    \draw[-,gray,very thin] (13a2.south) -- (1a2.north);
    \draw[-,gray,very thin] (13a2.south) -- (2a3.north);
    \draw[-,gray,very thin] (23a1.south) -- (1a2.north);
    \draw[-,gray,very thin] (23a1.south) -- (1a3.north);
     
    \end{tikzpicture}
    \caption{The Hasse diagram of the $\N=3$ MI-poset.
    }
    \label{fig:N3-MI-poset}
\end{figure}
 
Klein's condition, and the structure of KC-PMIs, can be described more conveniently by introducing a partial order in $\mgs_{\N}$ \cite{He:2022bmi}
\begin{equation}
\label{eq:mi_order}
    \mi(\uJ:\uK)\preceq\mi(\uJ':\uK') \quad \iff \quad \uJ\subseteq\uJ'\;\text{and}\;\uK\subseteq\uK'\;\;\ \text{or}\;\;\uJ\subseteq\uK'\;\text{and}\;\uK\subseteq\uJ'.
\end{equation}
We call the resulting poset $(\mgs_{\N},\preceq)$ the $\N$-party \textit{mutual information poset}, or simply the MI-poset (as an example, the $\N=3$ MI-poset is shown in Fig.~\ref{fig:N3-MI-poset}).\footnote{\,To simplify the notation we will denote by $\mgs_{\N}$ (or simply $\mgs$) both the set of all MI instances for $\N$ parties, and the MI-poset.} It is then easy to see that a PMI $\pmi$ is a KC-PMI if and only if $\pmi$ is a \textit{down-set} in the MI-poset.\footnote{\,A down-set in a poset $\p$ is a subset $\x\subset\p$ such that for every $x\in\x$ and $y\prec x$, $y\in\x$. An \textit{up-set} is defined analogously.} However, it is important to keep in mind that an arbitrary down-set $\ds$ in the MI-poset is not necessarily a PMI. The reason is that $\ds$ is not necessarily ``closed'' under linear dependence, and even if it is, it does not necessarily respect all instances of subadditivity.\footnote{\,More precisely, as discussed in \cite{He:2022bmi}, given $\ds$, there might be some MI instance $\mi\notin\ds$ which is a linear combination of the instances in $\ds$. Furthermore, when this is not the case, one can associate to $\ds$ a linear subspace of entropy space, but this is not necessarily spanned by a face of the SAC.} We will present an explicit example in \S\ref{subsec:beta-sets}.

\subsection{Holographic graph models}
\label{subsec:graph_review}

As mentioned in the introduction, the HEC can be conveniently formulated entirely in terms of the so-called holographic graph models, which we will use in this work; we refer the reader to \cite{Bao:2015bfa} for a detailed explanation of why this is the case, and the mapping between these models and the holographic set-ups.

Following \cite{Bao:2015bfa}, let $\gf{G}=(V,E,W)$ be a weighted graph\footnote{\,We assume the graph is free of loops and multiple edges between any pair of vertices.} with vertex set $V$, edge set $E$, and edge weights $W$, which we assume to be non-negative. We choose an arbitrary subset $\partial V\subset V$ of vertices that we call \textit{boundary vertices},\footnote{\,This terminology is motivated by the holographic set-up, where for a given geometric state of the CFT, the graph is constructed from the partition of a time slice by RT surfaces \cite{Bao:2015bfa} and each vertex in $\partial V$ corresponds to a boundary region.} and we will occasionally refer to the remainig vertices (in $V\setminus\partial V$) as ``bulk vertices''. For an arbitrary number of parties $\N$, with $\N+1\leq |\partial V|$, we then assign a party to each boundary vertex in such a way that each party is assigned to at least one vertex. The weighted graph $\gf{G}$, with boundary vertices $\partial V$ labeled by the $\N+1$ parties, defines an \textit{$\N$-party holographic graph model}. 

To any $\N$-party holographic graph model one can associate an $\N$-party entropy vector by the so-called min-cut prescription: A \textit{vertex cut} is an arbitrary subset $C\subseteq V$, and we say that $C$ is ``homologous'' to a subsystem $\uJ$ if $C\cap\partial V=\uJ$. We call a vertex cut homologous to $\uJ$ a $\uJ$-cut and we denote it by $C_{\uJ}$. Given a cut $C$, we further say that an edge of $\gf{G}$ is cut by $C$ if it connects a vertex in $C$ to one in $\comp{C}$, where $\comp{C}$ is the complement of $C$ in $V$. Since the graph is weighted, any vertex cut $C$ has a natural \textit{cost} $\norm{C}$ given by the sum of the weights of all cut edges. To any subsystem $\uJ$ we then associate an ``entropy'' given by the minimal cost among all $\uJ$-cuts, i.e.,
\begin{equation}
    \ent_{\uJ}=\text{min}\left\{\norm{C_{\uJ}},\, \text{for all}\; C_{\uJ} \right\}.
\end{equation}

Clearly, for any given $\gf{G}$, this prescription gives a vector $\Vec{\ent}(\gf{G})\in\mathbb{R}^{\D}_{+}$, and the set of all such vectors, for all possible graph models at fixed $\N$, is the $\N$-party HEC. The fact that such vectors are indeed vectors of entropies (at least approximately) of quantum states is mathematically subtle, albeit naively obvious from a physical standpoint. This result was proven in \cite{hayden2016holographic}, which also showed that entropy vectors obtained from graph models can be realized by stabilizer states. Importantly, however, not all entropy vectors that can be realized by quantum states (or even stabilizer states) are obtainable from graph models \cite{Hayden:2011ag,Bao:2015bfa},\footnote{\,For the reader who is approaching this subject from the perspective of quantum information theory, this was in fact the reason why graph models were introduced by \cite{Bao:2015bfa}, namely to study entropy inequalities satisfied by geometric states in holography which instead can be violated by more general quantum states.} and when this is possible we will say that the entropy vector is realizable \textit{by a graph model} (or equivalently, \textit{holographically realizable}). Similarly, given a KC-PMI $\pmi$, we will say that it is realizable by a graph model if there is a graph model $\gf{G}$ such that $\vec{\ent}(\gf{G})\in\text{int}(\face_{\pmi})$. 

Given an arbitrary graph model $\gf{G}$ and an arbitrary subsystem $\uJ$, there might be more than one $\uJ$-cut that minimizes the cost, and any such cut is called a \textit{min-cut}. We call a min-cut for $\uJ$ a \textit{minimal min-cut} if it is minimal, with respect to set inclusion, among all min-cuts for $\uJ$, and we denote a minimal min-cut for $\uJ$ by $\mmC_{\uJ}$. 
The following properties of minimal min-cuts were established in \cite{Avis:2021xnz}.

\begin{thm}
\label{thm:minimal-min-cuts}
        For any holographic graph model $\gf{G}$ and subsystems $\uJ,\uK$\,:
    \begin{enumerate}[label={\emph{\footnotesize \roman*)}}]
        \item for any $\uJ$, there is a unique minimal min-cut $\mmC_{\uJ}$,
        \item   $\mmC_{\uJ}\subseteq \mmC_{\uK}$ if and only if $\uJ\subseteq\uK$\,, 
        \item $\mmC_{\uJ}\cap \mmC_{\uK}=\varnothing$ if and only if $\uJ\cap\uK=\varnothing$.
    \end{enumerate}
\end{thm}
\begin{proof}
    (i), (ii) and (iii) above correspond respectively to (a), (b) and (c) of Theorem 1 in \cite{Avis:2021xnz}. Unlike \cite{Avis:2021xnz}, here we are also considering subsystems which can include the purifier, but it is obvious that these have the same properties of subsystems which do not. Finally, this result was stated in \cite{Avis:2021xnz} for graph models obtained from complete graphs, but since the weights are allowed to vanish (in which case one can simply delete the corresponding edges), it holds for any graph model.
\end{proof}

 Notice that the unique minimal min-cut $\mmC_{\uJ}$ for a subsystem $\uJ$ is strictly contained in all min-cuts for $\uJ$. We will now use minimal min-cuts and \Cref{thm:minimal-min-cuts} to prove a few results\footnote{\,Some of these results are adaptations to minimal min-cuts of results already proven in \cite{Hernandez-Cuenca:2022pst}.} about connectivity and correlation which will generalize to arbitrary KC-PMIs in \S\ref{subsec:hypergraph-rep}.\footnote{\,For earlier related results in holography see for example \cite{Headrick:2013zda}, and \cite{Hernandez-Cuenca:2022pst} for graph models.} Given an arbitrary graph model $\gf{G}$ and an arbitrary subsystem $\uJ$, we will denote by $\gf{G}_{\uJ}$ the subgraph of $\gf{G}$ induced\footnote{\,Given an arbitrary graph $\gf{G}=(V,E)$ and a subset $W\subseteq V$, the subgraph of $\gf{G}$ \textit{induced} by $W$ is the graph whose vertex set is $W$, and whose edge set is the subset of $E$ of edges connecting pairs of vertices which are both in $W$.} by the minimal min-cut $\mmC_{\uJ}$ for $\uJ$.\footnote{\,In the holographic set-up, $\gf{G}_{\uJ}$ corresponds to the bulk homology region between the spacetime boundary region $\uJ$, and the RT surface that computes the entropy of $\uJ$ and lies closest to the region $\uJ$ on the boundary.  The following theorem is the graph analog of the relation between disconnected entanglement wedges and vanishing mutual information.} 

\begin{thm}
\label{thm:subgraph-connectivity}
    For any holographic graph model $\gf{G}$, subsystem $\uI$, and instance of the mutual information $\mi(\uJ:\uK)$, with $\uI=\uJ\cup\uK$,
    \begin{equation}
    \label{eq:disconnected_ew} 
    \mi(\uJ:\uK)=0 \quad \iff \quad  \gf{G}_{\uJ\cup\uK}=\gf{G}_{\uJ} \oplus \gf{G}_{\uK},
    \end{equation}
    where $\oplus$ denotes the disjoint union of graphs (or graph sum).
\end{thm}
\begin{proof}
   Given an arbitrary graph model $\gf{G}$, consider a subsystem $\uI$, an MI instance $\mi(\uJ:\uK)$, with $\uI=\uJ\cup\uK$, and suppose that $\gf{G}_{\uJ\cup\uK}=\gf{G}_{\uJ} \oplus \gf{G}_{\uK}$. In particular this implies that given any two vertices $v_1\in \mmC_\uJ$ and $v_2\in \mmC_\uK$, there is no edge in $\gf{G}$ connecting $v_1$ and $v_2$. Therefore $\lVert\mmC_{\uJ\cup\uK}\rVert=\lVert\mmC_{\uJ}\rVert+\lVert\mmC_{\uK}\rVert$ and $\mi(\uJ:\uK)=0$.

   Conversely, suppose that $\mi(\uJ:\uK)=0$, and therefore $\lVert\mmC_{\uJ\cup\uK}\rVert=\lVert\mmC_{\uJ}\rVert+\lVert\mmC_{\uK}\rVert$. Notice that since $\uJ\cap\uK=\varnothing$, \Cref{thm:minimal-min-cuts} implies that $\mmC_{\uJ}\cap\mmC_{\uK}=\varnothing$. Consider the cut $C_{\uJ\cup\uK}=\mmC_{\uJ}\cup\mmC_{\uK}$ and suppose it induces a connected subgraph. This means that there are vertices $v_1\in \mmC_\uJ$ and $v_2\in \mmC_\uK$ such that $e=v_1v_2$ is an edge of $\gf{G}$, implying that $\lVert C_{\uJ\cup\uK}\rVert\leq\lVert\mmC_{\uJ}\rVert+\lVert\mmC_{\uK}\rVert-w_e$, where $w_e$ is the weight of $e$ and the inequality comes from the fact that there might be other pairs of vertices with this property. Since $w_e>0$, it follows that $\lVert C_{\uJ\cup\uK}\rVert<\lVert\mmC_{\uJ\cup\uK}\rVert$, which is a contradiction, since $\mmC_{\uJ\cup\uK}$ is a min-cut. Therefore the subgraph induced by $C_{\uJ\cup\uK}$ is disconnected and $C_{\uJ\cup\uK}$ is a min-cut. But by definition of minimal min-cuts $\mmC_{\uJ\cup\uK}\subseteq C_{\uJ\cup\uK}$ and $\gf{G}_{\uJ\cup\uK}$ is also disconnected, implying that $\gf{G}_{\uJ\cup\uK}=\gf{G}_{\uJ} \oplus \gf{G}_{\uK}$. 
\end{proof}

Notice that the right hand side of Eq.~\eqref{eq:disconnected_ew} means that $\gf{G}_{\uJ\cup\uK}$ is disconnected, but does not say anything about the connectivity of $\gf{G}_{\uJ}$ and $\gf{G}_{\uK}$. Notice also that the simpler assumption
\begin{equation}
\label{eq:ew-union-intersection}
\mmC_{\uJ\cup\uK}=\mmC_{\uJ}\cup\mmC_{\uK} \quad \text{and} \quad \mmC_{\uJ}\cap\mmC_{\uK}=\varnothing
\end{equation}
does not in general imply that $\mi(\uJ:\uK)=0$, as exemplified in Fig.~\ref{subfig:ew_connect}, and that if $\mi(\uJ:\uK)=0$ then the decomposition of $\gf{G}_{\uJ\cup\uK}$ given in Eq.~\eqref{eq:disconnected_ew} does not hold for arbitrary (non-minimal) min-cuts, as exemplified in Fig.~\ref{subfig:non-min_min-cut}.

We will say that an arbitrary $\N$-party graph model $\gf{G}$ (not necessarily a tree) is \textit{simple} if it  has exactly $\N+1$ boundary vertices, i.e., if there is a bijection between parties and boundary vertices. A priori, this might already seem a restriction, in the sense that one might suspect that the set of all entropy vectors that can be realized by simple graph models is a subset of the set of those that can be realized by all possible graph models. This however is not the case, since given an arbitrary non-simple graph model $\gf{G}$, it is always possible to construct a new graph model $\gf{G}'$ which is simple and gives the same entropy vector, $\Vec{\ent}(\gf{G}')=\Vec{\ent}(\gf{G})$. The construction is straightforward and it simply amounts to identifying in $\gf{G}$ all boundary vertices corresponding to the same party $\ell$ \cite{Bao:2015bfa}. The following observation will be particularly convenient when working with simple graph models.

\begin{lemma}
\label{lem:simple-G-homology}
    For any simple graph model $\gf{G}$ and subsystem $\uI$, let 
    \begin{equation}
        \gf{G}_{\uI}=\bigoplus_{i=1}^m \; \gf{G}^{(i)}
    \end{equation}
    be the decomposition of $\gf{G}_\uI$ into its connected components. Then each connected component $\gf{G}^{(i)}$ is the subgraph of $\gf{G}$ induced by the minimal min-cut for some subsystem $\uJ_i\subseteq\uI$, and $\gf{G}_{\uI}$ takes the form
    \begin{equation}
    \label{eq:ew-decomposition}
        \gf{G}_{\uI}=\bigoplus_{i=1}^m \; \gf{G}_{\uJ_i},
    \end{equation}
    where the set $\{\uJ_i\}_{i\in[m]}$ is a partition of $\uI$.
\end{lemma}
\begin{proof}
     Consider a simple graph model $\gf{G}$, and suppose that for a subsystem $\uI$ the subgraph $\gf{G}_\uI$ is disconnected, so that $m\geq 2$ (otherwise the statement is trivial). Let $\gf{G}^{(i)}$ be an arbitrary connected component of $\gf{G}_\uI$ and $V(\gf{G}^{(i)})$ its set of vertices. Notice that $V(\gf{G}^{(i)})\cap\partial V\neq\varnothing$, otherwise $\mmC_\uI\setminus V(\gf{G}^{(i)})$ is still an $\uI$-cut and $\lVert\mmC_\uI\setminus V(\gf{G}^{(i)})\rVert<\lVert\mmC_\uI\rVert$, contradicting the assumption that $\mmC_\uI$ is a min-cut. Furthermore, since $\gf{G}$ is simple, there is by definition a bijection between $\partial V$ and $\nsp$, implying that $V(\gf{G}^{(i)})\cap\partial V=\uJ$ for some non-empty subsystem $\uJ$, and therefore that $V(\gf{G}^{(i)})$ is a $\uJ$-cut that we denote by $C_{\uJ}^*$. 
    
    To prove the lemma, it remains to be shown that $C_{\uJ}^*$ is the minimal min-cut for $\uJ$, i.e., $C_{\uJ}^*=\mmC_\uJ$, or in other words that $\gf{G}^{(i)}=\gf{G}_\uJ$. Since we are assuming $m\geq 2$, and the above reasoning applies to each connected component $\gf{G}^{(i)}$, it follows that $\uJ\subset\uI$, and by \Cref{thm:minimal-min-cuts}, that $\mmC_\uJ\subset\mmC_\uI$. Notice that since $\gf{G}^{(i)}$ is a $\uJ$-cut, it has in common with $V(\gf{G}_\uJ)$ at least all the boundary vertices $\uJ$. Suppose now that $\gf{G}_\uJ$ has vertices in some other connected component $\gf{G}^{(j)}$ of $\gf{G}_\uI$, where $j\neq i$. Since these are not boundary vertices, we can delete them from $V(\gf{G}_\uJ)$ to obtain a new $\uJ$-cut with smaller cost, contradicting the fact that $\mmC_\uJ$ is a min-cut. This implies that $\mmC_\uJ\subseteq C_{\uJ}^*$. However if the inclusion is strict, there is a contradiction with the assumption that $\gf{G}^{(i)}$ was a connected component of the minimal min-cut for $\uI$, implying that $C_{\uJ}^*=\mmC_\uJ$ and $\gf{G}^{(i)}=\gf{G}_\uJ$. 
    
    Finally, notice that the fact that $\{\uJ_i\}_{i\in[m]}$ is a partition of $\uI$ follows immediately from the fact that $\{V(\gf{G}_{\uJ_i})\}_{i\in[m]}$ is a partition of $V(\gf{G}_\uI)$. 
\end{proof}

\begin{figure}[tbp]
    \centering
    \begin{subfigure}{0.45\textwidth}
    \centering
    \begin{tikzpicture}[scale=1]
    \draw (0,0) -- (-1,0);
    \draw (0,0) -- (-1,-1);
    \draw (0,0) -- (-1,1);
    \draw (0,0) -- (2,0);
    \draw (2,0) -- (3,0);
    \draw (2,0) -- (3,1);
    \draw (2,0) -- (3,-1);
    
    \filldraw (-1,-1) circle (2pt);
    \filldraw (-1,0) circle (2pt);
    \filldraw (-1,1) circle (2pt);
    \filldraw (3,-1) circle (2pt);
    \filldraw (3,0) circle (2pt);
    \filldraw (3,1) circle (2pt);
    
    \filldraw[fill=bulkcol] (0,0) circle (2pt);
    \filldraw[fill=bulkcol] (2,0) circle (2pt);
    
    \node[] () at (-1.3,1) {{\scriptsize $1$}};
    \node[] () at (-1.3,0) {{\scriptsize $2$}};
    \node[] () at (-1.3,-1) {{\scriptsize $3$}};
    \node[] () at (3.3,1) {{\scriptsize $4$}};
    \node[] () at (3.3,0) {{\scriptsize $5$}};
    \node[] () at (3.3,-1) {{\scriptsize $0$}};
    
    \node[red] () at (-0.5,0.8) {{\scriptsize $2$}};
    \node[red] () at (-0.5,0.2) {{\scriptsize $2$}};
    \node[red] () at (-0.5,-0.75) {{\scriptsize $1$}};
    \node[red] () at (1,0.2) {{\scriptsize $1$}};
    \node[red] () at (2.5,0.8) {{\scriptsize $2$}};
    \node[red] () at (2.5,0.2) {{\scriptsize $2$}};
    \node[red] () at (2.5,-0.75) {{\scriptsize $1$}};
    \end{tikzpicture}
    \subcaption[]{}
    \label{subfig:ew_connect}
    \end{subfigure}
    \hspace{0.2cm}
    \begin{subfigure}{0.45\textwidth}
    \centering
    \begin{tikzpicture}[scale=1]
    \draw (0,0) -- (-1,-1);
    \draw (0,0) -- (-1,1);
    \draw (0,0) -- (2,0);
    \draw (2,0) -- (3,1);
    \draw (2,0) -- (3,-1);
    
    \filldraw (-1,-1) circle (2pt);
    \filldraw (-1,1) circle (2pt);
    \filldraw (3,-1) circle (2pt);
    \filldraw (3,1) circle (2pt);
    
    \filldraw[fill=bulkcol] (0,0) circle (2pt);
    \filldraw[fill=bulkcol] (2,0) circle (2pt);
    
    \node[] () at (-1.3,1) {{\scriptsize $2$}};
    \node[] () at (-1.3,-1) {{\scriptsize $1$}};
    \node[] () at (3.3,1) {{\scriptsize $3$}};
    \node[] () at (3.3,-1) {{\scriptsize $0$}};
    
    \node[red] () at (-0.5,0.8) {{\scriptsize $1$}};
    \node[red] () at (-0.5,-0.75) {{\scriptsize $1$}};
    \node[red] () at (1,-0.25) {{\scriptsize $1$}};
    \node[red] () at (2.5,0.8) {{\scriptsize $1$}};
    \node[red] () at (2.5,-0.75) {{\scriptsize $1$}};

    \node[] () at (0.1,0.3) {{\scriptsize $\sigma_1$}};
    \node[] () at (2,0.3) {{\scriptsize $\sigma_2$}};
    \end{tikzpicture}
    \subcaption[]{}
    \label{subfig:non-min_min-cut}
    \end{subfigure}

    \par\bigskip\bigskip

    \begin{subfigure}{0.45\textwidth}
    \centering
    \begin{tikzpicture}[scale=1]
    \draw (0,0) -- (-1,-1);
    \draw (0,0) -- (-1,1);
    \draw (0,0) -- (2,0);
    \draw (2,0) -- (3,1);
    \draw (2,0) -- (3,-1);
    \draw (1,0) -- (1,1);
    
    \filldraw (-1,-1) circle (2pt);
    \filldraw (-1,1) circle (2pt);
    \filldraw (3,-1) circle (2pt);
    \filldraw (3,1) circle (2pt);
    \filldraw (1,1) circle (2pt);
    
    \filldraw[fill=bulkcol] (0,0) circle (2pt);
    \filldraw[fill=bulkcol] (2,0) circle (2pt);
    \filldraw[fill=bulkcol] (1,0) circle (2pt);
    
    \node[] () at (-1.3,1) {{\scriptsize $2$}};
    \node[] () at (-1.3,-1) {{\scriptsize $1$}};
    \node[] () at (1,1.3) {{\scriptsize $3$}};
    \node[] () at (3.3,1) {{\scriptsize $2$}};
    \node[] () at (3.3,-1) {{\scriptsize $0$}};
    
    \node[red] () at (-0.5,0.8) {{\scriptsize $1$}};
    \node[red] () at (-0.5,-0.75) {{\scriptsize $1$}};
    \node[red] () at (0.5,-0.25) {{\scriptsize $1$}};
    \node[red] () at (1.5,-0.25) {{\scriptsize $1$}};
    \node[red] () at (2.5,0.8) {{\scriptsize $1$}};
    \node[red] () at (2.5,-0.75) {{\scriptsize $1$}};
    \node[red] () at (1.2,0.5) {{\scriptsize $3$}};

    \node[] () at (1,-0.25) {{\scriptsize $\sigma$}};
    \end{tikzpicture}
    \subcaption[]{}
    \label{subfig:non-cimple-G}
    \end{subfigure}
    \hspace{0.2cm}
    \begin{subfigure}{0.45\textwidth}
    \centering
    \begin{tikzpicture}[scale=1]
    \draw (0,0) -- (-1,-1);
    \draw (0,0) -- (-1,1);
    \draw (0,0) -- (2,0);
    \draw (2,0) -- (3,1);
    \draw (2,0) -- (3,-1);
    \draw (1,0) -- (1,1);
    
    \filldraw (-1,-1) circle (2pt);
    \filldraw (-1,1) circle (2pt);
    \filldraw (3,-1) circle (2pt);
    \filldraw (3,1) circle (2pt);
    \filldraw (1,1) circle (2pt);
    
    \filldraw[fill=bulkcol] (0,0) circle (2pt);
    \filldraw[fill=bulkcol] (2,0) circle (2pt);
    \filldraw[fill=bulkcol] (1,0) circle (2pt);
    
    \node[] () at (-1.3,1) {{\scriptsize $2$}};
    \node[] () at (-1.3,-1) {{\scriptsize $1$}};
    \node[] () at (1,1.3) {{\scriptsize $3$}};
    \node[] () at (3.3,1) {{\scriptsize $4$}};
    \node[] () at (3.3,-1) {{\scriptsize $0$}};
    
    \node[red] () at (-0.5,0.8) {{\scriptsize $1$}};
    \node[red] () at (-0.5,-0.75) {{\scriptsize $1$}};
    \node[red] () at (0.5,-0.25) {{\scriptsize $1$}};
    \node[red] () at (1.5,-0.25) {{\scriptsize $1$}};
    \node[red] () at (2.5,0.8) {{\scriptsize $1$}};
    \node[red] () at (2.5,-0.75) {{\scriptsize $1$}};
    \node[red] () at (1.2,0.5) {{\scriptsize $3$}};

    \node[] () at (1,-0.25) {{\scriptsize $\sigma$}};
    \end{tikzpicture}
    \subcaption[]{}
    \label{subfig:ew-decomposition}
    \end{subfigure}
    \caption{Graph models exemplifying some of the results and observations presented in the main text. Boundary vertices are depicted in black and are labeled by the parties (black numbers), gray vertices are bulk vertices, and red numbers indicate the weights of the various edges. (a) $\gf{G}_{12}$ and $\gf{G}_{45}$ satisfy Eq.~\eqref{eq:ew-union-intersection} but $\gf{G}_{1245}$ is connected and $\mi(12:45)>0$. (b) $\gf{G}_{23}=\gf{G}_2\oplus\gf{G}_3$ but the (non-minimal) min-cut $C_{23}=\{2,\sigma_1,\sigma_2,3\}$ induces a connected subgraph of $\gf{G}$. (c) A non-simple graph model (two boundary vertices are labeled by party 2) where $\mmC_{120}=V\setminus\{3,\sigma\}$, $\gf{G}_{012}$ is disconnected, but neither of its connected components is induced by the minimal min-cut (or in fact by any $\uJ$-cut) for a subsystem $\uJ$. By identifying the two boundary vertices labeled by party 2, the graph becomes simple and $\gf{G}_{012}$ becomes connected, in agreement with \Cref{lem:simple-G-homology}. (d) A graph model with the same topology and weights as the example in (c) but with a different labeling of boundary vertices, so that it is simple. Similarly to (c), $\mmC_{1240}=V\setminus\{3,\sigma\}$ but now we have the decomposition $\gf{G}_{1240}=\gf{G}_{12}\oplus\gf{G}_{40}$. Furthermore, $\mi(12:40)=0$ while any other MI instance involving precisely the parties $1,2,4,0$ (eg., $\mi(14:20)$) is strictly positive.}
    \label{fig:connectivity-examples}
\end{figure}

Notice that in \Cref{lem:simple-G-homology} the assumption that $\gf{G}$ is simple is crucial. Indeed even though a decomposition of $\gf{X}$ into its connected components is obviously possible for any graph model (not necessarily simple), it is in general not the case that each such component is induced by a $\uJ$-cut for some subsystem $\uJ$ (cf., Fig.~\ref{subfig:non-cimple-G}). 

We now derive an immediate corollary of \Cref{thm:subgraph-connectivity}, once restricted to simple graph models, that together with \Cref{thm:subgraph-connectivity} itself will have direct generalizations to structural properties of arbitrary KC-PMIs (cf., \Cref{thm:connectivity_of_H}). 

\begin{cor}
\label{cor:ew-decomposition}
    For any simple graph model $\gf{G}$ and subsystem $\uI$, let
    \begin{equation}
    \label{eq:ew-decomposition-v2}
        \gf{G}_{\uI}=\bigoplus_{i=1}^m \; \gf{G}_{\uI_i}
    \end{equation}
    be the decomposition of $\gf{G}_{\uI}$ into the sum of its connected components as given by \Cref{lem:simple-G-homology}.
    Then an instance of the mutual information $\mi(\uJ:\uK)$, with $\uJ\cup\uK=\uI$, vanishes if and only if $m\geq 2$ and for each $i\in [m]$, either $\uI_i\cap\uJ=\varnothing$ (so that $\uI_i\subseteq\uK$) or vice versa (swapping $\uJ$ and $\uK$).
\end{cor}
\begin{proof}
    Consider a simple graph model $\gf{G}$ and subsystem $\uI$ such that $\gf{G}_\uI$ decomposes as in Eq.~\eqref{eq:ew-decomposition-v2}, and suppose that $m\geq 2$. Consider an MI instance $\mi(\uJ:\uK)$, with $\uJ\cup\uK=\uI$, such that for each $i\in [m]$, $\uI_i\cap\uJ=\varnothing$ and therefore $\uI_i\subseteq\uK$, or vice versa. We can combine the terms in Eq.~\eqref{eq:ew-decomposition-v2} and rewrite $\gf{G}_{\uI}$ as $\gf{G}_{\uI}=\gf{G}_{\uJ}\oplus\gf{G}_{\uK}$. \Cref{thm:subgraph-connectivity} then implies that $\mi(\uJ:\uK)=0$. 

    In the other direction, consider an MI instance $\mi(\uJ:\uK)$, again with $\uJ\cup\uK=\uI$, but now such that for some $\uI_i$ we have $\uI_i\cap\uJ\neq\varnothing$ and $\uI_i\cap\uK\neq\varnothing$. Since by assumption $\gf{G}_{\uI_i}$ is connected, \Cref{thm:subgraph-connectivity} implies that $\mi(\uI_i\cap\uJ:\uI_i\cap\uK)>0$, and by strong subadditivity  $\mi(\uJ:\uK)\geq\mi(\uI_i\cap\uJ:\uI_i\cap\uK)>0$.
    
    Finally if $m=1$, $\uI_1=\uI$ and since for any MI instance $\mi(\uJ:\uK)$ with $\uJ\cup\uK=\uI$ we have by definition $\uJ\neq\varnothing$ and $\uK\neq\varnothing$, it follows that we necessarily have $\uI\cap\uJ\neq\varnothing$ and $\uI\cap\uK\neq\varnothing$, implying that $\mi(\uJ:\uK)>0$ for all possible choices of $\uJ$ and $\uK$.
\end{proof}

An example of \Cref{cor:ew-decomposition} is shown in Fig.~\ref{subfig:ew-decomposition}. Notice that while $\mi(\uJ:\uK)=0$ implies that $\gf{G}_{\uJ\cup\uK}$ is disconnected (cf., \Cref{thm:subgraph-connectivity}), \Cref{cor:ew-decomposition} shows that even for simple graphs, $\mi(\uJ:\uK)>0$ in general does not imply that $\gf{G}_{\uJ\cup\uK}$ is connected, only that there is a path from at least one boundary vertex in $\uJ$ to at least one in $\uK$. 

In \S\ref{sec:necessary}, we will be interested in a particular restricted class of graph models called \textit{simple trees}, which were first introduced in \cite{Hernandez-Cuenca:2022pst} and are defined as simple graph models with tree topology. We will derive a necessary condition that allows one to test whether a given entropy vector can be realized by a simple tree or not. An immediate implication of this condition will be the rigorous proof of the intuitive fact that not all entropy vectors realizable by graph models can be realized by simple trees (for $\N > 2$). Using \Cref{cor:ew-decomposition}, we can, however, already present the intuition behind such a necessary condition.

For sufficiently large $\N$, consider an entropy vector $\vec{\ent}$ in the interior of the SAC; then all instances of mutual information are positive when evaluated at $\vec{\ent}$. If there is a simple tree $\gf{T}$ that realizes $\vec{\ent}$, it follows from \Cref{cor:ew-decomposition} that $\gf{T}_\uI$ is connected for each subsystem $\uI$. However, this fact is in tension with the tree topology, since on any tree (with at least three vertices) one can always find a subtree $\gf{T}'$ such that its complement is disconnected. In more general situations, where for given $\vec{\ent}$ some instances of mutual information might vanish, the machinery of $\bsets$ and the correlation hypergraph, developed in \S\ref{sec:bsets-corr-hyp}, will be used to precisely capture the relevant combinatorial data about the KC-PMI associated to $\vec{\ent}$ that is necessary to determine whether this tension can be resolved.

\section{Correlation hypergraph representation of KC-PMIs}
\label{sec:bsets-corr-hyp}

\subsection{Definition of \texorpdfstring{$\bsets$}{beta-sets} and related classification}
\label{subsec:beta-sets}

Having reviewed the necessary definitions in \S\ref{subsec:KC-review}, we now introduce the main new tool of this work, namely a partition of the MI-poset which will allow us to describe KC-PMIs more efficiently. As we mentioned in the introduction, the culmination of this analysis, presented in \S\ref{subsec:hypergraph-rep}, will be the generalization of the relationship between correlation and connectivity that we reviewed in \S\ref{subsec:graph_review} for holographic graph models, to all KC-PMIs.  
Since the elements of the MI-poset are instances of the mutual information (whose arguments comprise two disjoint subsystems),  from now on we will implicitly assume that $\N\geq 2$.

The desired partition can be regarded as a natural coarse-graining of the MI poset
and will be dubbed $\bset$ (as shorthand for bi-partition set).
\begin{defi}[$\bset$]
\label{defi:b-set}
    Given $\N$ and a collection of parties $\uI$, we define the $\beta$\emph{-set} of $\uI$ as the set of instances of the mutual information corresponding to all possible non-trivial bipartitions of $\uI$, 
    \begin{equation}
    \label{eq:b-set}
        \bs{\uI}=\{\mi(\uJ:\uK),\ \text{with}\;\; \uJ\cup\uK=\uI,\;  \ \uJ\cap\uK=\varnothing,\;\; \text{and}\;\; |\uJ|,|\uK|\geq 1\}.
    \end{equation}
    We denote by $\mathfrak{B}_{\N}$ the set of all $\bsets$, for all $\uI\subseteq\nsp$ with cardinality $|\uI|\geq 2$.
\end{defi}

Given an arbitrary down-set $\ds$, we want to explore the structure of its complementary up-set $\cds$ (i.e., $\cds=\mgs_{\N}\setminus\ds$) from the point of view of $\bsets$, and for this purpose it is useful to introduce some dedicated terminology. Given $\ds$, we will say that a $\bset$ $\bs{\uI}$ is \emph{positive} if every element of $\bs{\uI}$ is in $\cds$. Similarly, we will say that $\bs{\uI}$ is \textit{vanishing} if every element of $\bs{\uI}$ is in $\ds$, and \textit{partial} (short for `partially vanishing') if $\bs{\uI}$ has at least one element in $\ds$ and at least one in $\cds$. Notice that a down-set $\ds$ is simply a subset of $\mgs_\N$ with a particular order structure, and that the MI instances in $\ds$ have no associated value. The motivation for this choice of terminology is that when $\ds$ \textit{is} a KC-PMI $\pmi$, then for every entropy vector in the interior of the face $\face_\pmi$, each MI instance in $\ds$ is vanishing and each one in $\cds$ is strictly positive. 

There is also a deeper classification of $\bsets$, again for a given $\ds$, which will play a key role in what follows, but in order to formulate it, we need to introduce a few more definitions. For an arbitrary poset $\p$, it is often convenient to consider down-sets and up-sets \textit{generated} by a subset $\x\subseteq\p$, which we define respectively as follows:
\begin{align}
\label{eq:order-sets}
    \begin{split}
        & \downarrow\! \x = \{y\in\p,\; \text{ such that } \exists\, x\in\x\;\; \text{with}\;\; y\preceq x\}\\
        & \uparrow\! \x = \{y\in\p,\; \text{ such that } \exists\, x\in\x\;\; \text{with}\;\; y\succeq x\}.
    \end{split}
\end{align}
To simplify notation, we will write the down-set generated by a single element (called principal down-set) as $\downarrow\! x$ instead of $\downarrow\! \{x\}$ (and similarly for up-sets). For an arbitrary down-set $\ds$ in the MI-poset, we denote by $\ac$ and $\hat{\mathcal{A}}$ the antichains that generate $\cds$ and $\ds$,\footnote{\,Any subset $\x$ of a poset $\p$ can naturally be seen as a poset with the induced partial order.
Given these order relations, a totally ordered collection of elements is called a \textit{chain}, and a collection of elements which are pairwise incomparable is an \textit{antichain}.}
\begin{align}
\begin{split}
    & \ac=\text{Min}(\cds)\quad \text{or equivalently} \quad \uparrow\!\ac=\cds \\
    & \hat{\mathcal{A}}=\text{Max}(\ds) \quad \text{or equivalently} \quad \downarrow\!\hat{\mathcal{A}}=\ds, \\
\end{split}
\end{align}
where $\text{Min}(\cds)$ and $\text{Max}(\ds)$ are the sets of \textit{minimal elements} of $\cds$ and \textit{maximal elements} of $\ds$, respectively.\footnote{\,An element $x\in\p$ of a poset $\p$ is \textit{minimal} if there is no element $y\in\p$ such that $y\prec x$. This is not to be confused with \textit{the minimum element} of $\p$ (which does not necessarily exist), defined as the element $x_{\perp}$ such that $x_{\perp}\preceq y$ for all $y\in\p$. The \textit{maximal} elements, $\text{Max}(\p)$, and \textit{the maximum element} $x_{\top}$ are defined analogously.} Notice that since the poset is finite, these sets always exist. Furthermore, for any $\mi\in\cds$ such that $\mi\notin\ac$, there is always some $\mi'\in\ac$ such that $\mi'\prec\mi$ (and similarly for $\ds$ and $\hat{\mathcal{A}}$). Relative to the elements of $\ac$, we will then say that for a given $\ds$, a $\bset$ is \textit{essential} if it contains at least one element of $\ac$, \textit{completely essential} if all its elements are in $\ac$, and \textit{non-essential} if none of its elements are in $\ac$.

Using these definitions, we then introduce the following collections of $\bsets$ for a given down-set $\ds$.

\begin{defi}[Fundamental sets of $\bsets$]
\label{defi:special-sets}
Given $\N$ and any down-set $\ds$ of the \emph{MI}-poset, we define the following fundamental sets of $\bsets$:
\begin{enumerate}[label={\emph{\footnotesize \roman*)}},wide=0pt]
\item classification based on $\ds$ and $\cds$ 
    {\small
    \begin{align*}
    & \pos(\ds) \coloneq   \{ \bs{\uI} , \ \mi>0\ \ \forall\ \mi \in \bs{\uI} \},\, \text{the set of \emph{positive} $\bsets$ of $\ds$} \\
    & \van(\ds) \coloneq   \{ \bs{\uI} , \ \mi=0\ \ \forall\ \mi \in \bs{\uI} \},\, \text{the set of \emph{vanishing} $\bsets$ of $\ds$}  \\
    & \parv(\ds) \coloneq   \{ \bs{\uI} , \ \exists\ \mi_1, \mi_2 \in\bs{\uI} \text{ with } \mi_1=0 \text{ and } \mi_2>0 \},\, \text{the set of \emph{partial} $\bsets$ of $\ds$}
    \end{align*}
    \par}
\item   classification based on the elements of $\ac$ 
    {\small
    \begin{align*}
    &\hspace{-0.95cm} \ess(\ds) \coloneq   \{ \bs{\uI} , \ \exists\ \mi  \in\bs{\uI} \text{ with } \mi \in \ac \},\, \text{the set of \emph{essential} $\bsets$ of $\ds$}  \\
    &\hspace{-0.95cm} \cess(\ds) \coloneq \{ \bs{\uI} , \ \mi \in \ac \ \ \forall\ \mi \in \bs{\uI} \},\, \text{the set of \emph{completely essential} $\bsets$ of $\ds$}\\
    &\hspace{-0.95cm} \ness(\ds) \coloneq \{ \bs{\uI} , \ \nexists\ \mi  \in\bs{\uI} \text{ with } \mi \in \ac \},\, \text{the set of \emph{non-essential} $\bsets$ of $\ds$}
\end{align*}
\par}
\end{enumerate}

\noindent where each $\mi$ is an instance of mutual information whose arguments partition the system $\uI$ as in Eq.~\eqref{eq:b-set}.
\end{defi}

\begin{figure}
    \centering
    \begin{subfigure}{0.49\textwidth}
    \centering
        \renewcommand{\arraystretch}{1.4}
        \begin{tabular}{l||c|c|c|}
             & $\pos$ & $\parv$ & $\van$ \\
            \hline \hline
            $\cess$ & {\cellcolor{Mcol} $\cess$} & $\varnothing$ & $\varnothing$ \\
            \hline
            $\ncess$ & {\cellcolor{EminusMcol} $\mathfrak{B}^{^{(\!\vee\setminus\vee\!\!\!\vee)\cap+}}$} & {\cellcolor{EminusPcol} $\mathfrak{B}^{^{\!\vee\setminus +}}$} & $\varnothing$ \\
            \hline
            $\ness$ & {\cellcolor{PminusEcol} $\pnote$} & {\cellcolor{Qcol} $\mathfrak{B}^{^{\!\oplus\setminus\vee}}$} & {\cellcolor{Vcol} $\van$} \\
            \hline
        \end{tabular}
    \caption{}
    \label{fig:classification-ds}
    \end{subfigure}
    \begin{subfigure}{0.49\textwidth}
    \centering
        \renewcommand{\arraystretch}{1.4}
        \begin{tabular}{l||c|c|c|}
            & $\pos$ & $\parv$ & $\van$ \\
            \hline \hline
            $\cess$ & {\cellcolor{Mcol} $\cess$} & $\varnothing$ & $\varnothing$ \\
            \hline
            $\ncess$ & {\cellcolor{EminusMcol} $\ncess$} & $\varnothing$ & $\varnothing$ \\
            \hline
            $\ness$ & {\cellcolor{PminusEcol} $\pnote$} & {\cellcolor{Qcol} $\parv$} & \cellcolor{Vcol}$\van$ \\
            \hline
        \end{tabular}
    \caption{}
    \label{fig:classification-pmi}
    \end{subfigure}
    \caption{
    A partition of the set $\mathfrak{B}_{\N}$ of all $\bsets$, based on the two classifications in \Cref{defi:special-sets}, (a) for arbitrary down-sets, and (b) for KC-PMIs, 
    where for clarity we have omitted the explicit dependence on the choice of $\ds$ or $\pmi$. 
    The coloring of the cells is a convention which we adopt in subsequent figures to indicate the corresponding $\bset$ classification.
    The simplification of (b) compared to (a) is the key consequence of \Cref{thm:positiveness}. Each cell in each table corresponds to the intersection of the two classes indicated by the row and column. For example, $\cess$ (yellow cell) is the intersection of $\cess$ and $\pos$, and it equals $\cess$ because of the inclusion $\cess\subseteq\pos$. Similar simple relations can be used to derive the expressions in other cells and to see that some are necessarily empty; we leave the details as an exercise for the reader. To compactly represent the sets in the various cells, we have adopted a shorthand notation where the exponent denotes the operations that needs to be taken among various subsets of $\mathfrak{B}_\N$. For example, for positive essential but not completely essential $\bsets$, $\ncess\cap\pos$ is shortened to $\mathfrak{B}^{^{(\!\vee\setminus\vee\!\!\!\vee)\cap+}}$.
    Notice that these expressions are not necessarily unique; for example we could equivalently re-express $\mathfrak{B}^{^{\!\vee\setminus +}}$ as $\mathfrak{B}^{^{\!\oplus\setminus\bigtriangledown}}$.
    }
    \label{fig:classification-all}
\end{figure}

Notice that since any $\bset$ is either positive, vanishing, or partial, for any down-set we obviously have 
\begin{align}
\label{eq:trivial_relations}
   & \pos(\ds)\cap\van(\ds) = \pos(\ds)\cap\parv(\ds) = \van(\ds)\cap\parv(\ds) = \varnothing \nonumber\\
    & \pos(\ds) \cup  \parv(\ds) \cup \van(\ds)  = \mathfrak{B}_{\N},
\end{align}
and similarly, since each $\bset$ is either essential or not, 
\begin{align}
\label{eq:trivial_relations-2}
   & \ess(\ds)\cap\ness(\ds) = \cess(\ds)\cap\ness(\ds) = \varnothing \nonumber\\
    & \ess(\ds)\cup\ness(\ds)=\mathfrak{B}_{\N}.
\end{align}
Furthermore, since any completely essential $\bset$ is essential, $\cess(\ds)\subseteq\ess(\ds)$, and we denote by $\ncess(\ds)\coloneq\ess(\ds)\setminus\cess(\ds)$ the set of $\bsets$ which are essential but not completely essential.

In the following it will also be useful to consider how these two classifications ``intersect'' to give a finer partition of $\mathfrak{B}_\N$; a schematic representation of the various cases is shown in Fig.~\ref{fig:classification-ds}. In particular, we will make an extensive use of the shorthand notation $\pnote\coloneq\pos\setminus\ess$ for the set of positive but not essential $\bsets$.
As a concrete example, the full classification of $\bsets$, according to the classes in Fig.~\ref{fig:classification-ds}, is shown in Fig.~\ref{fig:D-not-P} for the down-set of the $\N=3$ MI-poset with
\begin{equation}
\label{eq:ds-not-p-ac}
    \ac=\{\mi(1:3),\, \mi(01:2),\, \mi(12:0),\, \mi(02:3),\, \mi(03:2),\, \mi(23:0),\, \mi(13:2)\}.
\end{equation}

\begin{figure}[tb]    
    \centering
        \begin{tikzpicture}

    \draw[draw=none,rounded corners,fill=PminusEcol] (1.3,2.6) rectangle (12.1,3.4);

    \draw[draw=none,rounded corners,fill=EminusPcol] (-0.55,0.6) rectangle (2.9,1.4);
    \draw[draw=none,rounded corners,fill=Qcol] (3.05,0.6) rectangle (6.5,1.4);
    \draw[draw=none,rounded corners,fill=Mcol] (6.65,0.6) rectangle (10.1,1.4);
    \draw[draw=none,rounded corners,fill=EminusMcol] (10.25,0.6) rectangle (13.7,1.4);

    \draw[draw=none,rounded corners,fill=Vcol] (1.2,-1.4) rectangle (2.3,-0.6);
    \draw[draw=none,rounded corners,fill=Vcol] (3.2,-1.4) rectangle (4.3,-0.6);
    \draw[draw=none,rounded corners,fill=Vcol] (5.2,-1.4) rectangle (6.3,-0.6);
    \draw[draw=none,rounded corners,fill=Vcol] (7.2,-1.4) rectangle (8.3,-0.6);
    \draw[draw=none,rounded corners,fill=Mcol] (9.2,-1.4) rectangle (10.3,-0.6);
    \draw[draw=none,rounded corners,fill=Vcol] (11.2,-1.4) rectangle (12.3,-0.6);

    \node () at (2.2,3) {{\scriptsize $\mi(0\!:\!123)_{_{\!+}}$}};
     \node () at (3.7,3) {{\scriptsize $\mi(1\!:\!023)_{_{\!+}}$}};
     \node () at (5.2,3) {{\scriptsize $\mi(2\!:\!013)_{_{\!+}}$}};
     \node () at (6.7,3) {{\scriptsize $\mi(3\!:\!012)_{_{\!+}}$}};
     \node () at (8.2,3) {{\scriptsize $\mi(01\!:\!23)_{_{\!+}}$}};
     \node () at (9.7,3) {{\scriptsize $\mi(02\!:\!13)_{_{\!+}}$}};
     \node () at (11.2,3) {{\scriptsize $\mi(03\!:\!12)_{_{\!+}}$}};
     
    \node () at (0.1,1) {{\scriptsize $\mi(01\!:\!2)_{_{\!+}}$}};
    \node () at (1.2,1) {{\scriptsize $\mi(02\!:\!1)_{_0}$}};
    \node () at (2.3,1) {{\scriptsize $\mi(12\!:\!0)_{_{\!+}}$}};
    \node () at (3.7,1) {{\scriptsize $\mi(01\!:\!3)_{_{\!+}}$}};
    \node () at (4.8,1) {{\scriptsize $\mi(03\!:\!1)_{_{\!+}}$}};
    \node () at (5.9,1) {{\scriptsize $\mi(13\!:\!0)_{_0}$}};
    \node () at (7.3,1) {{\scriptsize $\mi(02\!:\!3)_{_{\!+}}$}};
    \node () at (8.4,1) {{\scriptsize $\mi(03\!:\!2)_{_{\!+}}$}};
    \node () at (9.5,1) {{\scriptsize $\mi(23\!:\!0)_{_{\!+}}$}};
    \node () at (10.9,1) {{\scriptsize $\mi(12\!:\!3)_{_{\!+}}$}};
    \node () at (12,1) {{\scriptsize $\mi(13\!:\!2)_{_{\!+}}$}};
    \node () at (13.1,1) {{\scriptsize $\mi(23\!:\!1)_{_{\!+}}$}};

    \node () at (1.8,-1) {{\scriptsize $\mi(0\!:\!1)_{_0}$}};
     \node () at (3.8,-1) {{\scriptsize $\mi(0\!:\!2)_{_0}$}};
     \node () at (5.8,-1) {{\scriptsize $\mi(1\!:\!2)_{_0}$}};
     \node () at (7.8,-1) {{\scriptsize $\mi(0\!:\!3)_{_0}$}};
     \node () at (9.8,-1) {{\scriptsize $\mi(1\!:\!3)_{_{\!+}}$}};
     \node () at (11.8,-1) {{\scriptsize $\mi(2\!:\!3)_{_0}$}};

    \node (0123) at (6.7,3.7) {{\scriptsize $\bs{0123}$}};
    
    \node (012) at (1.2,1.7) {{\scriptsize $\bs{012}$}};
    \node (013) at (4.8,1.7) {{\scriptsize $\bs{013}$}};
    \node (023) at (8.4,1.7) {{\scriptsize $\bs{023}$}};
    \node (123) at (12,1.7) {{\scriptsize $\bs{123}$}};

    \node (01) at (1.8,-0.3) {{\scriptsize $\bs{01}$}};
    \node (02) at (3.8,-0.3) {{\scriptsize $\bs{02}$}};
    \node (12) at (5.8,-0.3) {{\scriptsize $\bs{12}$}};
    \node (03) at (7.8,-0.3) {{\scriptsize $\bs{03}$}};
    \node (13) at (9.8,-0.3) {{\scriptsize $\bs{13}$}};
    \node (23) at (11.8,-0.3) {{\scriptsize $\bs{23}$}};

    \draw[-,gray,very thin] (6.7,2.6) -- (012.north);
    \draw[-,gray,very thin] (6.7,2.6) -- (013.north);
    \draw[-,gray,very thin] (6.7,2.6) -- (023.north);
    \draw[-,gray,very thin] (6.7,2.6) -- (123.north);

    \draw[-,gray,very thin] (1.175,0.6) -- (01.north);
    \draw[-,gray,very thin] (1.175,0.6) -- (02.north);
    \draw[-,gray,very thin] (1.175,0.6) -- (12.north);
    \draw[-,gray,very thin] (4.775,0.6) -- (01.north);
    \draw[-,gray,very thin] (4.775,0.6) -- (03.north);
    \draw[-,gray,very thin] (4.775,0.6) -- (13.north);
    \draw[-,gray,very thin] (8.375,0.6) -- (02.north);
    \draw[-,gray,very thin] (8.375,0.6) -- (03.north);
    \draw[-,gray,very thin] (8.375,0.6) -- (23.north);
    \draw[-,gray,very thin] (11.975,0.6) -- (12.north);
    \draw[-,gray,very thin] (11.975,0.6) -- (13.north);
    \draw[-,gray,very thin] (11.975,0.6) -- (23.north);

    \draw[-,very thick] (-0.2,0.75) -- (0.3,0.75);
    \draw[-,very thick] (2,0.75) -- (2.5,0.75);
    \draw[-,very thick] (7,0.75) -- (7.5,0.75);
    \draw[-,very thick] (8.1,0.75) -- (8.6,0.75);
    \draw[-,very thick] (9.2,0.75) -- (9.7,0.75);
    \draw[-,very thick] (11.7,0.75) -- (12.2,0.75);

    \draw[-,very thick] (9.5,-1.25) -- (10,-1.25);
     
    \end{tikzpicture}
    \caption{An example of a down-set $\ds$ in the $\N=3$ MI-poset specified by Eq.~\eqref{eq:ds-not-p-ac}, with the corresponding $\bset$ classification of \Cref{defi:special-sets}.
    Each $\bset$ is represented by a colored box, according to the color-coding established in Fig.~\ref{fig:classification-ds}. 
    As discussed below, the gray lines indicate the $\lat{$\beta$}^3$ lattice structure  (more precisely its Hasse diagram), from which we have omitted $\bs{0},\bs{1},\bs{2},\bs{3}$ and the minimum $\bs{\varnothing}$. 
    The indices $0$ and $+$ on each term are used to specify whether the given MI instance is in $\ds$ or $\cds$, respectively, and should not be interpreted as the sign of the value of the corresponding MI instance (particularly since, as explained in the main text, this down-set is not a PMI of any entropy vector). The key aspect of this example is that $\bs{012}$ is both essential and partial (green), since it contains elements of both $\ac$ and $\ds$.
    }
    \label{fig:D-not-P}
\end{figure}

As for the set of instances of the mutual information $\mgs_{\N}$, it will be convenient to introduce a partial order in $\mathfrak{B}_{\N}$. Since any $\bset$ is labeled by a subsystem $\uJ$, this partial order is naturally given by inclusion\footnote{\,Although we use the same symbol $\preceq$ for the partial order among MI instances as well as among $\bsets$, it should be clear from context which one we mean.}
\begin{equation}
    \bs{\uJ}\preceq\bs{\uK}\quad \iff \quad \uJ\subseteq\uK.
\end{equation}
Notice that given two $\bsets$ $\bs{\uJ}\prec\bs{\uJ'}$ and two MI instances $\mi\in\bs{\uJ}$ and $\mi'\in\bs{\uJ'}$, it is not generally the case that $\mi\prec\mi'$, since $\mi$ and $\mi'$ can be incomparable. It is however always true that 
\begin{equation}
    \bs{\uJ}\prec\bs{\uJ'}\quad \implies\quad \mi\nsucceq\mi'.
\end{equation}

To realize a nicer mathematical structure, the set $\mathfrak{B}_{\N}$ of all $\bsets$ at given $\N$ can naturally be extended by formally adding single-party $\beta(\ell)$ and no-party $\beta(\varnothing)$ $\bsets$.  
We can interpret these additional elements as corresponding to MI instances where at least one argument is the empty set, and the other contains the single party $\ell$; explicitly
\begin{equation}
\label{eq:trivial_bsets}
    \bs{\ell}:=\{\mi(\ell:\varnothing)\}\qquad \text{and} \qquad \bs{\varnothing}:=\{\mi(\varnothing:\varnothing)\}.
\end{equation}
With this extension, and the partial order we just introduced for $\mathfrak{B}_{\N}$, the poset of all $\bsets$ at fixed $\N$ is isomorphic to the power-set lattice of the set of parties $\nsp$, and we will denote it by $\lat{$\beta$}^{\N}$. In this lattice, the fundamental operations of meet and join are then simply
\begin{align}
\begin{split}
    & \bs{\uJ}\wedge\bs{\uK}=\bs{\uJ\cap\uK} \\
    & \bs{\uJ}\vee\bs{\uK}=\bs{\uJ\cup\uK}. \\
\end{split}
\end{align}
Since $\ent_{\varnothing}=0$, the new $\bsets$ in Eq.~\eqref{eq:trivial_bsets} (and their MI instances) can be viewed as vanishing identically for any down-set $\ds$. For this reason, and to simplify the presentation, we will always implicitly omit them in what follows, but still refer to $\mathfrak{B}_\N$ as a lattice.

Already at this stage, it should be clear that the lattice of $\bsets$ is a much simpler and nicer structure than the MI-poset. The cardinality of the MI-poset is the Stirling number of the second kind $\genfrac\{\}{0pt}{1}{\N+2}{3}$ \cite{Hubeny:2018ijt}, whereas for the lattice of $\bsets$ it is $2^{\N+1}$, which for large $\N$ gives an exponentially small ratio,
\begin{equation}
    \frac{2^{\N+1}}{\genfrac\{\}{0pt}{1}{\N+2}{3}} \sim \left(\frac{2}{3}\right)^{\N+1}.
\end{equation}
Furthermore, from an algebraic perspective, for $\N\geq 3$ the MI-poset is not even a lattice \cite{He:2022bmi}, while $\lat{$\beta$}^{\N}$ is a distributive\footnote{\,See for example \cite{birkhoff1967lattice} for the definition and its main properties.} lattice for any $\N$. These properties of $\lat{$\beta$}^{\N}$, however, are only useful if it is possible to describe a KC-PMI purely in terms of $\bsets$. Showing that this is indeed the case is the main focus of this subsection, while in the next, we will develop the necessary technology to further explore the relations among the various sets in \Cref{defi:special-sets}.

\begin{figure}[tb]    
    \centering
\begin{tikzpicture}

    \draw[draw=none,rounded corners,fill=PminusEcol] (1.3,2.6) rectangle (12.1,3.4);

    \draw[draw=none,rounded corners,fill=Mcol] (-0.55,0.6) rectangle (2.9,1.4);
    \draw[draw=none,rounded corners,fill=Mcol] (3.05,0.6) rectangle (6.5,1.4);
    \draw[draw=none,rounded corners,fill=Mcol] (6.65,0.6) rectangle (10.1,1.4);
    \draw[draw=none,rounded corners,fill=Mcol] (10.25,0.6) rectangle (13.7,1.4);

    \draw[draw=none,rounded corners,fill=Vcol] (1.2,-1.4) rectangle (2.3,-0.6);
    \draw[draw=none,rounded corners,fill=Vcol] (3.2,-1.4) rectangle (4.3,-0.6);
    \draw[draw=none,rounded corners,fill=Vcol] (5.2,-1.4) rectangle (6.3,-0.6);
    \draw[draw=none,rounded corners,fill=Vcol] (7.2,-1.4) rectangle (8.3,-0.6);
    \draw[draw=none,rounded corners,fill=Vcol] (9.2,-1.4) rectangle (10.3,-0.6);
    \draw[draw=none,rounded corners,fill=Vcol] (11.2,-1.4) rectangle (12.3,-0.6);

    \node () at (2.2,3) {{\scriptsize $\mi(0\!:\!123)_{_{\!+}}$}};
     \node () at (3.7,3) {{\scriptsize $\mi(1\!:\!023)_{_{\!+}}$}};
     \node () at (5.2,3) {{\scriptsize $\mi(2\!:\!013)_{_{\!+}}$}};
     \node () at (6.7,3) {{\scriptsize $\mi(3\!:\!012)_{_{\!+}}$}};
     \node () at (8.2,3) {{\scriptsize $\mi(01\!:\!23)_{_{\!+}}$}};
     \node () at (9.7,3) {{\scriptsize $\mi(02\!:\!13)_{_{\!+}}$}};
     \node () at (11.2,3) {{\scriptsize $\mi(03\!:\!12)_{_{\!+}}$}};
     
    \node () at (0.1,1) {{\scriptsize $\mi(01\!:\!2)_{_{\!+}}$}};
    \node () at (1.2,1) {{\scriptsize $\mi(02\!:\!1)_{_{\!+}}$}};
    \node () at (2.3,1) {{\scriptsize $\mi(12\!:\!0)_{_{\!+}}$}};
    \node () at (3.7,1) {{\scriptsize $\mi(01\!:\!3)_{_{\!+}}$}};
    \node () at (4.8,1) {{\scriptsize $\mi(03\!:\!1)_{_{\!+}}$}};
    \node () at (5.9,1) {{\scriptsize $\mi(13\!:\!0)_{_{\!+}}$}};
    \node () at (7.3,1) {{\scriptsize $\mi(02\!:\!3)_{_{\!+}}$}};
    \node () at (8.4,1) {{\scriptsize $\mi(03\!:\!2)_{_{\!+}}$}};
    \node () at (9.5,1) {{\scriptsize $\mi(23\!:\!0)_{_{\!+}}$}};
    \node () at (10.9,1) {{\scriptsize $\mi(12\!:\!3)_{_{\!+}}$}};
    \node () at (12,1) {{\scriptsize $\mi(13\!:\!2)_{_{\!+}}$}};
    \node () at (13.1,1) {{\scriptsize $\mi(23\!:\!1)_{_{\!+}}$}};

    \node () at (1.8,-1) {{\scriptsize $\mi(0\!:\!1)_{_0}$}};
     \node () at (3.8,-1) {{\scriptsize $\mi(0\!:\!2)_{_0}$}};
     \node () at (5.8,-1) {{\scriptsize $\mi(1\!:\!2)_{_0}$}};
     \node () at (7.8,-1) {{\scriptsize $\mi(0\!:\!3)_{_0}$}};
     \node () at (9.8,-1) {{\scriptsize $\mi(1\!:\!3)_{_0}$}};
     \node () at (11.8,-1) {{\scriptsize $\mi(2\!:\!3)_{_0}$}};

    \node (0123) at (6.7,3.7) {{\scriptsize $\bs{0123}$}};
    
    \node (012) at (1.2,1.7) {{\scriptsize $\bs{012}$}};
    \node (013) at (4.8,1.7) {{\scriptsize $\bs{013}$}};
    \node (023) at (8.4,1.7) {{\scriptsize $\bs{023}$}};
    \node (123) at (12,1.7) {{\scriptsize $\bs{123}$}};

    \node (01) at (1.8,-0.3) {{\scriptsize $\bs{01}$}};
    \node (02) at (3.8,-0.3) {{\scriptsize $\bs{02}$}};
    \node (12) at (5.8,-0.3) {{\scriptsize $\bs{12}$}};
    \node (03) at (7.8,-0.3) {{\scriptsize $\bs{03}$}};
    \node (13) at (9.8,-0.3) {{\scriptsize $\bs{13}$}};
    \node (23) at (11.8,-0.3) {{\scriptsize $\bs{23}$}};

    \draw[-,gray,very thin] (6.7,2.6) -- (012.north);
    \draw[-,gray,very thin] (6.7,2.6) -- (013.north);
    \draw[-,gray,very thin] (6.7,2.6) -- (023.north);
    \draw[-,gray,very thin] (6.7,2.6) -- (123.north);

    \draw[-,gray,very thin] (1.175,0.6) -- (01.north);
    \draw[-,gray,very thin] (1.175,0.6) -- (02.north);
    \draw[-,gray,very thin] (1.175,0.6) -- (12.north);
    \draw[-,gray,very thin] (4.775,0.6) -- (01.north);
    \draw[-,gray,very thin] (4.775,0.6) -- (03.north);
    \draw[-,gray,very thin] (4.775,0.6) -- (13.north);
    \draw[-,gray,very thin] (8.375,0.6) -- (02.north);
    \draw[-,gray,very thin] (8.375,0.6) -- (03.north);
    \draw[-,gray,very thin] (8.375,0.6) -- (23.north);
    \draw[-,gray,very thin] (11.975,0.6) -- (12.north);
    \draw[-,gray,very thin] (11.975,0.6) -- (13.north);
    \draw[-,gray,very thin] (11.975,0.6) -- (23.north);

    \draw[-,very thick] (-0.2,0.75) -- (0.3,0.75);
    \draw[-,very thick] (0.9,0.75) -- (1.4,0.75);
    \draw[-,very thick] (2,0.75) -- (2.5,0.75);
    \draw[-,very thick] (3.4,0.75) -- (3.9,0.75);
    \draw[-,very thick] (4.5,0.75) -- (5,0.75);
    \draw[-,very thick] (5.6,0.75) -- (6.1,0.75);
    \draw[-,very thick] (7,0.75) -- (7.5,0.75);
    \draw[-,very thick] (8.1,0.75) -- (8.6,0.75);
    \draw[-,very thick] (9.2,0.75) -- (9.7,0.75);
    \draw[-,very thick] (10.6,0.75) -- (11.1,0.75);
    \draw[-,very thick] (11.7,0.75) -- (12.2,0.75);
    \draw[-,very thick] (12.8,0.75) -- (13.3,0.75);
     
    \end{tikzpicture}
    \caption{
    A simple example of an $\N=3$ KC-PMI $\pmi$, specified by Eq.~\eqref{eq:ex1}, showing the reconstruction of $\pmi$ from the set of its positive $\bsets$ (cf., Eq.~\eqref{eq:recovery2}). Notice that unlike the (non-PMI) example in Fig.~\ref{fig:D-not-P}, here all essential $\bsets$ are positive, and furthermore the symbols 0 and + now do correspond to the actual sign of the various MI instances for an arbitrary entropy vector in the interior of the face $\face_\pmi$.
    }
    \label{fig:PT3}
\end{figure}

We begin with an example of a KC-PMI $\pmi$ where the sets in \Cref{defi:special-sets} take a particularly simple form, and it is immediate to see how $\pmi$ can be ``reconstructed'' if one is only given the set of essential $\bsets$ $\ess(\pmi)$, or the set of positive ones $\pos(\pmi)$. For $\N=3$, consider the KC-PMI $\pmi$ corresponding to the set of all bottom elements of the MI-poset, i.e.,
\begin{equation}
\label{eq:ex1}
    \pmi = \{\mi(\ell:\ell'),\, \ \forall\ell,\ell'\in\llbracket 3 \rrbracket\}.
\end{equation}
Geometrically $\face_{\pmi}$ is in fact an extreme ray of the SAC$_3$, generated by the entropy vector
\begin{equation}
    \vec{\ent}=(1,1,1,2,2,2,1),
\end{equation}
which is realized by the (pure) 4-party perfect state.\footnote{\,This is the absolutely maximally entangled state on four parties (cf., \cite{Linden:2013kal,helwig2013}), and its entropy vector is realized by a star graph with four leaves and unit edge weights \cite{Bao:2015bfa}. (A leaf in a graph is a vertex with degree one.)} The antichain $\ac$ is
\begin{equation}
    \ac = \{\mi(\uI:\uK),\, \ |\uI|=2\;\; \text{and}\;\; |\uK|=1\},
\end{equation}
and the classification of the $\bsets$ is illustrated in Fig.~\ref{fig:PT3}. 

To see how the reconstruction works, it is convenient to temporarily introduce the following map $\Uparrow$ from the power-set of the lattice of $\bsets$ to the power-set of the MI-poset, i.e., from a collection $\xx$ of $\bsets$ to a collection of instances of the mutual information
\begin{equation}
\label{eq:recovery1}
    \Uparrow\! \xx\; = \bigcup_{\bs{\uI}\,\in\,\xx} \uparrow\! \{\mi(\uJ:\uK) \in \bs{\uI}\}.
\end{equation}
In words, $\Uparrow\! \xx$ is the set obtained by taking the union of all up-sets of all instances of the mutual information in all $\bs{\uI}$ in the collection $\xx$.\footnote{\,Note that this subset $\Uparrow\! \xx \subseteq \mgs_{\N}$ is in general a (proper) subset of the set of the instances of mutual information in the $\bsets$ contained in $\uparrow\!\bs{\uI} \subseteq \mathfrak{B}_{\N}$ for all $\bs{\uI}$ in the collection $\xx$, i.e.,
    $$
        \Uparrow\! \xx\; \subseteq 
        \bigcup_{
            \genfrac {}{}{0pt}{2}{\bs{\uK} \in \uparrow\bs{\uI}}{\bs{\uI}\,\in\,\xx}
            }
        \{\mi \in \bs{\uK}\}.
    $$
}$^{,}$\footnote{\,Analogously to $\uparrow\!\bs{\uI}$ and $\uparrow\!\mi$, we also write $\Uparrow\!\bs{\uI}$ for $\Uparrow\!\{\bs{\uI}\}$.} With this definition, it is then easy to see that in the present example we have
\begin{equation}
\label{eq:recovery2}
    \cpmi\, = \;\; \Uparrow \ess (\pmi)\; = \;\; \Uparrow \pos (\pmi).
\end{equation}

For an arbitrary down-set one might naturally suspect that such a simple formula would not hold. In fact, it is straightforward to see that the first equality in Eq.~\eqref{eq:recovery2} is necessarily false for any $\ds$ with a $\bset$ which is both essential and partial, like the example in Fig.~\ref{fig:D-not-P}, where $\cds$ is generated by the antichain in Eq.~\eqref{eq:ds-not-p-ac}. The reason is that, by definition, $\Uparrow \ess (\pmi)$ constructs an up-set $\cds$ that includes all elements in each essential $\bset$, but if a $\bset$ is essential and partial, it also contains elements of $\ds$. To see that the second expression in Eq.~\eqref{eq:recovery2} also fails for this example, notice that since $\bs{012}$ is partial and all $\bsets$ below it are vanishing, none of its elements belong to the up-set of the MI-poset obtained from the last term in Eq.~\eqref{eq:recovery2}. 
Therefore, since $\cds$ also contains $\mi(01:2)$ and $\mi(12:0)$ (which are both elements of $\ac$), the reconstruction fails. While this example shows that in general Eq.~\eqref{eq:recovery2} does not reconstruct an up-set from the collection of essential or positive $\bsets$, the key point is that $\ds$ in this case is \textit{not} a PMI. The reason is that if we impose that all MI instances in $\ds$ vanish, then by linear dependence there are MI instances in $\cds$ that also have to vanish. Specifically, we have the following contradiction
\begin{align}
    0\neq\mi(01:2) & =\ent_{01}+\ent_{2}-\ent_{012}\nonumber\\
    & = \ent_0+\ent_1+\ent_2-\ent_{02}-\ent_1\nonumber\\
    & = \mi(0:2)=0,
\end{align}
where in the second line we used the fact that $\mi(0:1)=0$ and 
$\mi(02:1)=0$. 
This observation generalizes to arbitrary KC-PMIs:  

\begin{thm}
\label{thm:positiveness}
    For any $\N$ and \emph{KC-PMI} $\pmi$, every essential $\bset$ is positive, i.e., 
    \begin{equation}
    \label{eq:recoverythm}
        \ess(\pmi)\subseteq\pos(\pmi).
    \end{equation}
\end{thm}
\begin{proof}
    We proceed by contradiction. For a given $\N$ and KC-PMI $\pmi$, consider some $\bs{\uI}\in\ess(\pmi)$.\footnote{\,Notice that $\ess=\varnothing$ if and only if $\ac=\varnothing$, which corresponds to the trivial (0-dimensional) KC-PMI where $\van=\mathfrak{B}_{\N}$. In this case Eq.~\eqref{eq:recoverythm} holds trivially, as does Eq.~\eqref{eq:recovery2} since $\Uparrow\varnothing=\varnothing$; however we will mostly ignore this case in what follows.} By definition this means that there is an MI instance $\mi(\uJ:\uK)\in\bs{\uI}$ which belongs to $\ac$. Suppose now that there is an MI instance $\mi(\uJ':\uK')\in\bs{\uI}$ such that $\mi(\uJ':\uK')=0$. Since $\uJ\cup\uK=\uJ'\cup\uK'=\uI$, we can rewrite $\mi(\uJ:\uK)$ as follows
    \begin{align}
    \label{eq:vanishing_assumption}
        \mi(\uJ:\uK) & =\mi\left((\uJ\cap\uJ')\cup(\uJ\cap\uK'):(\uK\cap\uJ')\cup(\uK\cap\uK')\right)\nonumber\\
        & = \ent\left((\uJ\cap\uJ')\cup(\uJ\cap\uK')\right)+\ent\left((\uK\cap\uJ')\cup(\uK\cap\uK')\right)-\ent(\uI)\nonumber\\
        & =\ent(\uJ\cap\uJ')+\ent(\uJ\cap\uK')+\ent(\uK\cap\uJ')+\ent(\uK\cap\uK')-\ent(\uI),
    \end{align}
    where in the last step we used the fact that
    \begin{equation}
        \mi(\uJ\cap\uJ':\uJ\cap\uK')=\mi(\uK\cap\uJ':\uK\cap\uK')=0,
    \end{equation}
    since each of these MI instances is in the principal down-set of $\mi(\uJ':\uK')$, which is assumed to vanish. Furthermore, this also implies
    \begin{align}
    \label{eq:minimalilty_assumption}
        \ent(\uI) & = \ent(\uJ')+\ent(\uK')\nonumber\\
        & = \ent\left((\uJ'\cap\uJ)\cup(\uJ'\cap\uK)\right)+\ent\left((\uK'\cap\uJ)\cup(\uK'\cap\uK)\right)\nonumber\\
        & = \ent(\uJ'\cap\uJ)+\ent(\uJ'\cap\uK)+\ent(\uK'\cap\uJ)+\ent(\uK'\cap\uK),
    \end{align}
    where in the last step we used the fact that since $\mi(\uJ:\uK)\in\ac$, any $\mi(\uJ'':\uK'')$ with $\mi(\uJ'':\uK'')\prec\mi(\uJ:\uK)$ has to vanish. Combining Eq.~\eqref{eq:vanishing_assumption} and Eq.~\eqref{eq:minimalilty_assumption} we get $\mi(\uJ:\uK)=0$, which contradicts the assumption that $\mi(\uJ:\uK)\in\ac$. 
    
    Finally, notice that in general some of the intersections appearing in Eq.~\eqref{eq:vanishing_assumption} and Eq.~\eqref{eq:minimalilty_assumption} can be empty. However we can follow the same argument with the only difference that some terms drop from these equations because they trivially vanish. We leave the details as a simple exercise.
\end{proof}

\begin{figure}[tb]    
    \centering
\begin{tikzpicture}

    \draw[draw=none,rounded corners,fill=Mcol] (1.3,2.6) rectangle (12.1,3.4);

    \draw[draw=none,rounded corners,fill=Vcol] (-0.55,0.6) rectangle (2.9,1.4);
    \draw[draw=none,rounded corners,fill=Vcol] (3.05,0.6) rectangle (6.5,1.4);
    \draw[draw=none,rounded corners,fill=Vcol] (6.65,0.6) rectangle (10.1,1.4);
    \draw[draw=none,rounded corners,fill=Vcol] (10.25,0.6) rectangle (13.7,1.4);

    \draw[draw=none,rounded corners,fill=Vcol] (1.2,-1.4) rectangle (2.3,-0.6);
    \draw[draw=none,rounded corners,fill=Vcol] (3.2,-1.4) rectangle (4.3,-0.6);
    \draw[draw=none,rounded corners,fill=Vcol] (5.2,-1.4) rectangle (6.3,-0.6);
    \draw[draw=none,rounded corners,fill=Vcol] (7.2,-1.4) rectangle (8.3,-0.6);
    \draw[draw=none,rounded corners,fill=Vcol] (9.2,-1.4) rectangle (10.3,-0.6);
    \draw[draw=none,rounded corners,fill=Vcol] (11.2,-1.4) rectangle (12.3,-0.6);

    \node () at (2.2,3) {{\scriptsize $\mi(0\!:\!123)_{_{\!+}}$}};
     \node () at (3.7,3) {{\scriptsize $\mi(1\!:\!023)_{_{\!+}}$}};
     \node () at (5.2,3) {{\scriptsize $\mi(2\!:\!013)_{_{\!+}}$}};
     \node () at (6.7,3) {{\scriptsize $\mi(3\!:\!012)_{_{\!+}}$}};
     \node () at (8.2,3) {{\scriptsize $\mi(01\!:\!23)_{_{\!+}}$}};
     \node () at (9.7,3) {{\scriptsize $\mi(02\!:\!13)_{_{\!+}}$}};
     \node () at (11.2,3) {{\scriptsize $\mi(03\!:\!12)_{_{\!+}}$}};
     
    \node () at (0.1,1) {{\scriptsize $\mi(01\!:\!2)_{_0}$}};
    \node () at (1.2,1) {{\scriptsize $\mi(02\!:\!1)_{_0}$}};
    \node () at (2.3,1) {{\scriptsize $\mi(12\!:\!0)_{_0}$}};
    \node () at (3.7,1) {{\scriptsize $\mi(01\!:\!3)_{_0}$}};
    \node () at (4.8,1) {{\scriptsize $\mi(03\!:\!1)_{_0}$}};
    \node () at (5.9,1) {{\scriptsize $\mi(13\!:\!0)_{_0}$}};
    \node () at (7.3,1) {{\scriptsize $\mi(02\!:\!3)_{_0}$}};
    \node () at (8.4,1) {{\scriptsize $\mi(03\!:\!2)_{_0}$}};
    \node () at (9.5,1) {{\scriptsize $\mi(23\!:\!0)_{_0}$}};
    \node () at (10.9,1) {{\scriptsize $\mi(12\!:\!3)_{_0}$}};
    \node () at (12,1) {{\scriptsize $\mi(13\!:\!2)_{_0}$}};
    \node () at (13.1,1) {{\scriptsize $\mi(23\!:\!1)_{_0}$}};

    \node () at (1.8,-1) {{\scriptsize $\mi(0\!:\!1)_{_0}$}};
     \node () at (3.8,-1) {{\scriptsize $\mi(0\!:\!2)_{_0}$}};
     \node () at (5.8,-1) {{\scriptsize $\mi(1\!:\!2)_{_0}$}};
     \node () at (7.8,-1) {{\scriptsize $\mi(0\!:\!3)_{_0}$}};
     \node () at (9.8,-1) {{\scriptsize $\mi(1\!:\!3)_{_0}$}};
     \node () at (11.8,-1) {{\scriptsize $\mi(2\!:\!3)_{_0}$}};

    \node (0123) at (6.7,3.7) {{\scriptsize $\bs{0123}$}};
    
    \node (012) at (1.2,1.7) {{\scriptsize $\bs{012}$}};
    \node (013) at (4.8,1.7) {{\scriptsize $\bs{013}$}};
    \node (023) at (8.4,1.7) {{\scriptsize $\bs{023}$}};
    \node (123) at (12,1.7) {{\scriptsize $\bs{123}$}};

    \node (01) at (1.8,-0.3) {{\scriptsize $\bs{01}$}};
    \node (02) at (3.8,-0.3) {{\scriptsize $\bs{02}$}};
    \node (12) at (5.8,-0.3) {{\scriptsize $\bs{12}$}};
    \node (03) at (7.8,-0.3) {{\scriptsize $\bs{03}$}};
    \node (13) at (9.8,-0.3) {{\scriptsize $\bs{13}$}};
    \node (23) at (11.8,-0.3) {{\scriptsize $\bs{23}$}};

    \draw[-,gray,very thin] (6.7,2.6) -- (012.north);
    \draw[-,gray,very thin] (6.7,2.6) -- (013.north);
    \draw[-,gray,very thin] (6.7,2.6) -- (023.north);
    \draw[-,gray,very thin] (6.7,2.6) -- (123.north);

    \draw[-,gray,very thin] (1.175,0.6) -- (01.north);
    \draw[-,gray,very thin] (1.175,0.6) -- (02.north);
    \draw[-,gray,very thin] (1.175,0.6) -- (12.north);
    \draw[-,gray,very thin] (4.775,0.6) -- (01.north);
    \draw[-,gray,very thin] (4.775,0.6) -- (03.north);
    \draw[-,gray,very thin] (4.775,0.6) -- (13.north);
    \draw[-,gray,very thin] (8.375,0.6) -- (02.north);
    \draw[-,gray,very thin] (8.375,0.6) -- (03.north);
    \draw[-,gray,very thin] (8.375,0.6) -- (23.north);
    \draw[-,gray,very thin] (11.975,0.6) -- (12.north);
    \draw[-,gray,very thin] (11.975,0.6) -- (13.north);
    \draw[-,gray,very thin] (11.975,0.6) -- (23.north);

    \draw[-,very thick] (1.9,2.75) -- (2.4,2.75);
    \draw[-,very thick] (3.4,2.75) -- (3.9,2.75);
    \draw[-,very thick] (4.9,2.75) -- (5.4,2.75);
    \draw[-,very thick] (6.4,2.75) -- (6.9,2.75);
    \draw[-,very thick] (7.9,2.75) -- (8.4,2.75);
    \draw[-,very thick] (9.4,2.75) -- (9.9,2.75);
    \draw[-,very thick] (10.9,2.75) -- (11.4,2.75);
     
    \end{tikzpicture}
    \caption{An example of an $\N=3$ down-set $\ds$ which is not a PMI, and for which the reconstruction in Eq.~\eqref{eq:recovery2} gives the correct answer. Notice that $\mi(01:2)+\mi(2:3)=2\ent_{2}=\mi(2:013)$, and that if we set $\mi(01:2)=\mi(2:3)=0$ we necessarily have $\mi(2:013)=0$, which instead here is chosen to be positive. This implies that $\ds$ is not a PMI. 
    }
    \label{fig:EsubsetP-not-suff}
\end{figure}

\Cref{thm:positiveness} immediately implies that for KC-PMIs the classification of $\bsets$ given in Fig.~\ref{fig:classification-ds} simplifies, and instead takes the form shown in Fig.~\ref{fig:classification-pmi}. Furthermore, it implies that any KC-PMI can be reconstructed given only its set of positive $\bsets$, or even just the essential ones: 

\begin{cor}
\label{cor:ds_reconstruction}
    For any $\N$ and \emph{KC-PMI} $\pmi$, the relation Eq.~\eqref{eq:recovery2} holds, i.e., $\cpmi$ is the up-set of the collection of all instances of the mutual information in all $\bsets$ in $\ess(\pmi)$, or equivalently, in $\pos(\pmi)$. 
\end{cor}
\begin{proof}
    Let us first assume that for a fixed KC-PMI $\pmi$ we are given $\ess(\pmi)$. By definition, any element of $\ac$ is contained in some $\bs{\uJ}\in\ess(\pmi)$ for some $\uJ$, and since $\uparrow\!\ac=\cpmi$, we have $\Uparrow\!\ess(\pmi)\supseteq\cpmi$. By \Cref{thm:positiveness}, for every $\bs{\uJ}\in\ess(\pmi)$ and every element $\mi\in\bs{\uJ}$, we have $\mi\in\cpmi$. Since $\cpmi$ is an up-set, $\cpmi\supseteq\;\uparrow\!\mi$, implying $\cpmi\supseteq\,\Uparrow\!\ess(\pmi)$ and therefore $\cpmi=\;\Uparrow\!\ess(\pmi)$. By a similar argument, the same result also holds if we replace $\ess(\pmi)$ with $\pos(\pmi)$, since by \Cref{thm:positiveness}, $\pos(\pmi)\supseteq\ess(\pmi)$.
\end{proof}

Given this result, one may wonder if the converse could also hold. Namely, for any down-set $\ds$ in which every essential $\bset$ is positive, is $\ds$ necessarily a KC-PMI?
It is easy to see that this is not the case: We provide an example in Fig.~\ref{fig:EsubsetP-not-suff}, and we will come back to this point in the discussion \S\ref{sec:discussion}. Finally, it is obvious that if a PMI $\pmi$ is not a KC-PMI, Eq.\eqref{eq:recovery2} does not reconstruct $\cpmi$ correctly, since the result of Eq.\eqref{eq:recovery2} is always an up-set, while if $\pmi$ is not a down-set, $\cpmi$ is not an up-set. However, this issue is immaterial given that as explained in \S\ref{subsec:KC-review}, any PMI which does not satisfy KC is incompatible with SSA.

To conclude this subsection, let us briefly comment on the physical interpretation of these results in the context of holography, where they can be understood more intuitively. In this context, if the entanglement wedge of a subsystem $\uI$ is connected (in the sense specified in e.g.\ \cite{Hernandez-Cuenca:2023iqh} and in absence of degeneracies), it implies positivity of the mutual information for \emph{any} bipartition of $\uI$ (since all bipartitions invoke the same HRT surface for $\uI$, which is obviously distinct from the union of the HRT surfaces for the components of any bipartition). By definition, this is equivalent to the statement that $\bs{X}$ is positive.\footnote{\,In the language of holographic graph models, this situation essentially corresponds to the special case of \Cref{cor:ew-decomposition} when $m=1$; we will come back to this point in \S\ref{sec:necessary} (cf., \Cref{cor:ew_connectivity-v2}).} Furthermore, by entanglement wedge nesting \cite{Headrick:2013zda,Wall_2014,Headrick_2014}, it is clear that if two distinct connected entanglement wedges, say those for $\uJ$ and $\uK$, intersect, they automatically induce a connected entanglement wedge for the joint system $\uJ \cup \uK$, which then does not constitute an independent piece of data (an example of this situation is illustrated in Fig.~\ref{fig:ew-crossing-hol}). In other words, it is just the connected entanglement wedges which do \emph{not} contain any smaller intersecting entanglement wedges immediately inside, which are the key (and sufficient) building blocks to determine the full PMI,\footnote{\,This recoverability of the full PMI from these building blocks follows from the fact that, as in \Cref{cor:ew-decomposition} for graph models, an MI instance $\mi(\uJ:\uK)$ vanishes if and only if there is no subsystem $\uI$, with $\uI\cap\uJ\neq 0$ and $\uI\cap\uK\neq 0$, whose minimal entanglement wedge is connected.} and these correspond precisely to the essential $\bsets$.\footnote{\,Since the full PMI can be reconstructed from these ``fundamental'' connected entanglement wedges, and the reconstruction is only possible if they capture all elements of $\ac$, it is clear that in the holographic setting, all essential $\bsets$ are positive and do in fact correspond to these building blocks.} 

Of course, the central point of our framework is that the possibility of reconstructing a KC-PMI from a certain set of positive $\bsets$ (namely the essential ones) transcends holography, so \Cref{thm:positiveness} generalizes well beyond this intuitive anchor point. Indeed, we will see in the next subsection that this holographic characterization of the fundamental building blocks as connected entanglement wedges which do not contain any smaller intersecting entanglement wedges, and the way these determine the vanishing or essential positive instances of the mutual information, generalize to arbitrary KC-PMIs (cf., \Cref{lem:max-structure}, \Cref{thm:max-structure} and \Cref{thm:partition}).\footnote{\,Of course these statements will also have a corresponding formulation in terms of graph models, starting from the crucial relation between the sign of MI instances and the connected components of the subgraphs induced by minimal min-cuts highlighted in \Cref{cor:ew-decomposition}. We will further comment on this correspondence in \S\ref{subsec:hypergraph-rep}, as the correlation hypergraph will make it even more transparent.}

\begin{figure}
    \centering
    \begin{tikzpicture}
    \centering

    \node[anchor=south west,inner sep=0] at (0,0) {\includegraphics[width=0.18\linewidth]{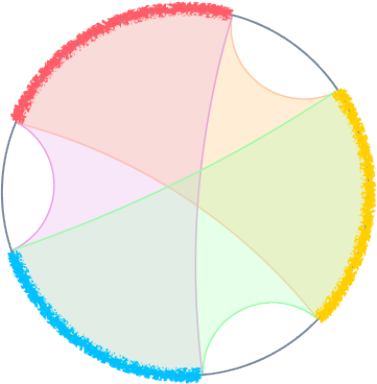}};

    \node[anchor=south west,inner sep=0] at (4,0) {\includegraphics[width=0.18\linewidth]{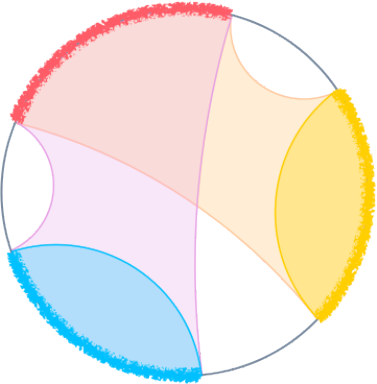}};

    \node[anchor=south west,inner sep=0] at (8,0) {\includegraphics[width=0.18\linewidth]{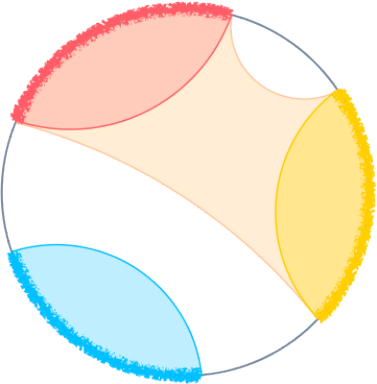}};

    \node[anchor=south west,inner sep=0] at (12,0) {\includegraphics[width=0.18\linewidth]{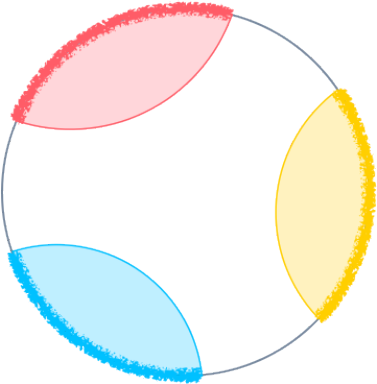}};

    \node[] () at (0.5,0) {{\scriptsize $3$}};
    \node[] () at (4.5,0) {{\scriptsize $3$}};
    \node[] () at (8.5,0) {{\scriptsize $3$}};
    \node[] () at (12.5,0) {{\scriptsize $3$}};

    \node[] () at (0.6,2.9) {{\scriptsize $1$}};
    \node[] () at (4.6,2.9) {{\scriptsize $1$}};
    \node[] () at (8.6,2.9) {{\scriptsize $1$}};
    \node[] () at (12.6,2.9) {{\scriptsize $1$}};

    \node[] () at (2.9,1.4) {{\scriptsize $2$}};
    \node[] () at (6.9,1.4) {{\scriptsize $2$}};
    \node[] () at (10.9,1.4) {{\scriptsize $2$}};
    \node[] () at (14.9,1.4) {{\scriptsize $2$}};

    \node[align=center] () at (1.5,-1.6) {{\small $\bs{123}\in\pnote$} \\{\small $\bs{12}\in\pos$} \\ {\small $\bs{13}\in\pos$} \\ {\small $\bs{23}\in\pos$}};

    \node[align=center] () at (5.5,-1.6) {{\small $\bs{123}\in\pnote$} \\{\small $\bs{12}\in\pos$} \\ {\small $\bs{13}\in\pos$} \\ {\small $\bs{23}\in\van$}};

    \node[align=center] () at (9.5,-1.6) {{\small $\bs{123}\in\ncess$} \\{\small $\bs{12}\in\pos$} \\ {\small $\bs{13}\in\van$} \\ {\small $\bs{23}\in\van$}};

    \node[align=center] () at (13.5,-1.6) {{\small $\bs{123}\in\cess$} \\{\small $\bs{12}\in\van$} \\ {\small $\bs{13}\in\van$} \\ {\small $\bs{23}\in\van$}};

    \end{tikzpicture}
    \caption{An example of a holographic set-up with a choice of three boundary regions such that the entanglement wedge of $123$ is connected (not drawn). The figures show several possible cases for the entanglement wedges of the various pairs of boundary regions and the corresponding classification of $\bsets$. Notice in particular that while $\bs{123}$ is always positive (since we are assuming that its entanglement wedge is connected), its level of essentiality depends on whether or not there are internal intersecting entanglement wedges. 
    }
    \label{fig:ew-crossing-hol}
\end{figure}

\subsection{Characterization of KC-PMIs in terms of \texorpdfstring{$\bsets$}{beta-sets}}
\label{subsec:KC-PMIs-BS-language}

\Cref{thm:positiveness}, and in particular \Cref{cor:ds_reconstruction}, are at the core of everything we do in this work, and we will now use them to prove a series of corollaries with several practical implications. Since many results will be phrased in terms of ordered sets, we will very often use the operations $\uparrow$ and $\downarrow$ defined in Eq.~\eqref{eq:order-sets}. We stress that these operations always take value in the same space as their argument, unlike $\Uparrow$ defined in Eq.~\eqref{eq:recovery1}, whose only purpose was to formulate Eq.~\eqref{eq:recovery2} and \Cref{cor:ds_reconstruction}, and which consequently will not appear again. In fact all the subsequent use of $\uparrow$ and $\downarrow$ will refer to operations which take place within $\mathfrak{B}_{\N}$ rather than within $\mgs_{\N}$. Instead, we will mostly utilize $\prec$ to indicate order relations within $\mgs_{\N}$.

From \Cref{cor:ds_reconstruction}, it is clear that if we are not given a KC-PMI $\pmi$ explicitly, but only its set $\ess(\pmi)$ of essential $\bsets$, or alternatively that of positive ones ($\pos(\pmi)$), it is possible to derive all the sets in \Cref{defi:special-sets}, as well as determine which MI instances vanish in a partial $\bset$ and which ones belong to $\ac$ in an essential but not completely essential one. Indeed, we can first use Eq.~\eqref{eq:recovery2} to reconstruct $\cpmi$, find its minimal elements to obtain $\ac$, take the complement of $\cpmi$ to obtain $\pmi$, and use this data to characterize any $\bset$ and its MI instances. The goal of this subsection will be to show that this characterization is possible by working entirely with $\bsets$, and without ever explicitly reconstructing $\cpmi$ and $\pmi$. The plan is to first focus on the direct reconstruction of all the $\bsets$ in \Cref{defi:special-sets}, which will be completed by \Cref{cor:P-from-E}, and then proceed with the characterization of the MI instances in each $\bset$ (cf., \Cref{thm:partition}). We begin with another useful corollary.\footnote{\,To simplify notation, from now on we will usually keep implicit the dependence on $\pmi$ of the various sets in \Cref{defi:special-sets} and subsequently.}
\begin{cor}
\label{cor:order_relations}
    For any $\N$ and \emph{KC-PMI} $\pmi$:
    \begin{enumerate}[label={\emph{\footnotesize \roman*)}}]
        \item $\pos\cup\parv$ is an up-set, and for any $\bs{\uJ}\in(\pos\cup\parv)\setminus\cess$ there is some $\bs{\uK}\prec\bs{\uJ}$ such that $\bs{\uK}\in\cess$,
        \item $\cess$ is an antichain, and for any $\bs{\uJ}\in\cess$ and $\bs{\uK}\prec\bs{\uJ}$, $\bs{\uK}\in\van$,
        \item $\van$ is the down-set complementary to $\pos\cup\parv$, i.e., $\van=\;\downarrow \van = \comp{(\pos\cup\parv)}$.
    \end{enumerate}
\end{cor}
\begin{proof}
    i) For any $\bs{\uJ}\in\pos\cup\parv$ there is (by definition) at least one MI instance $\mi(\uJ_1:\uJ_2)\in\bs{\uJ}$ such that $\mi(\uJ_1:\uJ_2)\in\cpmi$. Consider some $\bs{\uK}\succ\bs{\uJ}$ (assuming $\uJ\neq\nsp$, otherwise the argument is trivial) and the MI instance $\mi(\uJ_1:\uJ_2\uI)$, where $\uI=\uK\setminus\uJ$. Since $\mi(\uJ_1:\uJ_2\uI)\succ\mi(\uJ_1:\uJ_2)$, and $\cpmi$ is an up-set, it follows that $\mi(\uJ_1:\uJ_2\uI)\in\cpmi$, and since $\mi(\uJ_1:\uJ_2\uI)\in\bs{\uK}$, we have $\bs{\uK}\in\pos\cup\parv$. This shows that $\pos\cup\parv$ is an up-set.
    
    Consider now some $\bs{\uJ}\in(\pos\cup\parv)\setminus\cess$. By \Cref{defi:special-sets} there is some $\mi\in\bs{\uJ}$ such that $\mi\in\cpmi$ while $\mi\notin\ac$, and there exists therefore some $\mi'\prec\mi$ such that $\mi'\in\ac$. Denoting by $\bs{\uK}$ the $\bset$ that contains $\mi'$, it follows from \Cref{thm:positiveness} that $\bs{\uK}\in\pos$. If $\bs{\uK}\in\cess$ the statement is proven, otherwise we can repeat the same argument. Since the lattice $\lat{$\beta$}^{\N}$ is finite, and at each iteration we obtain a $\bs{\uK'}\in\pos$, the procedure can stop only when we obtain some $\bs{\uK'}\in\cess$ with $\bs{\uK'}\prec\bs{\uJ}$.

    ii) Consider some $\bs{\uJ}\in\cess$ and $\mi\in\bs{\uJ}$. By the definition of $\cess$ it follows that for any $\mi'\prec\mi$ we have $\mi'\in\pmi$, and therefore that for any $\bs{\uK}\prec\bs{\uJ}$, $\bs{\uK}\in\van$. This shows that all elements of $\cess$ are incomparable, and $\cess$ is an antichain.

    iii) Follows immediately from (i) and Eq.~\eqref{eq:trivial_relations}, since in any poset the complement of any up-set is a down-set.
\end{proof}

Using this result, we can immediately derive a few additional relations among the sets in \Cref{defi:special-sets}.  
\begin{cor}
\label{cor:more_reconstruction}
    For any $\N$ and \emph{KC-PMI} $\pmi$:
    \begin{enumerate}[parsep=0pt,label={\emph{\footnotesize \roman*)}}]
        \item $\cess=\text{\emph{Min}}(\ess)=\text{\emph{Min}}(\pos)$,
        \item $\uparrow\cess\,=\;\uparrow\!\ess\,=\;\uparrow\!\pos\,=\,\pos\cup\parv$,
        \item $\parv=\; \uparrow\!\pos\setminus\pos$.
        \item $\van=\comp{(\uparrow\!\pos)}$,
    \end{enumerate}
\end{cor}
\begin{proof}
Note that by \Cref{defi:special-sets} and \Cref{thm:positiveness}, $\cess\subseteq\ess\subseteq\pos$.

    i) From (i) in \Cref{cor:order_relations} we have $\text{Min}(\pos)\subseteq\cess$ and hence $\text{Min}(\ess)\subseteq\cess$. Either of these inclusions is strict only if there are $\bs{\uJ},\bs{\uK}\in\cess$ with $\bs{\uJ}\prec\bs{\uK}$. This however contradicts (ii) in \Cref{cor:order_relations}.

    ii)     The first two equalities follow from part (i) by taking the upset of each expression.
    From (i) in \Cref{cor:order_relations} it follows that $\pos\cup\parv\subseteq\,\uparrow\!\cess$. If the inclusion is strict, then there is some $\bs{\uJ}\in\;\uparrow\!\cess$ such that $\bs{\uJ}\notin\pos\cup\parv$, but this implies $\bs{\uJ}\in\van$ by Eq.~\eqref{eq:trivial_relations}, which contradicts (iii) of \Cref{cor:order_relations} (since $\varnothing \ne \;\downarrow \bs{\uJ} \cap \cess \nsubseteq \van$ means that $\van$ cannot be a downset).
    
    iii) This is obvious from (ii).

    iv) This is immediate from (ii) above and (iii) in \Cref{cor:order_relations}.
\end{proof}

Before we proceed further, let us summarize the implications of \Cref{cor:order_relations} and \Cref{cor:more_reconstruction}, and highlight their limitations. \Cref{cor:order_relations} shows that the order structure in $\lat{$\beta$}^{\N}$ gives a clear characterization of the set of positive or partial $\bsets$ ($\pos\cup\parv$) as the complementary up-set to the down-set of vanishing ones ($\van$).  Furthermore, it shows that the set of completely essential $\bsets$ ($\cess$) is an antichain, and \Cref{cor:more_reconstruction} shows that its up-set is precisely  $\pos\cup\parv$. If one is given the set of positive $\bsets$ for a KC-PMI then, \Cref{cor:more_reconstruction} shows how one can directly derive $\cess$, $\parv$ and $\van$. Notice however that it is not yet clear how to extract the subset of essential $\bsets$.  Similarly, if one is only given the set of essential $\bsets$, one can derive $\cess$ and $\van$, but it is not clear how to separate $\pos$ from $\parv$ in $\uparrow\!\ess$ to obtain all the positive ones. The next few results focus on completing these derivations, for which the following lemma will be useful.
\begin{lemma}
\label{lem:join}
    For any \emph{KC-PMI}, and any distinct $\bsets$ $\bs{\uJ_1},\bs{\uJ_2}\in\pos$,
    \begin{equation}
        \uJ_1\cap\uJ_2\neq\varnothing\quad \implies\quad \bs{\uJ_1\cup\uJ_2}\in\pnote.
    \end{equation}
\end{lemma}
\begin{proof}
    Suppose there are $\bs{\uJ_1},\bs{\uJ_2}\in\pos$ such that $\uJ_1\cap\uJ_2\neq\varnothing$, and consider an MI instance $\mi(\uK_1:\uK_2)\in\bs{\uJ_1\cup\uJ_2}$. Notice that since $\uJ_1\cap\uJ_2\neq\varnothing$, for at least one of the $\uJ_i$'s (i.e., $i\in [2]$), we have
    $\uK_1\cap\uJ_i\neq\varnothing$ and $\uK_2\cap\uJ_i\neq\varnothing$. Since $\uK_1\cup\uK_2\supseteq\uJ_i$, we have $\mi(\uK_1\cap\uJ_i:\uK_2\cap\uJ_i)\in\bs{\uJ_i}$ and $\mi(\uK_1\cap\uJ_i:\uK_2\cap\uJ_i)\prec\mi(\uK_1:\uK_2)$. By assumption, $\bs{\uJ_i}\in\pos$ and therefore $\mi(\uK_1\cap\uJ_i:\uK_2\cap\uJ_i)\in\cpmi$, implying that $\mi(\uK_1:\uK_2)\in\cpmi$ but $\mi(\uK_1:\uK_2)\notin\ac$. Since this argument is independent from the choice of $\mi(\uK_1:\uK_2)\in\bs{\uJ_1\cup\uJ_2}$, it follows that $\bs{\uJ_1\cup\uJ_2}\in\pnote$.
\end{proof}

It is also useful to introduce another set of $\bsets$, this time associated to a specific subsystem $\uJ$. For an arbitrary down-set $\ds$ in the MI-poset, and an arbitrary $\bset$ $\bs{\uJ}$, we denote by $\posj(\ds)$ the set of positive $\bsets$ which are strictly less than $\bs{\uJ}$, i.e.,
\begin{equation}
    \posj(\ds)\coloneq\{\bs{\uK}\in\pos(\ds),\; \uK\subset \uJ\}. 
\end{equation}
Viewing this set as a poset, with the induced partial order from $\lat{$\beta$}^{\N}$, we will be mostly interested in the subset $\text{Max}(\posj)$ of its maximal elements which we write as
\begin{equation}
\label{eq:max-structure}
    \text{Max}(\posj)=\left\{\bs{\uK_1},\bs{\uK_2},\ldots,\bs{\uK_n}\right\}.
\end{equation}
The structure of this set is characterized by the following simple result.

\begin{lemma}
\label{lem:max-structure}
    For any \emph{KC-PMI} and any subsystem $\uJ$ such that $n:=|\text{\emph{Max}}(\posj)|\geq 2$, one of the following holds:
    \begin{enumerate}[label={\emph{\footnotesize \roman*)}}]
        \item in Eq.~\eqref{eq:max-structure}, $\uK_i\cap\uK_j=\varnothing$ for all $i,j\in[n]$ with $i\neq j$,
        \item in Eq.~\eqref{eq:max-structure}, $\uK_i\cap\uK_j\neq\varnothing$ and $\uK_i\cup\uK_j=\uJ$ for all $i,j\in[n]$ with $i\neq j$.
    \end{enumerate}
\end{lemma}
\begin{proof}
 For an arbitrary KC-PMI and subsystem $\uJ$ such that $n\geq 2$, consider two distinct elements  $\bs{\uK_i},\bs{\uK_j}\in\text{Max}(\posj)$. If $\uK_i\cap\uK_j\neq\varnothing$, then by \Cref{lem:join} $\bs{\uK_i\cup\uK_j}\in\pos$, and it must be that $\uK_i\cup\uK_j=\uJ$, otherwise neither $\bs{\uK_i}$ nor $\bs{\uK_j}$ is maximal in $\posj$. This proves the lemma for $n=2$. Suppose now that $n\geq 3$, and that there are distinct $\bs{\uK_i},\bs{\uK_j}\in\text{Max}(\posj)$ such that $\uK_i\cap\uK_j\neq\varnothing$. As we have just shown, this implies $\uK_i\cup\uK_j=\uJ$. Therefore, since all elements of $\text{Max}(\posj)$ are incomparable (i.e., any $\uK_k\in\text{Max}(\posj)$ distinct from $\uK_i$ and $\uK_j$ cannot be contained in either $\uK_i$ or $\uK_j$), we necessarily have that for any $\bs{\uK_k}\in\text{Max}(\posj)$, $\uK_i\cap\uK_k\neq\varnothing$ and $\uK_j\cap\uK_k\neq\varnothing$, and by the same argument as above, it must be $\uK_i\cup\uK_k=\uJ$ and $\uK_j\cup\uK_k=\uJ$.  
 Therefore, when $n\geq2$, either there are distinct $\bs{\uK_i},\bs{\uK_j}\in\text{Max}(\posj)$ such that $\uK_i\cap\uK_j\neq\varnothing$, and the structure of $\text{Max}(\posj)$ is the one described in (ii), or there is no such pair, which means (i).
\end{proof}

We can now prove a result that sharpens the distinction between essential $\bsets$ and those which are positive but not essential.  The theorem also highlights a certain structural similarity between the set of elements of $\ac$ in an essential but not completely essential $\bset$, and that of vanishing MI instances in a partial one. An example of this result is shown in Fig.~\ref{fig:N4-example}.

\begin{thm}
\label{thm:max-structure}
    For any \emph{KC-PMI} and any $\bset$ $\bs{\uJ}$:
    \begin{enumerate}[label={\emph{\footnotesize \roman*)}}]
        \item $\bs{\uJ}\in\van\cup\cess$ if and only if $\text{\emph{Max}}(\posj)=\varnothing$,
        \item $\bs{\uJ}\in\ncess\cup\parv$ 
        if and only if either $n=1$,  
        or $n>1$ and $\text{\emph{Max}}(\posj)$ has the structure described in \emph{(i)} of \Cref{lem:max-structure}, 
        \item $\bs{\uJ}\in\pnote$ if and only if $n\geq 2$ and $\text{\emph{Max}}(\posj)$ has the structure described in \emph{(ii)} of \Cref{lem:max-structure}.
    \end{enumerate}
\end{thm}
\begin{proof}
    We first consider the case of some $\bs{\uJ}$ such that $\text{Max}(\posj)=\varnothing$. It follows immediately from (i) of \Cref{cor:order_relations} that $\bs{\uJ}\in\van\cup\cess$, proving the ``if'' direction of (i). The ``only if'' direction is trivial from the definition of $\van$ and $\cess$.

    To prove (ii), suppose that $n=1$, or that $n>1$ and $\text{Max}(\posj)$ has the structure described in (i) of \Cref{lem:max-structure}. Notice that in general $\bigcup_{i=1}^n \uK_i \subseteq \uJ$, not necessarily with equality, and that the inclusion is necessarily strict when $n=1$. It follows from (i) that $\bs{\uJ}\notin\van\cup\cess$, implying  $\bs{\uJ}\in(\pos\setminus\cess)\cup\parv$, and we need to show that if $\bs{\uJ}\in\pos\setminus\cess$ then $\bs{\uJ}\in\ncess$. Consider an MI instance $\mi(\uJ_1:\uJ_2)\in\bs{\uJ}$ such that for each $i\in [n]$ either $\uK_i\subset\uJ_1$ or $\uK_i\subset\uJ_2$. If $\bs{\uJ}\in\pos\setminus\cess$ then $\mi\in\cpmi$ but $\mi\notin\ac$ and there is some $\mi'(\uJ_1':\uJ_2')$ such that $\mi'\prec\mi$ and $\mi'\in\ac$. By \Cref{thm:positiveness}, $\bs{\uJ'_1\cup\uJ'_2}\in\pos$, and by the definition of $\text{Max}(\posj)$ and \Cref{lem:join}, it must be that for each $i\in [n]$, either $\uJ'_1\cup\uJ'_2\subseteq\uK_i$ or $(\uJ'_1\cup\uJ'_2)\cap\uK_i=\varnothing$. By the assumption about $\mi(\uJ_1:\uJ_2)$ however, the only possibility for $\mi'$ to satisfy this condition is that either $\uJ'_1=\varnothing$ or $\uJ'_2=\varnothing$, but in either case $\mi'\notin\ac$, which is a contradiction.
    
    In the ``only if'' direction, suppose that $\bs{\uJ}\in\ncess\cup\parv$
    and notice that from (i) we have $\text{Max}(\posj)\neq\varnothing$, implying $n\geq 1$. Suppose that $n>1$ (otherwise the statement is already proven) and notice that if $\bs{\uJ}\in\parv$, then $\text{Max}(\posj)$ must have the structure described in (i) of \Cref{lem:max-structure}, otherwise \Cref{lem:join} implies $\bs{\uJ}\in\pos$. Suppose then that $\bs{\uJ}\in\ncess$, and that $\text{Max}(\posj)$ has the structure described in (ii) of \Cref{lem:max-structure}. It follows that for every $\mi\in\bs{\uJ}$ there is some $\bs{\uK_i}\in\text{Max}(\posj)$ and some $\mi'\in\bs{\uK_i}$ such that $\mi'\prec\mi$ and $\mi'\in\cpmi$. This however contradicts the assumption that $\bs{\uJ}$ is essential.
    
    Finally, (iii) now follows by elimination from (i) and (ii), \Cref{lem:max-structure} and \Cref{defi:special-sets}, as there are no other possibilities.
\end{proof}

\begin{figure}[tbp]
    \centering
    \begin{subfigure}{0.29\textwidth}
    \centering
    \begin{tikzpicture}[scale=0.5]
    \draw (-1.5,0) -- (1.5,0);
    \draw[dashed] (-1.5,0) -- (0,-1.5);
    \draw[dashed] (0,-3) -- (0,-1.5);
    \draw[dashed] (1.5,0) -- (0,-1.5);
    
    \draw[dotted] (1.5,0) -- (1.2,-2);
    \draw[dotted] (2.5,-3) -- (1.2,-2);
    \draw[dotted] (0,-3) -- (1.2,-2);
    
    \filldraw[Ccol] (-1.5,0) circle (35pt);
    \filldraw[Ccol] (1.5,0) circle (35pt);
    \filldraw[Ccol] (0,-3) circle (20pt);
    \filldraw[Ccol] (2.5,-3) circle (20pt);
    \node () at (-1.5,0) {{\footnotesize $1$}};
    \node () at (1.5,0) {{\footnotesize $2$}};
    \node () at (0,-3) {{\footnotesize $3$}};
    \node () at (2.5,-3) {{\footnotesize $4$}};

    \node () at (2.5,-4) {};
    \end{tikzpicture}
    \subcaption[]{}
    \end{subfigure}
    \hspace{2cm}
    \begin{subfigure}{0.49\textwidth}
    \centering
    \begin{tikzpicture}[scale=1]
    \draw (0,0) -- (0,-1);
    \draw (0,0) -- (-2,-1);
    \draw (0,0) -- (2,-1);
    \draw (0,0) -- (1,-1);
    \draw (0,-1) -- (-0.5,-2);
    \draw (0,-1) -- (0.5,-2);
    \draw (-2,-1) -- (-2.5,-2);
    \draw (-2,-1) -- (-1.5,-2);
    \draw (2,-1) -- (2.5,-2);
    \draw (2,-1) -- (1.5,-2);
    \filldraw (1,-1) circle (2pt);
    \filldraw (-2.5,-2) circle (2pt);
    \filldraw (-1.5,-2) circle (2pt);
    \filldraw (2.5,-2) circle (2pt);
    \filldraw (1.5,-2) circle (2pt);
    \filldraw (-0.5,-2) circle (2pt);
    \filldraw (0.5,-2) circle (2pt);
    
    \filldraw[fill=bulkcol] (0,0) circle (2pt);
    \filldraw[fill=bulkcol] (-2,-1) circle (2pt);
    \filldraw[fill=bulkcol] (0,-1) circle (2pt);
    \filldraw[fill=bulkcol] (2,-1) circle (2pt);
    
    \node[] () at (1,-1.3) {{\scriptsize $3$}};
    \node[] () at (-2.5,-2.3) {{\scriptsize $4$}};
    \node[] () at (-1.5,-2.3) {{\scriptsize $0$}};
    \node[] () at (2.5,-2.3) {{\scriptsize $1$}};
    \node[] () at (1.5,-2.3) {{\scriptsize $2$}};
    \node[] () at (-0.5,-2.3) {{\scriptsize $1$}};
    \node[] () at (0.5,-2.3) {{\scriptsize $0$}};
    \node[red] () at (-2.45,-1.5) {{\scriptsize $3$}};
    \node[red] () at (-1.55,-1.5) {{\scriptsize $7$}};
    \node[red] () at (-0.45,-1.5) {{\scriptsize $3$}};
    \node[red] () at (0.45,-1.5) {{\scriptsize $1$}};
    \node[red] () at (2.45,-1.5) {{\scriptsize $1$}};
    \node[red] () at (1.55,-1.5) {{\scriptsize $5$}};
    \node[red] () at (-1.35,-0.5) {{\scriptsize $8$}};
    \node[red] () at (-0.15,-0.5) {{\scriptsize $2$}};
    \node[red] () at (0.4,-0.6) {{\scriptsize $4$}};
    \node[red] () at (1.3,-0.5) {{\scriptsize $4$}};
    \end{tikzpicture}
    \subcaption[]{}
    \end{subfigure}
    \par\bigskip\bigskip
    \begin{subfigure}{\textwidth}
    \centering
    \begin{tikzpicture}
    
    \node[fill=PminusEcol, align=center, rounded corners] (01234) at (0,6) {{\footnotesize $\bs{01234}$}};
    
    \node[fill=PminusEcol, align=center, rounded corners] (0123) at (-4,4.5) {{\footnotesize $\bs{0123}$}};
    \node[fill=PminusEcol, align=center, rounded corners] (0124) at (-2,4.5) {{\footnotesize $\bs{0124}$}};
    \node[fill=PminusEcol, align=center, rounded corners] (0134) at (0,4.5) {{\footnotesize $\bs{0134}$}};
    \node[fill=PminusEcol, align=center, rounded corners] (0234) at (2,4.5) {{\footnotesize $\bs{0234}$}};
    \node[fill=PminusEcol, align=center, rounded corners] (1234) at (4,4.5) {{\footnotesize $\bs{1234}$}};

    \node[fill=PminusEcol, align=center, rounded corners] (012) at (-6.3,3) {{\footnotesize $\bs{012}$}};
    \node[fill=PminusEcol, align=center, rounded corners] (013) at (-4.9,3) {{\footnotesize $\bs{013}$}};
    \node[fill=PminusEcol, align=center, rounded corners] (014) at (-3.5,3) {{\footnotesize $\bs{014}$}};
    \node[fill=PminusEcol, align=center, rounded corners] (023) at (-2.1,3) {{\footnotesize $\bs{023}$}};
    \node[fill=PminusEcol, align=center, rounded corners] (024) at (-0.7,3) {{\footnotesize $\bs{024}$}};
    \node[fill=PminusEcol, align=center, rounded corners] (034) at (0.7,3) {{\footnotesize $\bs{034}$}};
    \node[fill=EminusMcol, align=center, rounded corners] (123) at (2.1,3) {{\footnotesize $\bs{123}$}};
    \node[fill=Qcol, align=center, rounded corners] (124) at (3.5,3) {{\footnotesize $\bs{124}$}};
    \node[fill=Vcol, align=center, rounded corners] (134) at (4.9,3) {{\footnotesize $\bs{134}$}};
    \node[fill=Mcol, align=center, rounded corners] (234) at (6.3,3) {{\footnotesize $\bs{234}$}};

    \node[fill=Mcol, align=center, rounded corners] (01) at (-6.3,1.5) {{\footnotesize $\bs{01}$}};
    \node[fill=Mcol, align=center, rounded corners] (02) at (-4.9,1.5) {{\footnotesize $\bs{02}$}};
    \node[fill=Mcol, align=center, rounded corners] (03) at (-3.5,1.5) {{\footnotesize $\bs{03}$}};
    \node[fill=Mcol, align=center, rounded corners] (04) at (-2.1,1.5) {{\footnotesize $\bs{04}$}};
    \node[fill=Mcol, align=center, rounded corners] (12) at (-0.7,1.5) {{\footnotesize $\bs{12}$}};
    \node[fill=Vcol, align=center, rounded corners] (13) at (0.7,1.5) {{\footnotesize $\bs{13}$}};
    \node[fill=Vcol, align=center, rounded corners] (14) at (2.1,1.5) {{\footnotesize $\bs{14}$}};
    \node[fill=Vcol, align=center, rounded corners] (23) at (3.5,1.5) {{\footnotesize $\bs{23}$}};
    \node[fill=Vcol, align=center, rounded corners] (24) at (4.9,1.5) {{\footnotesize $\bs{24}$}};
    \node[fill=Vcol, align=center, rounded corners] (34) at (6.3,1.5) {{\footnotesize $\bs{34}$}};

    \draw[dotted,thick] (01234.south) -- (0123.north);
    \draw[dotted,thick] (01234.south) -- (0124.north);
    \draw[dotted,thick] (01234.south) -- (0134.north);
    \draw[dotted,thick] (01234.south) -- (0234.north);
    \draw[dotted,thick] (01234.south) -- (1234.north);

    \draw[dotted,thick] (1234.south) -- (123.north);
    \draw[dotted,thick] (1234.south) -- (234.north);

    \draw[dotted,thick] (0123.south) -- (123.north);
    \draw[dotted,thick] (0123.south) -- (012.north);
    \draw[dotted,thick] (0123.south) -- (013.north);
    \draw[dotted,thick] (0123.south) -- (023.north);

    \draw[dotted,thick] (0124.south) -- (012.north);
    \draw[dotted,thick] (0124.south) -- (014.north);
    \draw[dotted,thick] (0124.south) -- (024.north);
    
    \draw[dotted,thick] (0134.south) -- (013.north);
    \draw[dotted,thick] (0134.south) -- (014.north);
    \draw[dotted,thick] (0134.south) -- (034.north);

    \draw[dotted,thick] (0234.south) -- (234.north);
    \draw[dotted,thick] (0234.south) -- (023.north);
    \draw[dotted,thick] (0234.south) -- (024.north);
    \draw[dotted,thick] (0234.south) -- (034.north);

    \draw[dotted,thick] (123.south) -- (12.north);
    \draw[dotted,thick] (124.south) -- (12.north);
    \draw[dotted,thick] (012.south) -- (12.north);
    \draw[dotted,thick] (012.south) -- (01.north);
    \draw[dotted,thick] (012.south) -- (02.north);
    \draw[dotted,thick] (013.south) -- (01.north);
    \draw[dotted,thick] (013.south) -- (03.north);
    \draw[dotted,thick] (014.south) -- (01.north);
    \draw[dotted,thick] (014.south) -- (04.north);
    \draw[dotted,thick] (023.south) -- (02.north);
    \draw[dotted,thick] (023.south) -- (03.north);
    \draw[dotted,thick] (024.south) -- (02.north);
    \draw[dotted,thick] (024.south) -- (04.north);
    \draw[dotted,thick] (034.south) -- (03.north);
    \draw[dotted,thick] (034.south) -- (04.north);
    
    \end{tikzpicture}
    \subcaption[]{}
    \end{subfigure}
    \caption{An example of an $\N=4$ holographic KC-PMI that illustrates \Cref{thm:max-structure}. (a) A choice of a holographic configuration where the shaded discs represent boundary regions and the lines represent RT surfaces as in \cite{Hubeny:2018trv,Hubeny:2018ijt}. Each line style corresponds to a connected RT surface anchored to the disks connected by the lines with the same style. (b) A holographic graph model realization.
    (c) The classification of $\bsets$ where, unlike in previous figures, the lines do not represent the partial order relations of $\lat{$\beta$}^4$ but instead connect a $\bset$ $\bs{\uJ}$ with the elements of $\text{Max}(\posj)$. 
    With respect to the statement of \Cref{thm:max-structure}, notice in particular that: (i) there are no downwards lines from $\bs{134}$ and $\bs{234}$, since $\bs{134}\in\van$ and $\bs{234}\in\cess$ 
    (ii) there is a single downwards line from $\bs{123}$ and $\bs{124}$, since $\bs{123}\in\ncess$ and $\bs{124}\in\parv$, (iii) there are at least two downwards lines from each $\bset$ in $\pnote$, and for each such $\bset$ $\bs{\uJ}$, the subsystems of the $\bsets$ reached by these lines have non-empty pairwise intersection and their union is $\uJ$.} 
    \label{fig:N4-example}
\end{figure}

\Cref{thm:max-structure} completes the first step of our plan for this subsection, as it immediately implies that for any KC-PMI we can reconstruct the set of essential $\bsets$ from that of positive ones. Indeed, if we are given $\pos$, we can determine $\text{Max}(\posj)$ for any $\uJ$, and therefore distinguish the elements of $\pos$ which are in $\ess$ from those in $\pnote$. On the other hand, notice that if for a KC-PMI we are only given the collection of essential $\bsets$, \Cref{thm:max-structure} does not show how to distinguish positive $\bsets$ from partial ones, since we need to know $\pos$ to determine $\posj$ for a given $\bs{\uJ}$. 

To proceed towards this reconstruction, suppose we are given the essential $\bsets$, and let $\xx$ denote the set of positive ones obtained from \Cref{lem:join} as follows. Start by setting $\xx=\ess$, and suppose there are $\bs{\uJ_1},\bs{\uJ_2}\in\ess$ with $\uJ_1\cap\uJ_2\neq\varnothing$. By \Cref{lem:join}, $\bs{\uJ_1\cup\uJ_2}\in\pnote$ and we add it to $\xx$. We then iterate this procedure, and construct new elements of $\pnote$ until we can no longer find two $\bs{\uJ_1},\bs{\uJ_2}\in\xx$ with $\uJ_1\cap\uJ_2\neq\varnothing$ such that $\bs{\uJ_1\cup\uJ_2}\notin\xx$. To complete the reconstruction of $\pos$ from $\ess$, it only remains to be shown that any element of $\pnote$ can be obtained via this procedure, i.e., that when the iteration stops we have $\xx=\pos$. This is proven by the next corollary.

\begin{cor}
\label{cor:P-from-E}
    For any \emph{KC-PMI}, and any $\bset$ $\bs{\uJ}\in\pnote$, there exists a sequence $\left(\bs{\uI_1},\bs{\uI_2},\ldots,\bs{\uI_r}\right)$ of essential $\bsets$ with $r\geq 2$, such that $\uJ=\bigcup_{i=1}^r \uI_i$ and for every $i\in\{2,3,\ldots,r\}$,
    \begin{equation}
        \uI_i\cap\left(\bigcup_{0<j<i}\uI_j\right)\neq\varnothing.
    \end{equation}
\end{cor}
\begin{proof}
    Consider some $\bs{\uJ}\in\pnote$, by (iii) of \Cref{thm:max-structure} it follows that there is a collection of positive $\bsets$ $\left\{\bs{\uK_1},\bs{\uK_2},\ldots,\bs{\uK_n}\right\}$ with $n\geq 2$ such that for all  $i,j\in[n]$ with $i\neq j$, $\uK_i\cap\uK_j\neq\varnothing$ and $\uK_i\cup\uK_j=\uJ$. If any two of these $\bsets$ are essential the statement is proven, so suppose that this is not the case. We can then take any two $\bs{\uK_i}$ and repeat the same argument for each non-essential one. We can describe this process by a strict\footnote{\,We call a binary tree strict if each node has $0$ or $2$ children, but never $1$.} binary rooted tree where the root is $\bs{\uJ}$, each node is a positive $\bset$, and a node is a leaf if and only if it is essential. Furthermore, for any node which is not a leaf, the subsystems of its two children have non-empty intersection and their union is the subsystem of the node. 
    
    We can then construct the sequence $\left(\bs{\uI_1},\bs{\uI_2},\ldots,\bs{\uI_r}\right)$ from the leaves as follows. We start from an arbitrary leaf that we set to be $\bs{\uI_1}$. We then move upwards along the tree to the first node that contains a party $\ell\notin\uI_1$ (in this first step, simply its parent), and then downwards to an arbitrary leaf that contains $\ell$ (notice that because of the structure of the subsystems of the nodes this is always possible) and that we set to be $\bs{\uI_2}$. If $\uI_1\cup\uI_2\subset\uJ$, we then move upwards again from $\bs{\uI_2}$ to the first ancestor that contains a party $\ell'\notin\uI_1\cup\uI_2$,
    and downwards again to another leaf that contains $\ell'$ and we set to be $\uI_3$. Proceeding in this fashion until we have comprised all of $\uJ$ (which is guaranteed by the time we have used all the leaves), we obtain the desired sequence.
\end{proof}

Having completed the characterization of the various sets in \Cref{defi:special-sets}, we now proceed to describe the MI instances in the various $\bsets$. Given an arbitrary KC-PMI $\pmi$ and a $\bset$ which is either positive but not essential, or completely essential, or vanishing, this is trivial from its definition: All instances in a positive but not essential $\bset$ are in $\cpmi$ but not in $\ac$, all instances in a completely essential $\bset$ are in $\ac$, and all instances in a vanishing $\bset$ are in $\pmi$. However, for a $\bset$ which is essential but not completely essential, only some MI instances are in $\ac$, and similarly, for a partial $\bset$, only some MI instances are in $\pmi$. Our next goal is to determine these MI instances purely in terms of $\bsets$.

We begin by introducing a useful terminology. Given an arbitrary subsystem $\uI$ and an arbitrary pair of disjoint subsystems $\{\uJ,\uK\}$, we will say that this pair \textit{splits} $\uI$ if\,\footnote{\,Note that the subsystems $\uJ\cup\uK$ and $\uI$ need not coincide, or even nest, so in general the split corresponds to a tri-partition of $\uI$. However, in the actual application below, we will typically be interested in the case $\uJ\cup\uK \supset \uI$.}
\begin{equation}
\label{eq:split}
    \uI\cap\uJ\neq\varnothing\quad \text{and}\quad \uI\cap\uK\neq\varnothing.
\end{equation}
Correspondingly, we will say that $\uI$ \textit{is split} by the pair $\{\uJ,\uK\}$, and that an MI instance $\mi(\uJ:\uK)$ splits $\uI$, if $\uI$ is split by $\{\uJ,\uK\}$. Notice in particular that for any given $\bs{\uJ}$, $\mi\in\bs{\uJ}$ and $\uJ'\supset\uJ$, it follows that any $\mi'\in\bs{\uJ'}$ with $\mi'\succ\mi$ splits $\uJ$.

\begin{thm}
\label{thm:partition}
    For any \emph{KC-PMI} $\pmi$, any\, $\bs{\uJ}\!\in\!\ess\cup\parv$, and any\, $\mi\!\in\!\bs{\uJ}$:
    \begin{enumerate}[label={\emph{\footnotesize \roman*)}}]
        \item if $\bs{\uJ}\in\parv$, then $\mi\in\pmi$ if and only if it does not split any $\uK$ such that $\bs{\uK}\in\text{\emph{Max}}(\posj)$,
        \item if $\bs{\uJ}\in\ess$, then $\mi\in\ac$ if and only if it does not split any $\uK$ such that $\bs{\uK}\in\text{\emph{Max}}(\posj)$.
    \end{enumerate}
\end{thm}
\begin{proof}
    i) One direction is obvious, if $\mi$ splits some $\uK$ with $\bs{\uK}\in\text{Max}(\posj)$, then there is some $\mi'\in\bs{\uK}$ such that $\mi'\prec\mi$, and since $\mi'\in\cpmi$, it must be $\mi\in\cpmi$. Conversely, consider some $\mi\in\bs{\uJ}$ that does not split any $\uK$ where $\bs{\uK}\in\text{Max}(\posj)$, and suppose that $\mi\in\cpmi$. It must then be that $\mi\notin\ac$, otherwise $\bs{\uJ}$ would be essential but by \Cref{thm:positiveness} any essential $\bset$ is positive and by assumption $\bs{\uJ}\in\parv$. If $\mi\notin\ac$, then there is some $\mi'\prec\mi$ such that $\mi'\in\ac$, and denoting by $\bs{\uK'}$ the $\bset$ that contains $\mi'$, it follows again by \Cref{thm:positiveness} that $\bs{\uK'}$ is positive. By the definition of $\text{Max}(\posj)$, there must then be some $\bs{\uK''}\in\text{Max}(\posj)$ such that $\bs{\uK'}\preceq\bs{\uK''}$. However, if $\mi'\in\bs{\uK'}$ it splits $\uK'$, and since $\uK'\subseteq\uK''\subseteq\uJ$, the order relation $\mi'\prec\mi$ is only possible if $\mi$ splits $\uK''$ (cf., the comment just before the statement of the theorem), which is a contradiction.

    ii) As before, one direction is obvious, if $\mi$ splits some $\uK$ with $\bs{\uK}\in\text{Max}(\posj)$, then there is some $\mi'\in\bs{\uK}$ such that $\mi'\prec\mi$, and since $\mi'\in\cpmi$, it follows that $\mi\notin\ac$. Conversely, consider some $\mi\in\bs{\uJ}$ that does not split any $\uK$ where $\bs{\uK}\in\text{Max}(\posj)$, and suppose that $\mi\in\cpmi$ (since we are assuming $\bs{\uJ}\in\ess$ and $\ess\subseteq\pos$ by \Cref{thm:positiveness}) but $\mi\notin\ac$. The proof then continues exactly as in the previous case.
\end{proof}

In light of this result, suppose that for a given KC-PMI $\pmi$ we know $\pos$ and we want to list all instances of the mutual information in $\pmi$ for a partial $\bset$ $\bs{\uJ}$ (or alternately all instances of the mutual information in $\ac$ if $\bs{\uJ}$ is essential). To this end, it is important to notice that for an arbitrary KC-PMI and an arbitrary subsystem $\uJ$ such that 
$\bs{\uJ}\in\ncess\cup\parv$, \Cref{thm:max-structure} does not imply that the union of all subsystems in the elements of $\text{Max}(\posj)$ is $\uJ$, i.e., in general we only have the inclusion 
\begin{equation}
\label{eq:max-inclusion}
    \bigcup_{\bs{\uK}\,\in\,\text{Max}(\posj)} \!\!\uK \;\; \subseteq \;\; \uJ.
\end{equation}
When equality holds, (ii) in \Cref{thm:max-structure} implies that the subsystems of the elements of $\text{Max}(\posj)$ form a non-trivial partition of $\uJ$.\footnote{\,By non-trivial here we mean that the partition contains more than one element, i.e., $n\geq 2$. This follows from the fact that by definition of $\posj$ we have the strict inclusion $\uK_i\subset\uJ$ for all $i\in [n]$.} 
On the other hand, if the inclusion in Eq.~\eqref{eq:max-inclusion} is strict, it is useful to consider a specific partition $\Gamma(\uJ)$ of $\uJ$ which takes the following form 
\begin{equation}
\label{eq:partition}
        \Gamma(\uJ)= \left\{ \uK_1,\uK_2,\ldots\uK_n,\ell_1,\ell_2,\ldots,\ell_{\tilde{n}} \right\},
\end{equation}
where $n\geq 0$ and $\tilde{n}=|\Delta|$ with
\begin{equation}
    \Delta=\,\uJ\,\setminus\,\bigcup_{i=1}^n\, \uK_i.
\end{equation}
Notice that we included here the possibility that $n=0$, corresponding to the case of $\bs{\uJ}\in\cess\cup\van$, where $\text{Max}(\posj)=\varnothing$, in which case Eq.~\eqref{eq:partition} reduces to a \textit{singleton partition}, where each element is a single party. Given an arbitrary $\bset$ $\bs{\uJ}\in\ncess\cup\parv$ then, the desired instances of the mutual information (as characterized by \Cref{thm:partition}) are then simply all the ones that can be obtained by collecting some elements of the partition into one argument, and all the others in the other argument, i.e., the those of the form
\begin{equation}
    \mi(\uI_1:\uI_2)\quad \text{with}\quad \uI_1\cup\uI_2=\uJ\quad \text{and}\quad \forall\,i\in[n],\;\; \uK_i\subseteq\uI_1\;\; \text{or}\;\; \uK_i\subseteq\uI_2.
\end{equation}

Having completed not just the classification of all types of the $\bsets$ but also the characterization of the instances of the mutual information in the various $\bsets$, we conclude this subsection by addressing a few further questions that one may ask about the reconstruction of KC-PMIs from $\bsets$ and the structure of the various sets in \Cref{defi:special-sets}. The focus of this subsection was on the reconstruction of a KC-PMI $\pmi$ from its set of positive or essential $\bsets$, and one may wonder if it might also be possible to recover $\pmi$ from the knowledge of only the completely essential ones. To see that in general this is not possible, it is sufficient to find an example of at least two distinct KC-PMIs with the same set $\cess$. The simplest examples appear already at $\N=3$, and we present one in Fig.~\ref{fig:N3-same-M}. 

We have also seen that for any KC-PMI there is a separation between the sets $\pos\cup\parv$ and $\van$ in the form of an up-set and its complementary down-set. It is then natural to ask if the order structure in $\lat{$\beta$}^{\N}$ might also be used to distinguish the elements of $\pos$ from those of $\parv$, and in particular, if $\pos$ and $\parv$ might be up-sets themselves.\footnote{\,Since they partition an up-set it is obvious that they cannot be down-sets. (The only exception is the trivial one where all parties are pairwise correlated so that $\pmi=\varnothing$ and $\pos(\ds)  = \mathfrak{B}_{\N}$, which is of course both an up-set and a down-set.)} The examples in Fig.~\ref{fig:N3-same-M} show that in general this is not the case, and likewise that $\ess$ in general is not an up-set. Finally, letting $\ket{\psi}_{0123}$ be a quantum state realizing one of the $\N=3$ KC-PMIs in Fig.~\ref{fig:N3-same-M}, we leave it as an exercise for the reader to verify that the $\N=4$ KC-PMI realized by the state $\ket{\psi}_{0123}\ket{0}_4$ provides an example where $\pnote$ is not an up-set.

\begin{figure}[tbp]
    \centering
    \begin{subfigure}{0.29\textwidth}
    \centering
    \begin{tikzpicture}[scale=0.5]
    \draw (-1.5,0) -- (1.5,0);
    \filldraw[Ccol] (-1.5,0) circle (30pt);
    \filldraw[Ccol] (1.5,0) circle (30pt);
    \filldraw[Ccol] (0,-4) circle (15pt);
    \node () at (-1.5,0) {{\footnotesize $1$}};
    \node () at (1.5,0) {{\footnotesize $2$}};
    \node () at (0,-4) {{\footnotesize $3$}};
    \node () at (0,-5) {};
    \end{tikzpicture}
    \subcaption[]{}
    \end{subfigure}
    \begin{subfigure}{0.69\textwidth}
    \centering
    \begin{tikzpicture}[scale=0.5]
    \node[fill=PminusEcol, align=center, rounded corners] (0123) at (0,4) {{\footnotesize $\bs{0123}$}};
    \node[fill=PminusEcol, align=center, rounded corners] (012) at (-6,2) {{\footnotesize $\bs{012}$}};
    \node[fill=PminusEcol, align=center, rounded corners] (013) at (-2,2) {{\footnotesize $\bs{013}$}};
    \node[fill=PminusEcol, align=center, rounded corners] (023) at (2,2) {{\footnotesize $\bs{023}$}};
    \node[fill=Qcol, align=center, rounded corners] (123) at (6,2) {{\footnotesize $\bs{123}$}};
    \node[fill=Mcol, align=center, rounded corners] (01) at (-7.5,0) {{\footnotesize $\bs{01}$}};
    \node[fill=Mcol, align=center, rounded corners] (02) at (-4.5,0) {{\footnotesize $\bs{02}$}};
    \node[fill=Mcol, align=center, rounded corners] (03) at (-1.5,0) {{\footnotesize $\bs{03}$}};
    \node[fill=Mcol, align=center, rounded corners] (12) at (1.5,0) {{\footnotesize $\bs{12}$}};
    \node[fill=Vcol, align=center, rounded corners] (13) at (4.5,0) {{\footnotesize $\bs{13}$}};
    \node[fill=Vcol, align=center, rounded corners] (23) at (7.5,0) {{\footnotesize $\bs{23}$}};
    \end{tikzpicture}
    \subcaption[]{}
    \end{subfigure}
    \par\bigskip\bigskip
    \begin{subfigure}{0.29\textwidth}
    \centering
    \begin{tikzpicture}[scale=0.5]
    \draw (-1.5,0) -- (1.5,0);
    \draw[dotted] (1.5,0) -- (0,-1);
    \draw[dotted] (0,-3) -- (0,-1);
    \draw[dotted] (-1.5,0) -- (0,-1);
    \filldraw[Ccol] (-1.5,0) circle (30pt);
    \filldraw[Ccol] (1.5,0) circle (30pt);
    \filldraw[Ccol] (0,-2.5) circle (20pt);
    \node () at (-1.5,0) {{\footnotesize $1$}};
    \node () at (1.5,0) {{\footnotesize $2$}};
    \node () at (0,-2.5) {{\footnotesize $3$}};
    \node () at (0,-3.5) {};
    \end{tikzpicture}
    \subcaption[]{}
    \end{subfigure}
    \begin{subfigure}{0.69\textwidth}
    \centering
    \begin{tikzpicture}[scale=0.5]
    \node[fill=PminusEcol, align=center, rounded corners] (0123) at (0,4) {{\footnotesize $\bs{0123}$}};
    \node[fill=PminusEcol, align=center, rounded corners] (012) at (-6,2) {{\footnotesize $\bs{012}$}};
    \node[fill=PminusEcol, align=center, rounded corners] (013) at (-2,2) {{\footnotesize $\bs{013}$}};
    \node[fill=PminusEcol, align=center, rounded corners] (023) at (2,2) {{\footnotesize $\bs{023}$}};
    \node[fill=EminusMcol, align=center, rounded corners] (123) at (6,2) {{\footnotesize $\bs{123}$}};
    \node[fill=Mcol, align=center, rounded corners] (01) at (-7.5,0) {{\footnotesize $\bs{01}$}};
    \node[fill=Mcol, align=center, rounded corners] (02) at (-4.5,0) {{\footnotesize $\bs{02}$}};
    \node[fill=Mcol, align=center, rounded corners] (03) at (-1.5,0) {{\footnotesize $\bs{03}$}};
    \node[fill=Mcol, align=center, rounded corners] (12) at (1.5,0) {{\footnotesize $\bs{12}$}};
    \node[fill=Vcol, align=center, rounded corners] (13) at (4.5,0) {{\footnotesize $\bs{13}$}};
    \node[fill=Vcol, align=center, rounded corners] (23) at (7.5,0) {{\footnotesize $\bs{23}$}};
    \end{tikzpicture}
    \subcaption[]{}
    \end{subfigure}
    \caption{Two distinct $\N=3$ holographic KC-PMIs with the same set of completely essential $\bsets$. (a),(c) holographic configurations presented as in Fig.~\ref{fig:N4-example}, in particular the region 3 is correlated with the union of 1 and 2 in (c) but not in (a). (b),(d) the lattice $\lat{$\beta$}^3$ showing the corresponding KC-PMIs, where we have omitted $\bs{\ell}$ for all parties, and the minimum $\bs{\varnothing}$.}
    \label{fig:N3-same-M}
\end{figure}

\subsection{Definition of the correlation hypergraph and its main properties}
\label{subsec:hypergraph-rep}

In the previous subsection we have demonstrated how an arbitrary KC-PMI can be described purely in terms of positive or essential $\bsets$. We now introduce a representation of a KC-PMI in terms of a certain hypergraph, and show how to translate most of the previous results into this language. As we will explain, this representation will make transparent the generalization, to arbitrary KC-PMIs, of the relationship between correlation and connectivity highlighted for holographic graph models in \Cref{subsec:graph_review}. We begin by stating the definition of the correlation hypergraph, which is based on positive $\bsets$.

\begin{defi}[Correlation hypergraph of a KC-PMI]
\label{defi:hypergarph_rep}
    For a given $\N$-party \emph{KC-PMI} $\pmi$, the correlation hypergraph of $\pmi$ is the hypergraph $\hp =(V,E)$ 
    with $\N+1$ vertices $V=\{v_\ell,\ell\in\nsp\}$ and hyperedges\,\footnote{\,We call any $e_{\uJ}$ a hyperedge even when $|e_{\uJ}|=2$.
    Notice that the one-to-one correspondence between hyperedges ($e_{\uJ}$) and subsystems ($\uJ$) is inherited from the one-to-one correspondence between vertices ($v_\ell$) and parties ($\ell$).} $E=\{e_{\uJ},\, \bs{\uJ}\in\pos\}$ where $e_{\uJ}=\{v_\ell,\,\ell\in\uJ\}$.
\end{defi}

In other words, given an arbitrary KC-PMI, its correlation hypergraph is constructed starting with a set of vertices corresponding to the parties (including the purifier), and by adding a hyperedge for each positive $\bset$ $\bs{\uI}$ comprising the vertices corresponding to all singletons in $\uI$.

Two simple examples of this representation are the 0-dimensional and the full dimensional KC-PMIs, whose corresponding faces of the SAC$_{\N}$ are respectively the origin of entropy space and the interior of the cone. In the first case, the correlation hypergraph is the \textit{trivial} hypergraph with $\N+1$ vertices and no hyperedges, while in the second, it is the complete hypergraph with $\N+1$ vertices. Another simple example is the hypergraph with four vertices representing the $\N=3$ KC-PMI illustrated in Fig.~\ref{fig:PT3}, whose hyperedges correspond to the yellow and red $\bsets$.

To proceed further, we will need a few more standard definitions from the theory of graphs and hypergraphs. Given an arbitrary hypergraph $\gf{H}=(V,E)$, and an arbitrary subset of vertices $V'\subseteq V$, the \textit{sub-hypergraph} of $\gf{H}$ \textit{induced} by $V'$ is the hypergraph $\gf{H}'=(V',E')$ with 
\begin{equation}
    E'=\{e\in E,\, e\subseteq V'\}.
\end{equation}
More generally, a \textit{sub-hypergraph} of $\gf{H}$ is any hypergraph such that its sets of vertices and hyperedges are subsets of those of $\gf{H}$. 

For a given KC-PMI $\pmi$, we now associate to an arbitrary subsystem $\uJ$ a certain sub-hypergraph of $\hp$, which we denote by $\subh{\uJ}$ and construct as follows. If $\bs{\uJ}\in\pos$, the subsystem $\uJ$ corresponds to a hyperedge $e_{\uJ}$, and $\subh{\uJ}$ is obtained by first extracting the sub-hypergraph of $\hp$ induced by the vertices in $e_{\uJ}$, and then deleting $e_{\uJ}$. On the other hand, if $\bs{\uJ}\notin\pos$, $\subh{\uJ}$ is simply the sub-hypergraph of $\hp$ induced by the vertices in $e_{\uJ}$. In summary, we define $\subh{\uJ}=(V',E')$ with
\begin{equation}
        V'=\{v_\ell\in V,\; \ell\in\uJ\}, \qquad E'=\{e_{\uK}\in E,\; \uK\subset\uJ\}.
\end{equation}
Notice in particular that the inclusion in the definition of $E'$ is strict. This choice is motivated by the fundamental role played by $\text{Max}(\posj)$ in \Cref{thm:max-structure}, which will be rephrased in the language of the correlation hypergraph below.

We will also need a notion of \textit{connectivity}. Following \cite{Bretto} we define an open path in an arbitrary hypergraph $\gf{H}$ as a vertex-hyperedge alternating sequence 
\begin{equation}
    v^1,e^1,v^2,e^2,\ldots,e^{m-1},v^m
\end{equation}
such that all vertices are distinct, all hyperedges are distinct, and $v_i,v_{i+1}\in e_i$ for all $i\in [m]$. We will say that a hypergraph is \textit{connected} if there is a path from any vertex to any other vertex, otherwise it is \textit{disconnected}. A hypergraph with a single vertex is connected, even if it does not have any hyperedge.\footnote{\,In fact, notice that even including in $\mathfrak{B}_\N$ the trivial $\bsets$ in Eq.~\eqref{eq:trivial_bsets}, the correlation hypergraph has no loops (i.e., a hyperedge connecting a vertex to itself), since for any $\bs{\uI}\in\pos$ we necessarily have $|\uI|\geq 2$.} Finally, analogously to the previous subsection, given an arbitrary hyperedge $e_{\uI}$ and a pair of disjoint subsystems $\{\uJ,\uK\}$, we will say that the pair \textit{splits} $e_{\uI}$ (or equivalently, that $e_{\uI}$ is split by this pair), if $\{\uJ,\uK\}$ splits $\uI$, and we will use the same terminology for an instance of the mutual information. 

With these definitions we can now rephrase the content of \Cref{thm:max-structure} and \Cref{thm:partition} from the previous subsection as follows.
\begin{thm}
\label{thm:connectivity_of_H} 
    For any \emph{KC-PMI}, and any subsystem $\uJ\subseteq\nsp$
    \begin{enumerate}[label={\emph{\footnotesize \roman*)}}]
        \item $\bs{\uJ}\in\van$ if and only if $\hp$ has no $e_{\uJ}$ hyperedge and $\subh{\uJ}$ is trivial,
        \item $\bs{\uJ}\in\cess$ if and only if $e_{\uJ}$ is a hyperedge of $\hp$ and $\subh{\uJ}$ is trivial,
        \item $\bs{\uJ}\in\parv$ if and only if $\hp$ has no $e_{\uJ}$ hyperedge, $\subh{\uJ}$ is disconnected, and it is non-trivial. Furthermore for any $\mi\in\bs{\uJ}$, $\mi\in\pmi$ if and only if it does not split any hyperedge of $\subh{\uJ}$,
        \item $\bs{\uJ}\in\ncess$ if and only if $e_{\uJ}$ is a hyperedge, $\subh{\uJ}$ is disconnected, and it is non-trivial. Furthermore for any $\mi\in\bs{\uJ}$, $\mi\in\ac$ if and only if it does not split any hyperedge of $\subh{\uJ}$,
        \item $\bs{\uJ}\in\pnote$ if and only if $|\uJ|>2$ and $\subh{\uJ}$ is connected.
    \end{enumerate}
\end{thm}
\begin{proof} 
    This theorem is a translation of \Cref{thm:max-structure} and \Cref{thm:partition} into the language of the correlation 
    hypergraph representation. It follows by noticing that by \Cref{defi:hypergarph_rep}, $\hp$ has a hyperedge $e_{\uJ}$ if and only if $\bs{\uJ}\in\pos$, and that the maximal hyperedges of $\subh{\uJ}$ (i.e., those that are not contained in any other hyperedge of $\subh{\uJ}$) correspond to the elements of $\text{Max}(\posj)$. In particular, the fact that $\subh{\uJ}$ is connected if and only if $|\uJ|>2$ and $\bs{\uJ}\in\pnote$ is implied by the fact that in this case (and only in this case) the subsystems of the elements of $\text{Max}(\posj)$ intersect pairwise, and the union of any two such subsystems is $\uJ$, guaranteeing that there is path between any two vertices of $\hp$. The condition $|\uJ|>2$ is required to exclude two special cases. If $|\uJ|=1$ then $\bs{\uJ}\in\van$ and $\subh{\uJ}$ is a single vertex without hyperedges, though it is connected by definition. On the other hand, if $|\uJ|=2$, $\bs{\uJ}$ can only be vanishing or completely essential, i.e., $\bs{\uJ}\in\van\cup\cess$.
\end{proof}

We have mentioned before that a particularly convenient way of listing the relevant instances of mutual information for a $\bset$ $\bs{\uJ}$ in $\parv$ or $\ess$ is the partition $\Gamma(\uJ)$ given in Eq.~\eqref{eq:partition}. This also naturally translates into the language of the correlation hypergraph representation of a KC-PMI.

\begin{cor}
\label{cor:connected-components}
    For any \emph{KC-PMI} $\pmi$, and any $\uJ$ such that $\bs{\uJ}\notin\pnote$, or equivalently $\bs{\uJ}\in\van\cup\parv\cup\ess$, the connected components of $\subh{\uJ}$ are the sub-hypergraphs of $\hp$ induced by the elements of the partition $\Gamma(\uJ)$ given in Eq.~\eqref{eq:partition}.\footnote{\,Technically, a sub-hypergraph of $\hp$ is not induced by a subsystem $\uJ$ but by all the vertices $v_\ell$ labeled by the parties $\ell\in\uJ$. However, for simplicity, we will occasionally use this terminology.}
\end{cor}
\begin{proof}
    This is another straightforward application of \Cref{thm:max-structure} and \Cref{thm:connectivity_of_H}, and we leave the details as an exercise for the reader.
\end{proof}

Let us now make a few observations about the natural interpretation that \Cref{thm:connectivity_of_H} and \Cref{cor:connected-components} provide to the correlation hypergraph representation of a KC-PMI, and highlight some of the remarkable similarities with the results reviewed in \S\ref{subsec:graph_review} for holographic graph models. It follows immediately from \Cref{thm:connectivity_of_H} that $\hp$ is connected if and only if $\bs{\nsp}$ is positive. When $\bs{\nsp}$ is not positive, then by \Cref{cor:connected-components} the connected components of $\subh{\nsp}$ are the elements of the partition $\Gamma(\nsp)$, and the instances of mutual information in $\bs{\nsp}$ that vanish are, again by \Cref{thm:connectivity_of_H}, the ones that do not split any subsystem in $\Gamma(\nsp)$. For a KC-PMI that is realizable by a graph model, it is straightforward to see that the same statements are true for any sensible\footnote{\,The only (irrelevant) exception is the case where a graph model is disconnected and there are connected components which do not contain any boundary vertex (and therefore do not contribute to any min-cut).} realization. The elements of $\bs{\nsp}$ are the maximal elements of the MI-poset, and they correspond to the entropies of the various (non-empty) subsystems 
\begin{equation}
    \mi(\,\uJ:\comp{\uJ}\,)=\ent_{\uJ}+\ent_{\comp{\uJ}}-\ent_{\nsp}=2\,\ent_{\uJ},
\end{equation}
where we used $\ent_{\nsp}=\ent_{\varnothing}=0$. By the down-set structure of $\pmi$, any subset of $\Gamma(\nsp)$ also corresponds to the decomposition of a non-positive $\bset$, and it follows that $\ent_{\uJ}=0$ for each element of $\Gamma(\nsp)$. Suppose now that $\pmi$ can be realized by a quantum state. Then, since the entropy of a quantum system vanishes if and only if the state is pure, it follows that any state realizing $\pmi$ must factorize, and that each factor corresponds to an element of $\Gamma(\nsp)$. In other words, if $\hp$ is disconnected, and $\pmi$ is realizable, the connected components of $\hp$ correspond to the factors of any state  realizing $\pmi$.

Notice that while we have presented this argument in the simple case where the full hypergraph is disconnected, for KC-PMIs which can be realized by a density matrix, we can apply the same logic to all of its marginals: If a marginal factorizes, by \Cref{thm:connectivity_of_H} and \Cref{cor:connected-components} the factors correspond to the connected components of $\subh{\uJ}$. We can then interpret the hyperedges of $\hp$ as representing the correlations among the various subsystems. For a KC-PMI that is realizable by a graph model $\gf{G}$, it follows from \Cref{cor:ew-decomposition} that a hyperedge $e_{\uJ}$ of the correlation hypergraph corresponds to a subsystem $\uJ$ whose minimal min-cut induces a connected subgraph $\gf{G}_{\uJ}$, while $\subh{\uJ}$ describes the internal structure of $\gf{G}_{\uJ}$ in terms of minimal min-cuts for smaller subsystems. The various statements of \Cref{thm:connectivity_of_H} then, applied to $\gf{G}$, include the relation between connectivity and sign of the mutual information described by \Cref{cor:ew-decomposition}, and further characterize the ``essential'' subsystems in complete analogy with the holographic set-up shown in Fig.~\ref{fig:ew-crossing-hol}.

In fact, we can go even further, and for an entropy vector $\vec{\ent}\in\text{int}(\face_{\pmi})$, use this analogy to associate to each hyperedge in $\hp$ the corresponding component of $\vec{\ent}$ (as for min-cuts on graphs and entanglement wedges in holography). The following result clarifies the meaning of this correspondence.

\begin{thm}
\label{thm:entropy-relations}
    For any \emph{KC-PMI} $\pmi$, the space of solutions to the following system of equations
    \begin{align}
        & \ent_{\uJ}=\sum_{\uK\,\in\,\Gamma(\uJ)} \ent_{\uK}  \qquad \forall\, \uJ\;\; \text{such that}\;\; \bs{\uJ}\notin\pos \label{eq:entropy-relations}\\
        & \ent_{\uJ}=\ent_{\uJ^c} \qquad\qquad\quad\; \forall\,\uJ\subseteq \nsp \label{eq:purity2}
    \end{align}
    is the subspace $\mathbb{S}_{\pmi}$ of $\pmi$.
\end{thm}
\begin{proof}
    The subspace $\mathbb{S}_{\pmi}$ of an arbitrary PMI $\pmi$ was defined in \S\ref{subsec:KC-review} as the space of solutions to the system of equations given in Eq.~\eqref{eq:pmi-equations} and Eq.~\eqref{eq:purity}. We then need to show that each of these equations is a linear combination of the equations in Eq.~\eqref{eq:entropy-relations} and Eq.~\eqref{eq:purity2}, and vice versa. Notice that Eq.~\eqref{eq:purity} and \eqref{eq:purity2} are exactly the same, so we only need to focus on Eq.~\eqref{eq:pmi-equations} and \eqref{eq:entropy-relations}.

    In the first direction, consider a KC-PMI $\pmi$ and an instance of  mutual information $\mi\in\pmi$. By \Cref{defi:special-sets}, the $\bset$ $\bs{\uJ}$ that contains $\mi$ is in $\parv\cup\van$. Consider the partition $\Gamma(\uJ)$ given in Eq.~\eqref{eq:partition}, and let us write it as $\Gamma(\uJ)=\{\uI_k\}_{k\in [m]}$, where $m=n+\tilde{n}\geq 1$ (so that $X_k$ can be a single party).
    By \Cref{thm:partition} then, any $\mi\in\bs{\uJ}$ such that $\mi\in\pmi$ takes the form $\mi(\;\bigcup_{a\in A}\uI_a : \bigcup_{b\in B}\uI_b \;)$
    for some non-empty $A,B\subset [m]$ such that $A\cap B=\varnothing$ and $A\cup B=[m]$. We can then rewrite Eq.~\eqref{eq:pmi-equations} as follows
    \begin{align}
    \label{eq:combination1}
        0=\mi\left(\bigcup_{a\in A}\uI_a : \bigcup_{b\in B}\uI_b \right) = \; &
        \ent\left(\bigcup_{a\in A}\uI_a\right)+\ent \left(\bigcup_{b\in B}\uI_b\right) -\ent (\uJ) \nonumber\\
        & = \sum_{a\in A} \ent(\uI_a) + \sum_{b\in B} \ent(\uI_b) - \ent (\uJ) \nonumber\\
        & = \sum_{k\in [m]} \ent(\uI_k) -\ent (\uJ)
    \end{align}
    which is Eq.~\eqref{eq:entropy-relations}. The second step of Eq.~\eqref{eq:combination1} follows from the fact that for any subset of $\{\uI_a\}_{a\in A}$ the $\bset$ $\bs{\bigcup_{a\in A}\uI_a}$ is also in $\parv\cup\van$, since any $\mi'$ whose arguments are unions of subsets of $\{\uI_a\}_{a\in A}$ satisfies $\mi'\prec\mi$, and we can use the same decomposition as in the first line of Eq.~\eqref{eq:combination1} (and similarly for any subset of $\{\uI_b\}_{b\in B}$). 
    
    The converse follows from a very similar reasoning, essentially following the implications of Eq.~\eqref{eq:combination1} in the opposite direction.
\end{proof}

We stress that the similarities with graph models that we have highlighted are even more surprising considering that \Cref{thm:connectivity_of_H}, \Cref{cor:connected-components}, and \Cref{thm:entropy-relations} hold for arbitrary KC-PMIs, even the ones that violate SSA \cite{He:2022bmi}. 
Importantly, however, it should be clear that there is no contradiction with the fact that all holographic graph models, (as well as hypergraph models of \cite{Bao:2020zgx})
satisfy SSA, since these are different constructions from the correlation hypergraph representation of a KC-PMI. For example, the graph model that realizes the extreme ray of the SAC$_3$ whose KC-PMI is shown in Fig.~\ref{fig:PT3} is a star graph with four leaves corresponding to the various parties, and unit weight for each edge \cite{Bao:2015bfa}. On the other hand, the hypergraph representation of the same KC-PMI is a hypergraph with four vertices (instead of five) and five hyperedges corresponding to the $\bsets$ in the top two rows of Fig.~\ref{fig:PT3}. Naively one can imagine to assign weights to these hyperedges to construct a hypergraph model, which is obviously possible. However, it is straightforward to verify that there is no choice of weights such that the min-cut prescription to compute an entropy vector from this hypergraph model gives an entropy vector realizing the original KC-PMI. Instead, for KC-PMIs that can be realized by graph or hypergraph models, the correlation hypergraph should be used to extract information about the necessary topology of these models \cite{graph-construction}.

We conclude this subsection with two additional results about the structure of the hypergraph representation of a KC-PMI, for which we need to review a few more concepts from graph theory. Given an arbitrary set $A$ and a collection of subsets $\{A_1,A_2,\ldots,A_n\}$, the \textit{intersection graph} of this collection is a graph with $n$ vertices $\{v_1,v_2,\ldots,v_n\}$ and edges specified as follows: Any two distinct vertices $v_i,v_j$ are connected by an edge if and only if $A_i\cap A_j\neq \varnothing$. Given a hypergraph $\gf{H}$, its \textit{line graph}, denoted $\gf{L}$, is defined as the intersection graph of the collection of all its hyperedges. We will also need the notion of a \textit{clique} of a graph $\gf{G}$, which is defined as an induced sub-graph of $\gf{G}$ which is by itself a complete graph. We call a clique of $\gf{G}$ which is not contained in any other clique a \textit{max-clique}.

Given the correlation hypergraph $\hp$ of a KC-PMI $\pmi$, let $\lhp$ be its line graph, and $Q$ any max-clique of $\lhp$. 
In other words, the vertices of $\lhp$ correspond to subsystems $\uJ$ with positive $\bsets$, and $Q$ is any maximal subset of vertices of $\lhp$ which are pairwise-intersecting.\footnote{\,A vertex is an isolated vertex if it is not contained in any (hyper)edge, and it is by itself a max-clique.} For example, for the KC-PMI in Fig.~\ref{fig:N4-example}, $\lhp$ and its set of max-cliques are described in Fig.~\ref{fig:N4-example-LH}.

\begin{figure}[tbp]
    \centering
    \begin{tikzpicture}[scale=1]
    \draw (0,0) -- (-2,0);
    \draw (0,0) -- (0,2);
    \draw (0,0) -- (2,0);
    \draw (2,0) -- (4,0);
    \draw (6,0) -- (8,0);
    
    \filldraw[fill=Mcol] (0,0) circle (4pt);
    \filldraw[fill=PminusEcol] (-2,0) circle (4pt);
    \filldraw[fill=Mcol] (0,2) circle (4pt);
    \filldraw[fill=Mcol] (2,0) circle (4pt);
    \filldraw[fill=EminusMcol] (4,0) circle (4pt);
    \filldraw[fill=Mcol] (6,0) circle (4pt);
    \filldraw[fill=Mcol] (8,0) circle (4pt);
    
    \node[] () at (0,-0.5) {{\scriptsize $12$}};
    \node[] () at (-2,-0.5) {{\scriptsize $034$}};
    \node[] () at (0,2.5) {{\scriptsize $03$}};
    \node[] () at (2,-0.5) {{\scriptsize $04$}};
    \node[] () at (4,-0.5) {{\scriptsize $123$}};
    \node[] () at (6,-0.5) {{\scriptsize $01$}};
    \node[] () at (8,-0.5) {{\scriptsize $234$}};

    \node[] () at (0,-1) {{}};
    
    \end{tikzpicture}
    \caption{The ``non-trivial part'' of the complement graph of $\lhp$ for the KC-PMI in Fig.~\ref{fig:N4-example}, obtained by first taking the complement of $\lhp$ (i.e., the graph with same vertex set but complementary edge set), and then deleting each isolated vertex (i.e., each vertex which is connected to every other vertex in the original graph $\lhp$). Each vertex in the graph above corresponds to a hyperedge $e_{\uJ}$ of $\hp$ and is labeled by the corresponding $\uJ$, with the same color-coding as used in Fig.~\ref{fig:N4-example} for the corresponding $\bset$ classification. In total, $\lhp$ has six max-cliques which can conveniently written as $Q=\Omega\cup \widetilde{Q}$. Here $\Omega$ is the set of vertices which have been omitted in the graph above, and $\widetilde{Q}$ is any choice of a maximal set of vertices in the graph above such that, for each edge, at most one endpoint is in $\widetilde{Q}$. For example, for one of the max-cliques, we have $\widetilde{Q}=\{12,123,01\}$. As a preview of what we will discuss in \S\ref{sec:necessary}, notice that (for example) the vertices $\{04,123,01,234\}$ form a cycle in $\lhp$ without a chord, and that therefore $\lhp$ is \textit{not} a chordal graph (all these notions will be reviewed below). By the necessary condition that we will prove in \S\ref{sec:necessary}, this implies that the KC-PMI $\pmi$ cannot be realized by a simple tree graph model, and indeed the graph realization of $\pmi$ given in Fig.~\ref{fig:N4-example} is a tree, but it is not simple.}
    \label{fig:N4-example-LH}
\end{figure}

We denote by $Q^{\cap}$ the intersection of the hyperedges of $\hp$ which are elements of $Q$, i.e.,
\begin{equation}
    Q^{\cap}=\bigcap_{e_{\uJ}\in Q}\,e_{\uJ} \ .
\end{equation}
Notice that while each element of $Q$ is a vertex of $\lhp$, and therefore a hyperedge of $\hp$ corresponding to some subsystem, here it is viewed as a set of vertices of $\hp$, each one corresponding to a single party.
While we defined $Q^{\cap}$ for the correlation hypergraph representation of a KC-PMI, the same object could of course be defined for arbitrary hypergraphs, and it should be obvious that in that case its cardinality $|Q^{\cap}|$ can take any value. On the other hand, for the hypergraph representation of a KC-PMI, it turns out that $|Q^{\cap}|$ satisfies a very stringent upper-bound, and its specific value is intimately related to certain structural properties of the PMI. In fact, for the typical case of interest, such as the examples in the holographic context shown in Figs.~\ref{fig:PT3} and \ref{fig:N4-example}, there is no communal intersection, so $|Q^{\cap}|=0$. The general case is captured by the following result.

\begin{thm}
\label{thm:clique-intersection}
    For any \emph{KC-PMI} $\pmi$, and max-clique $Q$ of $\lhp$
    \begin{enumerate}[label={\emph{\footnotesize \roman*)}}]
        \item $|Q^{\cap}|\in\{0,1,2\}$,
        \item if $|Q^{\cap}|>0$ then $\subh{\nsp\setminus\ell}$ is disconnected for any $v_\ell\in Q^{\cap}$,
        \item $|Q^{\cap}|=2$ if and only if $\hp$ is disconnected and, letting $Q^{\cap}=\{v_{\ell_1},v_{\ell_2}\}$, one of its connected components is the 2-vertex connected graph with vertices $\{v_{\ell_1},v_{\ell_2}\}$.
    \end{enumerate}
\end{thm}
\begin{proof}
    ii) Consider an arbitrary KC-PMI $\pmi$ and its correlation hypergraph $\hp$. Suppose there is a max-clique $Q$ of $\lhp$ such that $|Q^{\cap}|>0$, and let $v_{\ell}$ be an element of $Q^{\cap}$. Let $\uK=\nsp\setminus\ell$ and notice that for any $\uJ$ such that $|\uJ|\geq 2$ we have $\uK\cap\uJ\neq\varnothing$.
    It must then be that $\bs{\uK}\notin\pos$, otherwise $\hp$ has a hyperedge $e_{\uK}$ which is in any max-clique of $\lhp$, but if $e_{\uK}\in Q$ then $Q^{\cap}$ does not contain $v_\ell$. 
    Therefore $\bs{\uK}$ is either vanishing or partial, and by \Cref{thm:connectivity_of_H} $\subh{\nsp}$ is disconnected. 

    iii) One direction is obvious: if $\hp$ is disconnected and one of its connected components is $\{v_{\ell_1},v_{\ell_2}\}$,
    then the only hyperedge $e_{\uJ}$ such that $\{\ell_1,\ell_2\}\cap\uJ\neq\varnothing$ is $e_{\uJ}=\{v_{\ell_1},v_{\ell_2}\}$, implying that $Q=\{e_{\uJ}\}$ is a max-clique of $\lhp$, and $|Q^{\cap}|=2$.
    
    In the opposite direction, suppose there is a max-clique $Q$ of $\lhp$ such that $|Q^{\cap}|=2$, and let $Q^{\cap}=\{v_{\ell_1},v_{\ell_2}\}$. Notice that this requires that there is at least one positive $\bset$ $\bs{\uJ}$ such that $\{\ell_1,\ell_2\}\subseteq\uJ$.
    Furthermore, for any $\uJ$ such that $\ell_1\in\uJ$ and $\ell_2\notin\uJ$, it must be that $\bs{\uJ}\notin\pos$, otherwise $e_{\uJ}$ would belong to $Q$ (since any $e_{\uK}\in Q$ contains $v_{\ell_1}$ by the assumption about $Q^{\cap}$), but for any $e_{\uK}\in Q$ we would have $v_{\ell_2}\notin e_{\uJ}\cap e_{\uK}$ contradicting the assumption about $Q^{\cap}$. By the same argument, swapping $\ell_1$ and $\ell_2$, it also follows that $\bs{\uJ}\notin\pos$ for any $\uJ$ such that $\ell_1\notin\uJ$ and $\ell_2\in\uJ$. Therefore, all $\bsets$ $\bs{\uJ}$ such that $\uJ$ contains $\ell_1$ or $\ell_2$ but not both are either vanishing or partial. 

    Let us first consider the simpler case where they are all vanishing, and consider in particular $\bs{\uJ}$ for $\uJ=\nsp\setminus\ell_2$. Since $\bs{\uJ}\in\van$ we have $\mi(\ell_1:\nsp\setminus\{\ell_1,\ell_2\})\in\pmi$, and by replacing $\ell_1$ with $\ell_2$ in $\uJ$ we also have $\mi(\ell_2:\nsp\setminus\{\ell_1,\ell_2\})\in\pmi$. By the linear dependence among these MI instances it then follows that $\mi(\{\ell_1,\ell_2\}:\comp{\{\ell_1,\ell_2\}})\in\pmi$. Therefore $\hp$ is disconnected and there are two possibilities: Either $v_{\ell_1},v_{\ell_2}$ are both isolated vertices, or $\{v_{\ell_1},v_{\ell_2}\}$ is a connected component of $\hp$. In the first case, there is no $\bs{\uJ}\in\pos$ such that $\ell_1\in\uJ$ or $\ell_2\in\uJ$, and $Q^{\cap}=\varnothing$. So the only option is that $\bs{\ell_1\ell_2}\in\pos$ and $Q^{\cap}=\{v_{\ell_1},v_{\ell_2}\}$.

    Suppose now that there is at least one partial $\bs{\uJ}$ where $\uJ$ contains $\ell_1$ or $\ell_2$ but not both. Consider some $\bs{\uK}\in\text{Max}(\posj)$. Since by definition any such $\bs{\uK}$ is positive, it must be that $\ell_1\notin\uK$ (by the same argument used above, because $\ell_2\notin\uK$). This implies that $\ell_1$ is an element of the partition of $\uJ$ given by Eq.~\eqref{eq:partition}, and by \Cref{thm:partition} we have $\mi(\ell_1:\nsp\setminus\{\ell_1,\ell_2\})\in\pmi$. Swapping $\ell_1$ and $\ell_2$ we then also get $\mi(\ell_2:\nsp\setminus\{\ell_1,\ell_2\})\in\pmi$ and the rest of the proof is identical to the previous case. 
    
    i) For sufficiently large $\N$, suppose that for some KC-PMI there is a max-clique $Q$ in $\lhp$ such that $|Q^{\cap}|=r\geq 3$, and let $Q^{\cap}=\{v_{\ell_1},v_{\ell_2},\ldots,v_{\ell_r}\}$. Consider the first two elements of $Q^{\cap}$, it follows by (ii) that $\hp$ is disconnected and the only hyperedge that contains $v_{\ell_1}$ and $v_{\ell_2}$ is $e_{\uJ}$ with $\uJ=\{\ell_1\ell_2\}$, but in this case $Q^{\cap}$ does not contain any other element. Finally, to see that we can indeed have $|Q^{\cap}|=0$ it is sufficient to provide an example, which is given by the KC-PMI illustrated in Fig.~\ref{fig:PT3}.
\end{proof}

We stress that \Cref{thm:clique-intersection} also implies that the cases where $|Q^{\cap}|\in\{1,2\}$ are rather exceptional, since they require that at least one instance of the mutual information $\mi(\uJ:\uK)$ with $|\uJ|+|\uK|=\N$ vanishes; an example is shown in Fig.~\ref{fig:exmplae-clique-theorems}. In fact, if one is interested in ``genuine'' $\N$-party KC-PMIs, these cases can typically be ignored as they can be understood in terms of lifts of KC-PMIs for fewer parties \cite[Thm.6]{He:2022bmi}.\footnote{\,For the example in Fig.~\ref{fig:exmplae-clique-theorems}, the KC-PMI is realized by a product of three Bell pairs corresponding to the edges of the graph model, i.e.,  $\ket{\psi}_{12_1}\ket{\psi}_{02_2}\ket{\psi}_{32_3}$, where party 2 is a coarse-grained subsystem whose Hilbert space is the tensor product of those for parties $2_1,2_2,2_3$.} For the vast majority of KC-PMIs at any given $\N$, and in fact for all interesting cases, one effectively has $|Q^{\cap}|=0$.

\begin{figure}[tbp]
    \centering
    \begin{subfigure}{0.29\textwidth}
    \centering
    \begin{tikzpicture}[scale=1]
    \draw (0,0) -- (1,0);
    \draw (1,0) -- (1,1);
    \draw (1,0) -- (2,0);
    
    \filldraw (0,0) circle (2pt);
    \filldraw (1,0) circle (2pt);
    \filldraw (1,1) circle (2pt);
    \filldraw (2,0) circle (2pt);
    
    \node[] () at (1,-0.4) {{\scriptsize $2$}};
    \node[] () at (-0.3,0) {{\scriptsize $1$}};
    \node[] () at (1,1.3) {{\scriptsize $3$}};
    \node[] () at (2.4,0) {{\scriptsize $0$}};
    \node[] () at (1,-1) {};

    \node[red] () at (0.5,0.2) {{\scriptsize $1$}};
    \node[red] () at (1.15,0.5) {{\scriptsize $1$}};
    \node[red] () at (1.5,0.2) {{\scriptsize $1$}};
    \end{tikzpicture}
    \subcaption[]{}
    \end{subfigure}
    \begin{subfigure}{0.69\textwidth}
    \centering
    \begin{tikzpicture}[scale=0.5]
    \node[fill=PminusEcol, align=center, rounded corners] (0123) at (0,4) {{\footnotesize $\bs{0123}$}};
    \node[fill=PminusEcol, align=center, rounded corners] (012) at (-6,2) {{\footnotesize $\bs{012}$}};
    \node[fill=Vcol, align=center, rounded corners] (013) at (-2,2) {{\footnotesize $\bs{013}$}};
    \node[fill=PminusEcol, align=center, rounded corners] (023) at (2,2) {{\footnotesize $\bs{023}$}};
    \node[fill=PminusEcol, align=center, rounded corners] (123) at (6,2) {{\footnotesize $\bs{123}$}};
    \node[fill=Vcol, align=center, rounded corners] (01) at (-7.5,0) {{\footnotesize $\bs{01}$}};
    \node[fill=Mcol, align=center, rounded corners] (02) at (-4.5,0) {{\footnotesize $\bs{02}$}};
    \node[fill=Vcol, align=center, rounded corners] (03) at (-1.5,0) {{\footnotesize $\bs{03}$}};
    \node[fill=Mcol, align=center, rounded corners] (12) at (1.5,0) {{\footnotesize $\bs{12}$}};
    \node[fill=Vcol, align=center, rounded corners] (13) at (4.5,0) {{\footnotesize $\bs{13}$}};
    \node[fill=Mcol, align=center, rounded corners] (23) at (7.5,0) {{\footnotesize $\bs{23}$}};

    \node[red] () at (0,-1) {};
    \end{tikzpicture}
    \subcaption[]{}
    \end{subfigure}
    \par\bigskip\bigskip\bigskip
    \begin{subfigure}{0.49\textwidth}
    \centering
    \begin{tikzpicture}
    \draw[PminusEcol, rounded corners, very thick] (-1.5,-1.6) rectangle (4.5,2.8);
    \draw[PminusEcol, rounded corners, very thick] (-0.65,-0.7) rectangle (1.9,2.2);
    \draw[PminusEcol, rounded corners, very thick] (0.95,-0.9) rectangle (3.65,2.4);
    \draw[PminusEcol, rounded corners, very thick] (-1,-1.2) rectangle (4,0.5);
    
    \draw[Mcol!80!black, very thick] (0,0) -- (1.5,0);
    \draw[Mcol!80!black, very thick] (1.5,0) -- (1.5,1.5);
    \draw[Mcol!80!black, very thick] (1.5,0) -- (3,0);
    
    \filldraw (0,0) circle (2pt);
    \filldraw (1.5,0) circle (2pt);
    \filldraw (1.5,1.5) circle (2pt);
    \filldraw (3,0) circle (2pt);
    
    \node[] () at (1.5,-0.35) {{\scriptsize $v_2$}};
    \node[] () at (-0.35,0) {{\scriptsize $v_1$}};
    \node[] () at (1.5,1.8) {{\scriptsize $v_3$}};
    \node[] () at (3.4,0) {{\scriptsize $v_0$}};
    \node[] () at (1,-1) {};

    \node[] () at (0.7,0.2) {{\scriptsize $e_{12}$}};
    \node[] () at (1.25,0.75) {{\scriptsize $e_{23}$}};
    \node[] () at (2.25,0.2) {{\scriptsize $e_{02}$}};
    \node[] () at (-0.25,1.9) {{\scriptsize $e_{123}$}};
    \node[] () at (-0.65,-1) {{\scriptsize $e_{012}$}};
    \node[] () at (3.2,2.2) {{\scriptsize $e_{023}$}};
    \node[] () at (-1,2.5) {{\scriptsize $e_{0123}$}};
    \end{tikzpicture}
    \subcaption[]{}
    \end{subfigure}
    \begin{subfigure}{0.49\textwidth}
    \centering
    \begin{tikzpicture}
  \def\n{7} 
  \def\radius{2} 

  \newcommand{\getcolor}[1]{%
    \ifcase#1%
      \or PminusEcol%
      \or PminusEcol%
      \or PminusEcol%
      \or PminusEcol%
      \or Mcol%
      \or Mcol%
      \or Mcol%
    \fi%
  }
    
  \foreach \i in {1,...,\n} {
    \coordinate (P\i) at ({\radius*cos(360/\n*(\i-1))}, {\radius*sin(360/\n*(\i-1))});
  }

  \foreach \i in {1,...,\n} {
    \foreach \j in {1,...,\n} {
      \ifnum\i<\j
        \draw (P\i) -- (P\j);
      \fi
    }
  }

  \foreach \i in {1,...,\n} {
    \filldraw[fill=\getcolor{\i}, draw=black] (P\i) circle (3pt);
  }

    \node[] () at (-2.2,-1.3) {{\scriptsize $\{v_0,v_2\}$}};
    \node[] () at (-0.5,-2.3) {{\scriptsize $\{v_1,v_2\}$}};
    \node[] () at (1.9,-1.8) {{\scriptsize $\{v_2,v_3\}$}};

    \node[] () at (2.2,1.7) {{\scriptsize $\{v_0,v_1,v_2\}$}};
    \node[] () at (-2.5,1.2) {{\scriptsize $\{v_0,v_2,v_3\}$}};
    \node[] () at (3,0) {{\scriptsize $\{v_1,v_2,v_3\}$}};

    \node[] () at (-0.4,2.3) {{\scriptsize $\{v_0,v_1,v_2,v_3\}$}};
\end{tikzpicture}
    \subcaption[]{}
    \end{subfigure}
    \caption{A simple example of the special case of a KC-PMI with $|Q^\cap|=1$. (a) an $\N=3$ simple tree graph model $\gf{G}$ where a boundary vertex (2) is not a leaf. (b) the classification of the $\bsets$ for the KC-PMI $\pmi$ of $\gf{G}$. (c) the correlation hypergraph $\hp$ of $\pmi$. (d) the line graph $\lhp$ of $\hp$. In this simple example $\lhp$ is a complete graph and therefore it has a single max-clique $Q$ (the graph itself), for which $Q^\cap=\{v_2\}$. In agreement with \Cref{thm:clique-intersection}, the sub-hypergraph $\subh{013}$ is disconnected, as it is the trivial hypergraph with isolated vertices $\{v_0,v_1,v_3\}$, corresponding to the fact that $\bs{013}$ is vanishing. Notice that an analogous statement is also true for the graph model in (a), in the sense that $\gf{G}_{013}$ is completely disconnected.}
    \label{fig:exmplae-clique-theorems}
\end{figure}

Further insights into the structure of a KC-PMI can be gained by looking at another object closely related to $Q^\cap$. Given a KC-PMI, its hypergraph representation $\hp$, and a party $\ell\in\nsp$, we define the $\ell$-clique $ Q^{\ell}$ of $\lhp$ as follows
\begin{equation}
    Q^{\ell}=\{e_{\uJ},\; \ell\in\uJ\}.
\end{equation}
Clearly this set is a clique of $\lhp$,\footnote{\,As we will see momentarily it can be $Q^{\ell}=\varnothing$, but even this is a clique, despite being trivial.} although in general it is not necessarily a max-clique. The next result relates the inclusion of this particular clique into max-cliques, to the connectivity of the vertex $v_\ell$ in $\hp$.

\begin{thm}
\label{thm:ell-clique}
    For any \emph{KC-PMI} $\pmi$, and any party $\ell$, either $v_\ell$ is an isolated vertex of $\hp$ and $Q^{\ell}=\varnothing$, or $Q^{\ell}$ is contained in exactly one max-clique of $\lhp$.
\end{thm}
\begin{proof}
    For an arbitrary KC-PMI $\pmi$, consider a party $\ell$ and notice that it follows immediately from the definitions of $Q^\ell$  and $\hp$ that $Q^\ell=\varnothing$ if and only if $v_\ell$ is an isolated vertex. Suppose then that $Q^\ell$ is non-empty, and that it is not a max-clique (otherwise the statement is trivially true since a max-clique is not contained in any other clique). Let $Q$ be an arbitrary max-clique that contains $Q^\ell$ (since $Q^\ell$ is not maximal there is always at least one such max-clique), and notice that for any element $e_{\uJ}\in (Q\setminus Q^\ell)$ it must be that $\ell\notin\uJ$, otherwise $e_{\uJ}$ would be in $Q^\ell$. 
  
    Suppose now that there are two distinct max-cliques $Q_1,Q_2$ that contain $Q^\ell$, and consider some $e_{\uJ_1}\in Q_1\setminus Q_2$ and $e_{\uJ_2}\in Q_2\setminus Q_1$. Notice that such $e_{\uJ_1},e_{\uJ_2}$ always exist by the fact that $Q_1,Q_2$ are assumed to be distinct and maximal, so neither of them can contain the other. Furthermore, also notice that it is always possible to choose $e_{\uJ_1},e_{\uJ_2}$ such that $\uJ_1\cap\uJ_2=\varnothing$, since if for a chosen $e_{\uJ_1}$ we have $\uJ_1\cap\uJ_2\neq\varnothing$ for each possible choice of $e_{\uJ_2}$, then $e_{\uJ_2}\in Q_1$ and $Q_1$ is not maximal.

    Let us assume that we have made such a choice for $e_{\uJ_1},e_{\uJ_2}$. Since $\ell\notin\uJ_1$, it must be that $\bs{\comp{\uJ_1}}\notin\pos$, otherwise $\bs{\comp{\uJ_1}}\in Q^\ell$ (since $\ell\in\comp{\uJ_1}$), and $Q_1$ would not contain $Q^\ell$ because $\uJ_1\cap\comp{\uJ_1}=\varnothing$. This implies that $\bs{\comp{\uJ_1}}$ is either partial or vanishing. However, notice that from $\uJ_1\cap\uJ_2=\varnothing$ we have $\uJ_2\subseteq\comp{\uJ_1}$ and by assumption $\bs{\uJ_2}\in\pos$ (since $e_{\uJ_2}\in\lhp$). Therefore it must be that $\bs{\comp{\uJ_1}}\in\parv$ because $\bs{\uJ_2}\preceq\bs{\comp{\uJ_1}}$, and by (i) of \Cref{thm:max-structure} it follows that $\bs{\comp{\uJ_1}}$ cannot be vanishing.

    Let then $\uK=\nsp\setminus(\uJ_1\cup\uJ_2\cup\ell)$ and rewrite $\comp{\uJ_1}$ as $\comp{\uJ_1}=\ell\cup\uJ_2\cup\uK$. Consider the partition of $\comp{\uJ_1}$ given by Eq.~\eqref{eq:partition} and notice that any non-singleton component cannot contain $\ell$, since otherwise the corresponding $\bset$ would be positive and it would be in $Q^\ell$ but not in $Q_1$ (since it has empty intersection with $\uJ_1$), contradicting the assumption that $Q_1$ contains $Q^\ell$. Therefore by \Cref{thm:partition} it must be that $\mi(\ell:\uJ_2\cup\uK)\in\pmi$.
    
    Consider now $\uJ'_1=\uJ_1\cup\ell$, and suppose that $\bs{\uJ'_1}\in\pos$. Then $\bs{\uJ'_1}\in Q^\ell$ (by the definition of $Q^\ell$) and $\bs{\uJ'_1}\in Q_1$ (by the fact that $\bs{\uJ_1}\in Q_1$ and any $\uK$ which has non-empty intersection with $\uJ_1$ has non-empty intersection with $\uJ'_1$) consistently with the assumption that $Q^\ell\subset Q_1$. However, since $\uJ_2\cap\uJ_1=\varnothing$ and $\ell\notin\uJ_2$, it follows that $\uJ_2\cap\uJ'_1=\varnothing$ and $\bs{\uJ'_1}\notin Q_2$. This contradicts the assumption that $Q^\ell\subset Q_2$ and it implies that $\bs{\uJ'_1}$ is either vanishing or partial. It must then be that $\bs{\uJ'_1}\in\parv$, since $\uJ_1\subset\uJ'_1$ and $\bs{\uJ_1}$ is positive. As before, any non-trivial component of the partition of $\uJ'_1$ given by Eq.~\eqref{eq:partition} of $\uJ'_1$ cannot contain $\ell$, implying again by \Cref{thm:partition} that $\mi(\ell:\uJ_1)\in\pmi$. 
    
    Finally, by linear dependence, $\mi(\ell:\uJ_2\cup\uK)\in\pmi$ and $\mi(\ell:\uJ_1)\in\pmi$ imply $\mi(\ell:\nsp\setminus\ell)\in\pmi$ and by \Cref{cor:connected-components} $v_{\ell}$ is an isolated vertex. As shown above, this is only possible if $Q^\ell$ is empty, implying that if $Q^\ell$ is non-empty it is contained in exactly one max-clique of $\lhp$.
\end{proof}

We leave it as an exercise for the reader to show that for any KC-PMI, and any two distinct parties $\ell_1$ and $\ell_2$, the cliques $Q^{\ell_1}$ and $Q^{\ell_2}$ of $\lhp$ coincide if and only if $Q^{\cap}=\{v_{\ell_1},v_{\ell_2}\}$.\footnote{\,Hint: consider the partitions $\Gamma(\nsp\setminus\ell_1)$ and $\Gamma(\nsp\setminus\ell_2)$.}

We expect \Cref{thm:clique-intersection} and \Cref{thm:ell-clique} to play a central role in a systematic approach to constructing holographic graph models that realize a given KC-PMI \cite{graph-construction}. These results in fact arose in the course of developing such a procedure, but we include them here, as—like many of the properties of graph models discussed in this work—they extend naturally to arbitrary KC-PMIs.
In particular, given a KC-PMI~$\pmi$ that satisfies the necessary condition for realizability by a simple tree~$\gf{T}$, derived in the following subsection, we expect that \Cref{thm:clique-intersection} may be used to determine whether $\gf{T}$ has boundary vertices that are not leaves, as illustrated in Fig.~\ref{fig:exmplae-clique-theorems}. Barring this possibility, we expect all boundary vertices to be leaves, with internal vertices corresponding to the max-cliques of $\lhp$. We then expect \Cref{thm:ell-clique} to determine the internal vertex to which the leaf associated to a given party should be connected, thereby completing the full graph construction.

\section{Necessary condition for realizability by simple trees}
\label{sec:necessary}

We are now ready to derive a necessary condition for the realizability of a KC-PMI by a simple tree graph model. We begin with an observation that is simply a rewriting of a result that we have already presented in \S\ref{subsec:graph_review}.

\begin{cor}
\label{cor:ew_connectivity-v2}
    For any simple graph model $\gf{G}$ and subsystem $\uJ$, if $\bs{\uJ}$ is positive then the subgraph $\gf{G}_\uJ$ induced by the minimal min-cut $\mmC_{\uJ}$ for $\uJ$ is connected.\footnote{\,While we phrased this lemma for the minimal min-cut, it is clear that it holds more generally for any min-cut.}
\end{cor}
\begin{proof}
   The statement is a translation of the special case of \Cref{cor:ew-decomposition} for $m=1$ in the language of $\bsets$.
\end{proof}

The next lemma is again general for arbitrary simple graph models (not necessarily trees) and it establishes a relation between certain topological data about $\gf{G}$ and the hypergraph representation of the KC-PMI realized by $\gf{G}$. 

\begin{lemma}
\label{lem:isomorphic_LH}
    For any \emph{KC-PMI} $\pmi$ that can be realized by a holographic graph model, any $\vec{\ent}_*\in\text{\emph{int}}(\face_{\pmi})$, and any simple graph model $\gf{G}$ such that $\vec{\ent}(\gf{G})=\vec{\ent}_*$, let
    \begin{equation}
    \label{eq:sigma}
        \Sigma=\{\mmC_{\uJ},\; \forall\,\uJ\;\; \text{\emph{such that}}\;\; \bs{\uJ}\in\pos\}.
    \end{equation}
    Then the intersection graph of $\Sigma$ is isomorphic to the line graph $\lhp$ of the correlation hypergraph of $\pmi$.
\end{lemma}
\begin{proof}
    This follows immediately from the definition of $\Sigma$ in Eq.~\eqref{eq:sigma}, the definition of $\lhp$, and (iii) in \Cref{thm:minimal-min-cuts}.
\end{proof}

In the special case where $\gf{G}$ is a simple tree $\gf{T}$, \Cref{lem:isomorphic_LH} implies that the line graph $\lhp$ of the KC-PMI realized by $\gf{T}$ is isomorphic to the intersection graph of a family of subtrees of $\gf{T}$. Graphs of this type are sometimes called \textit{subtree graphs} and it is a well known fact in graph theory that they are completely characterized by the following result \cite{book:intersection_graphs}.

\begin{thm}[Buneman, Gavril, and Walter]
\label{thm:subtree}
     A graph is a subtree graph if and only if it is a chordal graph.\footnote{\,A graph is said to be \textit{chordal} if all cycles of four or more vertices have a chord.}
\end{thm}
\begin{proof}
    See for example \cite[Theorem 2.4]{book:intersection_graphs}.
\end{proof}

By combining \Cref{cor:ew_connectivity-v2}, \Cref{lem:isomorphic_LH} and \Cref{thm:subtree} we can easily derive the following necessary condition for the realizability of an entropy vector by a simple tree holographic graph model.

\begin{thm}
\label{thm:necessary-condition}
    An entropy vector $\vec{\ent}_*$ can be realized by a simple tree graph model only if its \emph{PMI} $\pmi$ is a \emph{KC-PMI} and the line graph $\lhp$ of the correlation hypergraph of $\pmi$ is a chordal graph.
\end{thm}
\begin{proof}
    Consider an entropy vector $\vec{\ent}_*$ whose PMI $\pmi$ is a KC-PMI (otherwise the proof is trivial because $\vec{\ent}_*$ violates at least one instance of SSA and it cannot be realized by any graph model). By \Cref{lem:isomorphic_LH}, the intersection graph of $\Sigma$ is isomorphic to the line graph $\lhp$ of the hypergraph representation $\hp$ of $\pmi$. Suppose now that there exists a simple tree graph model $\gf{T}$ such that $\vec{\ent}(\gf{T})=\vec{\ent}_*$. By \Cref{cor:ew_connectivity-v2}, any element of $\Sigma$ is connected, and it is then a subtree of $\gf{T}$. Therefore $\lhp$ is isomorphic to the intersection graph of a collection of subtrees of a tree, and by \Cref{thm:subtree} it is a chordal graph.
\end{proof}

Given this necessary condition, one may wonder about how efficiently it can be verified. Given an arbitrary entropy vector $\vec{\ent}\in\mathbb{R}^{\D}_+$ which belongs to the SAC, it is straightforward to check which instances of the mutual information vanish, and therefore determine the corresponding PMI $\pmi$ and its complement $\cpmi$. If $\vec{\ent}$ is chosen arbitrarily, $\pmi$ is not necessarily a KC-PMI. The KC condition can be verified by checking whether $\pmi$ is a down-set, but depending on the situation, in practice it might be faster to simply verify whether $\vec{\ent}$ satisfies all instances of SSA.\footnote{\,While in general SSA is a stronger condition than KC, it is anyway a necessary condition for the realizability of an entropy vector by any graph model.} Assuming that this is the case, it is then straightforward to determine the positive $\bsets$, since this just amounts to ascertaining which $\bsets$ do not contain any element of $\pmi$. For an arbitrary KC-PMI, the number of positive $\bsets$ can be large, and $\lhp$ has a very complicated structure. Nevertheless, it is a well known fact in graph theory that verifying whether a graph is chordal can be done in linear time using the lexicographic breadth-first search algorithm \cite{lex-breadth-first}.

Suppose now that for a given entropy vector $\vec{\ent}$ we have checked that it satisfies SSA, so that its PMI $\pmi$ is a KC-PMI, and that $\lhp$ \textit{not} chordal. Does this mean that $\vec{\ent}$ cannot be realized by \textit{any} graph model? To see that this is not the case, we can now make precise a comment that we made at the end of \S\ref{subsec:graph_review}, and show that for any $\N\geq 3$, any entropy vector in the interior of the $\N$-party SAC cannot be realized by a simple tree. Indeed, notice that for any such entropy vector all $\bsets$ are positive. The $\bsets$ $\bs{01},\bs{12},\bs{23},\bs{30}$ then correspond to vertices of $\lhp$, and form a cycle without a chord. 
On the other hand, it is trivial to construct, for any $\N$, entropy vectors in the interior of the $\N$-party SAC that can be realized by graph models. An example is the entropy vector of any graph model obtained from a complete graph with $\N+1$ vertices where $\partial V=V$, for an arbitrary choice of weights.

Finally, let us briefly comment on the sufficiency of the condition in \Cref{thm:necessary-condition}. Suppose that we are given an entropy vector $\vec{\ent}$ whose PMI is a KC-PMI, and such that $\lhp$ \textit{is} chordal. Does this imply that $\vec{\ent}$ is realizable by a simple tree graph model? While we currently have no proof, we have tested this possibility in several examples and suspect that this might well be the case. Notice that \Cref{thm:subtree} is an ``if and only if'' statement, implying that if $\lhp$ is chordal, then it is the intersection graph of a collection of subtrees of a tree. It is not hard to see \cite{graph-construction} that one can then construct a tree that accommodates a choice of cuts for the various subsystems whose induced subgraphs have the correct connectivity properties given in \Cref{thm:subgraph-connectivity} and \Cref{cor:ew-decomposition}.\footnote{\,A simple situation where (naively) this appears to be false is the case where $\vec{\ent}$ has some vanishing components. In this case, any graph model that reproduces $\vec{\ent}$ must be disconnected, and $\vec{\ent}$ cannot be realized by a tree even if $\lhp$ is chordal. However, this is not a true obstruction, as one can simply set to 0 some of the weights on the tree constructed from the data in $\lhp$, or equivalently delete the corresponding edges, to obtain a forest that realizes $\vec{\ent}$.} What remains to be shown is that there exists a choice of edge weights such that these cuts are indeed min-cuts, and their costs correctly reproduce the components of $\vec{\ent}$ \cite{2025wip}.

\section{Discussion}
\label{sec:discussion}

The main goal of this work was to show that many structural properties of holographically realizable entropy configurations, that can be deduced from the analysis of the connectivity of entanglement wedges, completely generalize to arbitrary quantum states. The key tool for our derivations is  the concept of $\bsets$, and we have shown that in this language the description of KC-PMIs takes on a remarkably nice form. In particular,  when reformulated in the language of the correlation hypergraph, several similarities with well known properties of holographic graph models \cite{Bao:2015bfa}, and more generally hypergraph models \cite{Bao:2020zgx}, become completely manifest. One immediate application of this technology was the derivation of a necessary condition (\Cref{thm:necessary-condition}) for the realizability of an entropy vector by a simple tree graph model \cite{Hernandez-Cuenca:2022pst}, which can be tested efficiently. We conclude by discussing several possible directions for future work.

\paragraph{Correlation ``schemes'':} It is important to note that, given an arbitrary number of parties $\N$, and an arbitrary hypergraph $\gf{H}$ with $\N+1$ vertices, $\gf{H}$ does not necessarily correspond to the correlation hypergraph
of a KC-PMI. In other words, there does not necessarily exist an $\N$-party KC-PMI $\pmi$ such that $\gf{H} = \hp$. This is essentially because an arbitrary choice of instances of the mutual information is not necessarily closed under linear dependence. We have already seen an example in Fig.~\ref{fig:EsubsetP-not-suff}, where we presented a down-set $\ds$ in the $\N=3$ MI-poset which is not a PMI, and for which the reconstruction formula Eq.~\eqref{eq:recovery2} gives the correct $\cds$. Notice that $\cds$ does admit a description in terms of a hypergraph according to \Cref{defi:hypergarph_rep}, specifically the hypergraph with 4 vertices and a single hyperedge that contains all of them. 

On the other hand, it is also true that an arbitrary hypergraph does not, in general, respect structural results such as those in \Cref{thm:connectivity_of_H}, \Cref{thm:clique-intersection}, and \Cref{thm:ell-clique}. As a simple example of a violation\footnote{\,Here and below use the term ``violation'' as a shorthand to indicate where the main statement of the theorem would not hold in absence of the KC-PMI assumption.  In other words, it is not a violation of the theorem as stated but rather of a putative generalization which relaxes the assumption.} of \Cref{thm:connectivity_of_H}, consider the hypergraph with three vertices $\{v_0,v_1,v_2\}$ and hyperedges $e_{01},e_{02}$. If we construct an up-set $\cds$ such that the $\bsets$ $\bs{01},\bs{02}$ are positive, then by \Cref{lem:join} it follows that $\bs{012}$ is necessarily positive, and should have been a hyperedge. These violations of \Cref{thm:connectivity_of_H} however, can easily be fixed by taking an appropriate ``closure'' of the set of hyperedges \cite{2025wip}. An example of violation of \Cref{thm:clique-intersection} is again provided by the down-set in Fig.~\ref{fig:EsubsetP-not-suff}, since the corresponding line graph has a single vertex and $|Q^\cap|=4$. Finally, as an example of violation of \Cref{thm:ell-clique}, we leave it as an exercise for the reader to verify that, given the following collection of $\bsets$ $\xx= \{\bs{123},\bs{24},\bs{30}\}$, the complementary down-set $\ds$ to the up-set $\cds=\,\Uparrow \!\xx$  is not a PMI. Furthermore, the set of positive $\bsets$ of $\cds$ specifies an hypergraph whose line graph has two max-cliques, and they both contain the clique $Q_1$. Notice that this example does not violate \Cref{thm:clique-intersection}, since for both max-cliques, $|Q^\cap|=0$.

Considering these examples, one can therefore imagine defining a ``correlation scheme'' as an arbitrary hypergraph that satisfies all the properties derived in this work. Although the structure of the space of such objects is not immediately clear, it may possess additional properties not described here. At the very least, it is reasonable to expect that this analysis could provide new intuition for deriving a better ``approximation" of the set of all KC-PMIs \cite{He:2022bmi}, in the same way that $\bsets$ improve upon the previous characterization based solely on down-sets. In particular, it builds in the key part of the linear dependence among MI instances through \Cref{thm:positiveness}.

\paragraph{Algorithmic derivation of KC-PMIs:} As the number of parties $\N$ grows, the explicit derivation of KC-PMIs quickly becomes a challenging computational problem due to the large dimension of entropy space $\D = 2^{\N} - 1$ and the number of inequalities $\genfrac\{\}{0pt}{1}{\N+2}{3}$. An algorithm to find all 1-dimensional KC-PMIs at any given $\N$ was developed in \cite{He:2024xzq} and was shown to be efficient enough to solve this problem at least up to $\N=6$. While this algorithm can be adapted to also compute higher-dimensional PMIs, this computation is even more challenging due to their large number. The explicit computation of KC-PMIs can be of great utility to investigate some of the questions outlined below, and it is reasonable to hope that a formulation in terms of $\bsets$ might make the algorithm in \cite{He:2024xzq} more efficient (for example via a reduction of the search space to the correlation schemes mentioned above, which automatically implement at least some of the linear dependences among the MI instances).

\paragraph{SSA-compatible PMIs:} It was shown in \cite{He:2022bmi} that for $\N \geq 4$, not all KC-PMIs are SSA-compatible, and it was later found in \cite{He:2024xzq} that starting from $\N \geq 6$, this holds even for 1-dimensional KC-PMIs. However, we stress that all the technology and results developed in this work apply to arbitrary KC-PMIs, regardless of whether they are SSA-compatible. This is quite surprising, given the numerous similarities between the correlation hypergraph representation of a KC-PMI introduced here and the graph or hypergraph models of \cite{Bao:2015bfa} and \cite{Bao:2020zgx}, which always yield entropy vectors obeying strong subadditivity (SSA).

Since SSA is clearly a necessary condition for the realizability of a KC-PMI by a quantum state, it is important to better understand how to characterize SSA-compatible KC-PMIs. Given the relative simplicity of the $\N=4$ case, the full set of KC-PMIs can be computed explicitly \cite{He:2022bmi} and might serve as a useful testing ground to analyze the structure of the correlation hypergraph representation in these cases. In particular, it would be important to determine whether it is possible to identify some structural property of the correlation hypergraph representation that, at least for this low party number, distinguishes SSA-compatible KC-PMIs from SSA-violating ones. If this is not possible, it may be necessary to refine the correlation hypergraph representation so that it encodes additional information about the PMI, such as the specific instances of SSA that are saturated.

\paragraph{Hypergraph models and stabilizer states:} The entropy vectors realized by the hypergraph models of \cite{Bao:2020zgx} were shown in \cite{Walter:2020zvt} to belong to the entropy cone of stabilizer states. Although this inclusion was proven to be strict in \cite{Bao:2020mqq}, it remains possible that all KC-PMIs realizable by quantum states can also be realized by hypergraph models. It is also still a possibility that these KC-PMIs are, in fact, simply the SSA-compatible ones. Given the similarities we have already mentioned, the correlation hypergraph representation of a KC-PMI introduced in this work appears to be the right starting point to explore these possibilities in greater detail.

In particular, a key difference between the correlation hypergraph representation of a KC-PMI and hypergraph models is that the former includes only boundary vertices. A natural starting point would therefore be a generalized version of \Cref{cor:ew_connectivity-v2} for hypergraphs, to establish a relation between connected cuts on a hypergraph model and the line graph of the correlation hypergraph of the underlying PMI. This step is crucial for obtaining a model in which the entropies are determined by the solution to a minimization problem. In fact, this is also the case for the even more restricted case of holographic models, to which we now turn.

\paragraph{Holographic graph models:} 
Our initial motivation for studying $\bsets$ was to investigate the necessary and sufficient conditions for the realizability of an entropy vector by a graph model, and in particular, to explore, for realizable entropy vectors, the dependence of the topology of a graph realization on the structure of the PMI. 

As we already mentioned at the end of \S\ref{sec:necessary}, while in this work we have only proven a necessary condition for realizability by simple trees, we believe that this condition is also sufficient. In other words, we conjecture that any KC-PMI $\pmi$ for which $\lhp$ is chordal — henceforth referred to as \textit{chordal} PMIs — can be realized by a simple tree graph model, and we hope to report on this soon \cite{2025wip}. If this is true, it would imply, via \cite{hayden2016holographic}, that all chordal PMIs can be realized by quantum states. This would provide, for arbitrary $\N$, a characterization of part of the boundary of the quantum entropy cone.

We have already seen that even if this result holds, the necessary condition we derived is certainly not necessary for the realizability of KC-PMIs by more general graphs. For $\N=6$, \cite{He:2024xzq,comments} provide an extensive list of 1-dimensional SSA-compatible PMIs that are non-chordal yet, but nevertheless can be realized by graph models. It is noteworthy that the majority of these KC-PMIs are realized in \cite{He:2024xzq} by graph models with tree topology, although they are (necessarily) not simple. Since any non-simple tree graph model can be transformed into a simple one by relabeling the boundary vertices and increasing the number of parties to some $\N' \geq \N$ \cite{Hernandez-Cuenca:2022pst}, it follows from our results that after this lift to $\N'$ parties, the resulting KC-PMIs are chordal.

One may then hope that studying the conditions under which a non-chordal KC-PMI lifts to a chordal one might shed new light on the conjectures of \cite{Hernandez-Cuenca:2022pst} regarding the building blocks of the holographic entropy cone.

\acknowledgments

V.H. would like to thank Sergio Hern\'andez-Cuenca and Remi Seddigh for useful conversations.  V.H. has been supported in part by the U.S. Department of Energy grant DE-SC0009999 and by funds from the University of California.   M.R. acknowledges support from UK Research and Innovation (UKRI) under the UK government’s Horizon Europe guarantee (EP/Y00468X/1).
V.H. acknowledges the hospitality of the Kavli Institute for Theoretical Physics (KITP) during early stages of this work. 

There is no underlying data associated with this work.

For the purpose of open access, the authors have applied a Creative Commons Attribution (CC BY) licence to any Author Accepted Manuscript version arising from this submission.

\bibliography{hypergraph_rep}
\bibliographystyle{utphys}

\end{document}